\documentclass[letterpaper,12pt,journal,onecolumn,comsoc, draftclsnofoot]{IEEEtran}
\IEEEoverridecommandlockouts

 \usepackage{cite}
\usepackage{graphicx,color,epsfig,rotating}
\usepackage{amsfonts,amsmath,amssymb}
\usepackage{algorithm}
\usepackage{caption}
\usepackage{subcaption}
 
\usepackage{algpseudocode}
\usepackage{placeins}
\usepackage{psfrag, graphicx}

\usepackage{booktabs}
\usepackage{longtable}
\usepackage{tabu}
\usepackage{multirow}
\usepackage{multicol}
\usepackage{accents}
\usepackage{mathrsfs}
\usepackage{amsbsy}

\newtheorem{definition}{Definition}
\newtheorem{theorem}{Theorem}
\newtheorem{lemma}{Lemma}
\newtheorem{corollary}{Corollary}

\newtheorem{proof}{Proof}
\newtheorem{remark}{Remark}

\def\beq{\begin{equation}}
\def\eeq{\end{equation}}
\def\beqa{\begin{eqnarray}}
\def\eeqa{\end{eqnarray}}
\def\beqan{\begin{eqnarray*}}
	\def\eeqan{\end{eqnarray*}}

\def\EE{{\mathbb{E}}}
\def\PP{{\mathbb{P}}}

\def\Wsf{ {  W}}  
\def\Xsf{ { X}}

\def\Usf{ {\sf U}}
\def\Vsf{ {\sf V}}

\def\Xc{\mathcal{X}}
\def\Yc{\mathcal{Y}}
\def\Uc{\mathcal{U}}
\def\Dc{\mathcal{D}}

\def\dbf{\mathbf{d}}

\def\R{\rho}
\def\Rscr{\boldsymbol\varrho}

\def\RX{R_{ach}} 
\def\REX{\bar{R}_{ach}} 

\def\Rstar{R^*(M)}                       
\def\regop{\mathfrak R^{*}}         
\def\Rlb{R^{LB}(M)}                     

\def\RGWstar{{R}_{GW\!\text{-}MR}^*(M)}                       
\def\RGW{{R}_{GW\!\text{-}MR}(M)}                       

\def\RGWstarR{R_{MR}^*(M,\Rscr)}            
\def\RGWRlb{R_{MR}^{LB}(M,\Rscr)}          
\def\Rach{R_{MR}}
\def\regopGW{\mathfrak R^{*}_{MR}}

\def\RstarE{{\bar R}^*(M)}                       
\def\regopE{ \bar{\mathfrak R}^{*}}        
\def\RlbE{\bar{R}^{LB}(M)}                     

\def\RGWstarE{\bar{R}_{GW\!\text{-}MR}^*(M)}                      
 \def\RGWE{\bar{R}_{GW\!\text{-}MR}(M)}                       

\def\RGWstarRE{\bar{R}_{MR}^*(M,\Rscr)}   
\def\RGWRlbE{\bar{R}_{MR}^{LB}(M,\Rscr)}   
\def\RachE{\bar{R}_{MR}}       							
\def\regopGWE{\bar{\mathfrak R}_{MR}^{*}}         

\def\GWregion{\mathfrak S_{GW}} 
\def\GWregionS{{\mathfrak{S}}_{SGW}}

\begin{document}

\bibliographystyle{IEEEtran}


\title{Rate-Memory Trade-Off for Caching and Delivery of Correlated Sources}

\author{Parisa Hassanzadeh, Antonia M. Tulino, Jaime Llorca, Elza Erkip
	\thanks{This work has been supported in part by NSF under grant \#1619129, and in part by NYU WIRELESS.}
\thanks{P. Hassanzadeh  and  E. Erkip are with the ECE Department of New York University, Brooklyn, NY. Email: \{ph990, elza\}@nyu.edu}
\thanks{J. Llorca  and A. Tulino are with Nokia Bell Labs, Holmdel, NJ, USA. Email:  \{jaime.llorca, a.tulino\}@nokia-bell-labs.com}
\thanks{A. Tulino is with the DIETI, University of Naples Federico II, Italy. Email:  antoniamaria.tulino@unina.it}
}

\maketitle

\begin{abstract}
	This paper studies the fundamental limits of content delivery in a cache-aided broadcast network for correlated content generated by a discrete memoryless source with arbitrary joint distribution. Each receiver is equipped with a cache of equal capacity, and the requested files are delivered over a shared error-free broadcast link. A class of achievable correlation-aware schemes based on a two-step source coding approach is proposed. Library files are first compressed, and then cached and delivered using a combination of correlation-unaware multiple-request cache-aided coded multicast schemes. The first step uses Gray-Wyner source coding to represent the library via private descriptions and descriptions that are common to more than one file. The second step then becomes a multiple-request caching problem, where the demand structure is dictated by the configuration of the compressed library, and it is interesting in its own right. 
	The performance of the proposed two-step scheme is evaluated by comparing its achievable rate with a lower bound on the optimal peak and average rate-memory trade-offs in a two-file multiple-receiver network, and in a three-file two-receiver network.
	Specifically, in a network with two files and two receivers, the achievable rate matches the lower bound for a significant memory regime and it is within half of the conditional entropy of files for all other memory values. In the three-file two-receiver network, the two-step strategy achieves the lower bound for large cache capacities, and it is within half of the joint entropy of two of the sources conditioned on the third one for all other cache sizes.

	
\end{abstract}

\section{Introduction}~\label{sec:Introduction}
Prefetching portions of popular content into cache memories distributed throughout the network in order to enable coded multicast transmissions useful for multiple receivers is regarded as a highly effective technique for reducing traffic load 
in wireless networks. 
The fundamental rate-memory trade-off in cache-aided broadcast networks with independent content has been studied in numerous works, including~\cite{maddah14fundamental,maddah14decentralized,ji15order,ji2015caching}, and more recently, with {\em correlated} content in  \cite{timo2016rate,gastpar2017caching,ISTC2016,JSAC2018,ITW2016,hassanzadeh2017rate,yang2017centralized,asilomar2017}.  
Exploiting content correlations becomes particularly critical as we move from static content distribution towards real-time delivery of rapidly changing personalized data (as in news updates, social networks, immersive video, augmented reality, etc.) in which exact content reuse is almost non-existent \cite{Antony17aoi}.
Such correlations are especially
relevant among files of the same category, such as
episodes of a TV show or same-sport recordings, which, even
if personalized, may share common backgrounds and scene
objects, and between multiple versions of dynamic data (e.g., news or social media updates).  

The works in \cite{timo2016rate} and  \cite{gastpar2017caching} consider a single-receiver single-cache multiple-file network with lossy reconstructions, and characterize the trade-offs between rate, cache capacity, and reconstruction distortions. The analysis in \cite{timo2016rate} also considers two receivers and one cache, in which again only local caching gains can be explored. The work in \cite{gastpar2017caching}  models the caching problem in a way that resembles the Gray-Wyner network \cite{gray1974source}.
Our prior works in \cite{ISTC2016,JSAC2018,ITW2016} focus on lossless reconstruction in a setting with an arbitrary number of files and receivers that allows exploring untapped global caching gains under correlated sources. A correlation-aware scheme is proposed in \cite{ISTC2016} and \cite{JSAC2018}, in which content is cached according to both the popularity of files and their correlation with the rest of the library. Cached information is then used as references for the compression of requested files during the delivery phase. 
Alternatively, our work in \cite{ITW2016}  addresses the  content dependency by first compressing the correlated library. A subset of the files, most representative of the library, are selected as references, referred to as {\em I-files}, and the remaining files are inter-compressed with respect to the selected files and referred to as {\em P-files}. This results in a compressed library where each file is made up of an I-file and a P-file, leading to a multiple-request caching problem. Differently from previous multiple-request schemes \cite{ji14groupcast,ji15efficient,ji2015caching,daniel2017optimization,sengupta2017improved}, the demand in  \cite{ITW2016} has a specific structure dictated by the configuration of the resulting compressed library in terms of I-files and P-files.

The  first information-theoretic characterization of the rate-memory trade-off in a cache-aided broadcast network with correlated content was studied in \cite{hassanzadeh2017rate} for the setting of two files and two receivers, each equipped with a cache. This paper introduces an achievable
two-step scheme that exploits content correlations by first 
jointly compressing the library files using the Gray-Wyner network \cite{gray1974source}, and then treating
the compressed content as independent files. It is shown in  \cite{hassanzadeh2017rate}  that
this strategy is optimal for a large memory regime, while
the gap to optimality is quantified for other memory values.

Building on the idea introduced in \cite{ITW2016} and \cite{hassanzadeh2017rate}, concurrent work in \cite{yang2017centralized} 
proposes a caching scheme for a network with arbitrary number of files and receivers, where the library has a specific correlation structure, i.e., each file is composed of multiple independent subfiles that are common among a fix set of files in the library. This, as in \cite{ITW2016}, leads to a  multiple-request caching problem where the demand has a particular configuration dictated by  the specific library structure. 




All previously cited works provide achievable caching schemes without analytically quantifying the gap to optimality, except for the special case of two files and two receivers in \cite{hassanzadeh2017rate}. In this paper, by focusing on lossless reconstructions, we extend the information-theoretic analysis of the broadcast caching network done in \cite{hassanzadeh2017rate} to arbitrary number of receivers, each equipped with its own cache. 
Differently from \cite{ISTC2016,JSAC2018,ITW2016} and \cite{yang2017centralized}, we  characterize the peak and average rate-memory region for files generated by a discrete memoryless source with arbitrary joint distribution, and we propose a class of optimal or near-optimal two-step schemes, for which
preliminary results were presented in   \cite{asilomar2017}. 
Our main contributions are summarized as follows:
\begin{itemize}
	\item We formulate the problem of efficient delivery of multiple correlated files over a broadcast caching network with arbitrary number of receivers via information-theoretic tools.
	
	\item We propose a class of correlation-aware two-step schemes, in which the files are first encoded based on the Gray-Wyner network \cite{gray1974source}, and in the second step, they are cached and delivered through a correlation-unaware multiple-request cache-aided coded multicast scheme. While most of the literature focuses on equal-length files, our multiple-request scheme in the second step is general enough to account for files compressed at different rates. 
	
	\item We discuss the optimality of the proposed two-step scheme in a two-file and $K$-receiver network by characterizing an upper bound on the peak and average rate-memory trade-offs for this class of schemes, and comparing it with a lower bound on the optimal rate-memory trade-offs derived in \cite{ ITLowerBound} using a cut-set argument on the corresponding cache-demand augmented graph \cite{llorca2013network}. 
	We identify a set of operating points in the achievable Gray-Wyner region \cite{gray1974source} for which the proposed two-step scheme is optimal over a range of cache capacities, and  approximates the optimal rate in the two-file network within half of the conditional entropy for all cache sizes. 
	
	\item We then extend the analysis to the three-file scenario since it captures the essence of the multiple-file case  without involving the exponential complexity of the multiple-file Gray-Wyner network. We show that for two receivers the proposed scheme is optimal for high memory sizes and its gap to optimality is less than half of the joint entropy of two of the sources conditioned on the third source for other memory. 
	
	\item As a means to designing an achievable scheme for the second step of the proposed Gray-Wyner-based methodology, we also present a novel near-optimal multiple-request caching scheme for a network with two receivers and three independent files, where each receiver requests two of the files. The proposed scheme uses {\em coded} cache placement to achieve optimality for cache capacities up to half of the library size.

\end{itemize}

The paper is organized as follows. Sec. \ref{sec:Problem Formulation} presents the information-theoretic problem formulation. In Sec. \ref{sec: GW-CACM scheme}, we introduce a class of two-step schemes based on the Gray-Wyner network.  The multiple-request caching problem, arising from the Gray-Wyner network, 
is discussed in Sec.~\ref{sec: general MR}, and a multiple-request scheme for the two-file network is proposed and analyzed in detail in Sec.~\ref{sec: two files GW-CACM}. Sec.~\ref{sec:Order Optimality} combines the multiple-request scheme proposed in Sec.~\ref{sec: two files GW-CACM} with the Gray-Wyner encoding step, and analyzes the optimality of the overall two-step scheme with respect to a lower bound on the rate-memory trade-off in a two-file network.  
Extensions to a three-file network are analyzed in Sec.~\ref{sec:three files}. After numerically analyzing the rate-memory trade-off using an illustrative example, the paper is concluded in Sec.~\ref{sec:Conclusions}. 


\section{Network Model and Problem Formulation}\label{sec:Problem Formulation}
We consider a broadcast caching network composed of one sender (e.g., base station)
with access to a library of $N$ uniformly popular files generated by an $N$-component discrete memoryless source (N-DMS). The N-DMS model $\Big(\Xc_1 \times\dots\times\Xc_N, \, p(x_1,\dots,x_N)\Big)$ consists of $N$ finite alphabets $\Xc_1,\dots,\Xc_N$ and a joint pmf  $p(x_1,\dots,x_N)$ over $\Xc_1 \times \dots\times\Xc_N$. 
For a block length $F$, library file $j \in \{1,\dots,N\}$ is represented by a sequence $\Xsf_j^F = (\Xsf_{j1},\dots,\Xsf_{jF})$, where $\Xsf_j^F\in \Xc^F_j$, and $(\Xsf_{1i},\dots,\Xsf_{Ni})$, $i\in\{1,\dots,F\}$ is generated i.i.d. according distribution $ p(x_1,\dots,x_N)$. 
The sender communicates with $K$ receivers, $\{r_1,\, r_2,\, \dots,\, r_K\}$, over a shared error-free broadcast link. Each receiver is equipped with a cache of size $MF$ bits,
where $M$ denotes the (normalized) cache capacity.

We assume that the system operates in two phases:  
a {\em caching phase} and a {\em delivery phase}. During the {caching phase}, which takes place at off-peak hours when network resources are abundant, receiver caches are filled with functions of the library files, such that during the {delivery phase}, when receiver demands are revealed and resources are limited, the sender broadcasts the shortest possible codeword that allows each receiver to losslessy recover its requested file.
We refer to the overall scheme, in which functions of the content are prefetched into receiver local caches, and are later used to reduce the delivery rate by transmitting coded versions of the requested files, as a {\em cache-aided coded multicast scheme} (CACM).
A CACM scheme consists of the following components:
\begin{itemize}
	\item {\textbf{Cache Encoder:}} During the caching phase, the cache encoder designs the cache content of receiver $r_k$ using a mapping 
	$$f^{\mathfrak C}_{r_k}:    \Xc_1^F\times\dots\times\Xc_N^F \rightarrow [1: 2^{MF}).$$
	The cache configuration of  receiver $r_k$  is denoted  by $Z_{r_k} = f^{\mathfrak  C}_{r_k}\Big(\{X_1^F,\dots, X_N^F\}\Big)$.

	\item{\textbf{Multicast Encoder:}} During the delivery phase, each receiver requests a file from the library. The demand realization, denoted by $\dbf = (d_{r_1},d_{r_2},\dots,d_{r_K}) \in \Dc \equiv \{1,\dots,N\}^K$, is revealed to the sender, where $d_{r_k} \in\{1,\dots,N\}$ denotes the index of the file requested by receiver $r_k$. The sender  uses a fixed-to-variable mapping 
	$$f^{\mathfrak M}:{\mathcal D} \times [1: 2^{MF})^K  \times \Xc_1^F\times\dots\times \Xc_N^F  \rightarrow \Yc^\star \;\footnote{We use $\star$ to indicate variable length.}$$ to generate and transmit a multicast codeword $Y_{\dbf} = f^{\mathfrak  M}\Big(\dbf,\{Z_{r_1},\dots,Z_{r_K}\}, \{X_1^F,\dots,X_N^F\}\Big)$
	over the shared link.

	\item{\textbf{Multicast Decoders:}} Each receiver $r_k$ uses a mapping 
	$$g^{\mathfrak  M}_{r_k} : \Dc \times \Yc^\star \times [1: 2^{MF}) \rightarrow \Xc_{d_{r_k}}^F$$
	to recover its requested file, $X_{d_{r_k}}^F$, using the received multicast codeword and its cache content as $\widehat{X}_{d_{r_k}}^F = g^{\mathfrak  M}_{r_k} (\dbf, Y_{\dbf},Z_{{r_k}})$.
	
\end{itemize}

The worst-case probability of error of a CACM scheme is given by
\begin{align} \label{perr}
	& P_e^{(F)} = \max_{\dbf\in\mathcal D} \; \PP\left( \bigcup\limits_{r_k\in\{1,\dots,K\}} \Big\{  \widehat{X}_{{d_{r_k}}}^F  \neq X_{{d_{r_k}}}^F \Big\}  \right).
\end{align}

In this paper, we consider two performance criteria: 
\begin{itemize}
	\item[$i$)] The peak multicast rate, $R^{(F)}$, which corresponds to the worst-case demand,
	\begin{equation} \label{peak-rate}
		R^{(F)} =    \max_{\dbf \in \mathcal D} \; \frac{\EE[L(Y_{\dbf})]}{F}, 
	\end{equation}
	where $L(Y)$ denotes the length (in bits) of the multicast codeword $Y$, and the expectation is over the source distribution.
	
	\item[$ii$)] The average multicast rate, $\bar R^{(F)}$, over all possible demands
	\begin{equation} \label{average-rate}
		\bar R^{(F)} =   \frac{\EE[L(Y_{\dbf})]}{F},
	\end{equation}
	where the expectation is over demands and source distribution.

\end{itemize}

\begin{definition} \label{def:achievable-peak}
	A peak rate-memory pair $(R,M)$ is {\em achievable} if there exists a sequence of CACM schemes for cache capacity $M$ and
	increasing file size $F$, such that $\lim_{F \rightarrow \infty} P_e^{(F)} = 0 \notag$, and $\limsup_{F \rightarrow \infty}
	R^{(F)} \leq  R$. 
\end{definition}

\begin{definition} \label{def:infimum-peak}
	The peak rate-memory region $\regop$ is the closure of the set of achievable peak rate-memory pairs $(R,M)$ and the optimal peak rate-memory function, $\Rstar$, is
	$$\Rstar= \inf \{R:  (R,M) \in  \regop\}.$$
\end{definition}

\begin{definition} \label{def:achievable-average} 
	An average rate-memory pair $(R,M)$ is {\em achievable} if there exists a sequence of CACM schemes for cache capacity $M$ and increasing file size $F$, such that $\lim_{F \rightarrow \infty} P_e^{(F)} = 0 \notag$, and $\limsup_{F \rightarrow \infty} \bar R^{(F)} \leq  R$.
\end{definition}

\begin{definition} \label{def:infimum-average}
	The average rate-memory region $\regopE$ is the closure of the set of achievable average rate-memory pairs $(R,M)$ and the optimal average rate-memory function, $\RstarE$, is
	$$\RstarE= \inf \{R:  (R,M) \in  \regopE\}. $$
\end{definition}

\section{Gray-Wyner-Network-Based Two-Step Achievable  Schemes}\label{sec: GW-CACM scheme}

In this section, we propose a class of CACM schemes, based on a two-step lossless source coding setup, as depicted in Fig.~\ref{fig:caching schemes}.  
The first step involves lossless Gray-Wyner source coding \cite{gray1974source}, and the second step is a Multiple-Request CACM scheme. 

\begin{figure}[H] \centering
	\includegraphics[width=0.8\linewidth]{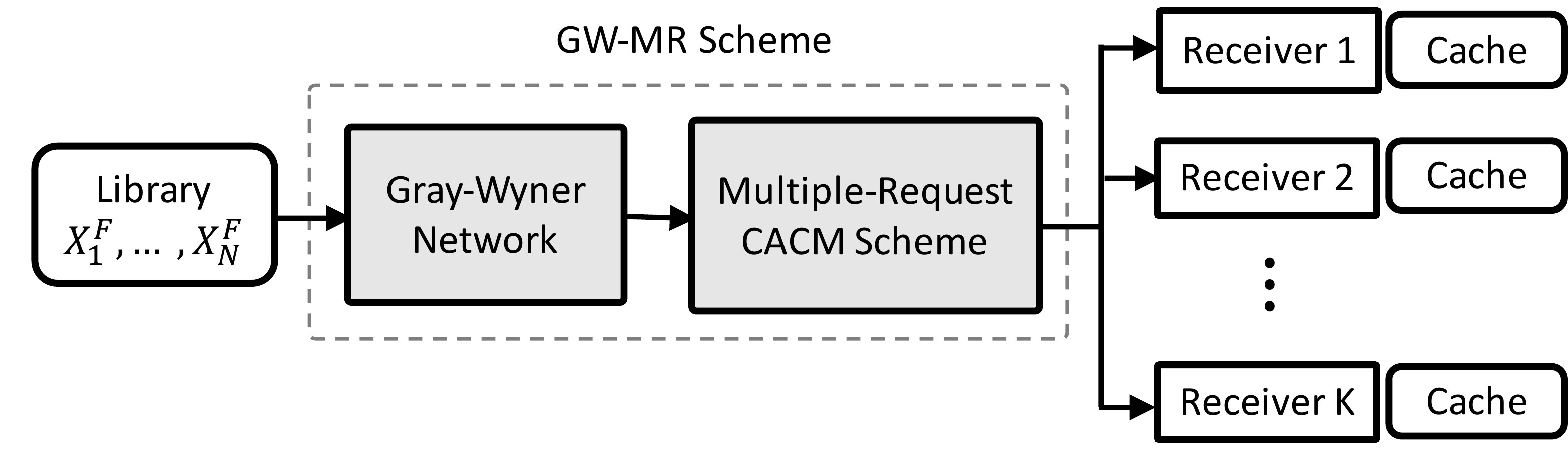}
	\caption{
		Two-step correlation-aware scheme, composed of a Gray-Wyner source coding step followed by a multiple-request CACM step.} 
	\label{fig:caching schemes}
\end{figure}

The proposed two-step scheme exploits the correlation among the library content by first compressing the library using the Gray-Wyner network, described in detail in Sec.~\ref{subsec: GW network}, and depicted for two and three files in Fig.~\ref{fig:GW}. 
The  Gray-Wyner network represents the library using $N$ private descriptions and $\binom{N}{\ell}$ descriptions that are common to $\ell$ files for $\ell\in\{2,\dots,N\}$, thereby transforming the caching problem with correlated content and receivers requesting only one file, into a caching problem with a larger number of files 
where receivers require multiple descriptions. 
Note that this multiple-request caching problem has a specific class of demands. We assume that the multiple-request CACM scheme in the second step is agnostic to the correlation among the content generated by the Gray-Wyner network, i.e., the second step is correlation-unaware. The two steps are jointly designed to optimize the performance of the overall scheme, which is referred to as the {\em Gray-Wyner Multiple-Request CACM}  (GW-MR) scheme.  Before formally describing the GW-MR scheme, we briefly review the Gray-Wyner Network.

\subsection{Gray-Wyner Network}\label{subsec: GW network}
The Gray-Wyner network was first introduced for two files in \cite{gray1974source}, in which a 2-DMS $(\Xsf_1, \Xsf_2 )$ is represented by three descriptions $ \{\Wsf_0, \Wsf_1, \Wsf_2\}$, where $\Wsf_0 \in [1:2^{F\R_0})$ is referred to as the common description, and $\Wsf_1 \in [1:2^{F\R_1})$ and $\Wsf_2 \in [1:2^{F\R_2})$ are the corresponding private descriptions as depicted in Fig.~\ref{fig:GW}(a). The descriptions are such that file $X_1^F$ can be losslessly recovered from descriptions $(\Wsf_0, \Wsf_1)$, and file $X_2^F$ can be losslessly recovered from descriptions $(\Wsf_0, \Wsf_2)$, both asymptotically, as block length $F \rightarrow \infty$. In \cite{gray1974source}, Gray and Wyner fully characterized the rate region for lossless reconstruction of both files, which is restated in the following Theorem.
\begin{theorem}[Gray-Wyner Rate Region]\label{thm:GW region}
	The optimal rate region for the two-file Gray-Wyner network, $\GWregion$, is 
	\begin{align}\label{eq:gwregion2}
		\GWregion = cl\bigg\{ \bigcup \Big\{  (\R_0,\, \R_1,\, \R_2) : \;  
		\R_0   \geq   I(\Xsf_1, \Xsf_2;\Usf),  \; 
		\R_1   \geq  H(\Xsf_1|\Usf),    \;
		\R_2  \geq   H(\Xsf_2|\Usf) \Big\} \bigg\},
	\end{align}
	where $cl\{S\}$ denotes the closure of set $S$, and the union is over all choices of $\Usf$  for some $p(u|x_1,x_2)$ with  $|\Uc| \leq |\Xc_1||\Xc_2| + 2$.
\end{theorem}

\begin{figure}
	\begin{subfigure}{0.5\linewidth}\centering 
		\includegraphics[width=\linewidth]{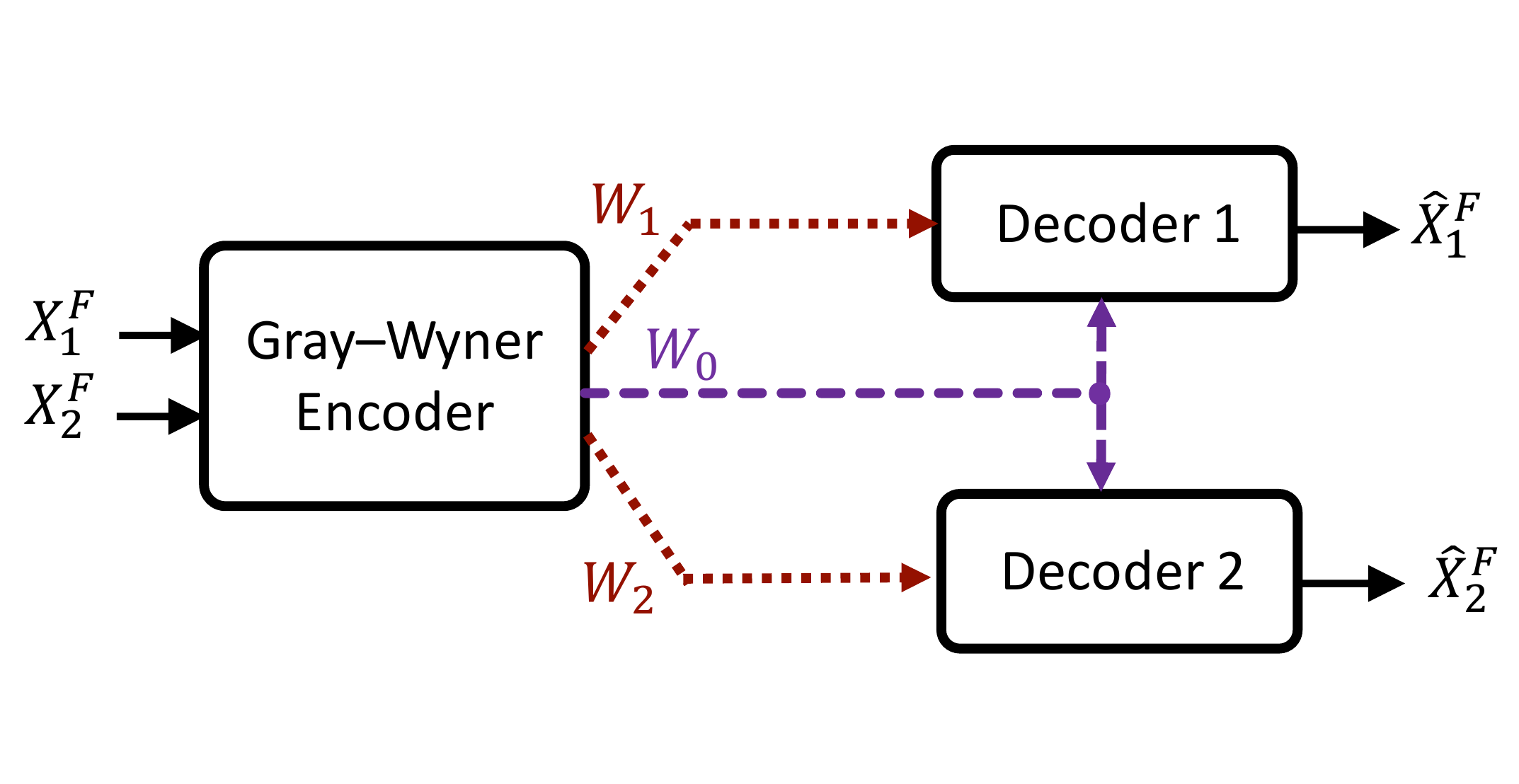} 
		\subcaption{}
	\end{subfigure}
	\begin{subfigure}{0.5\linewidth}\centering 
		\includegraphics[width=\linewidth]{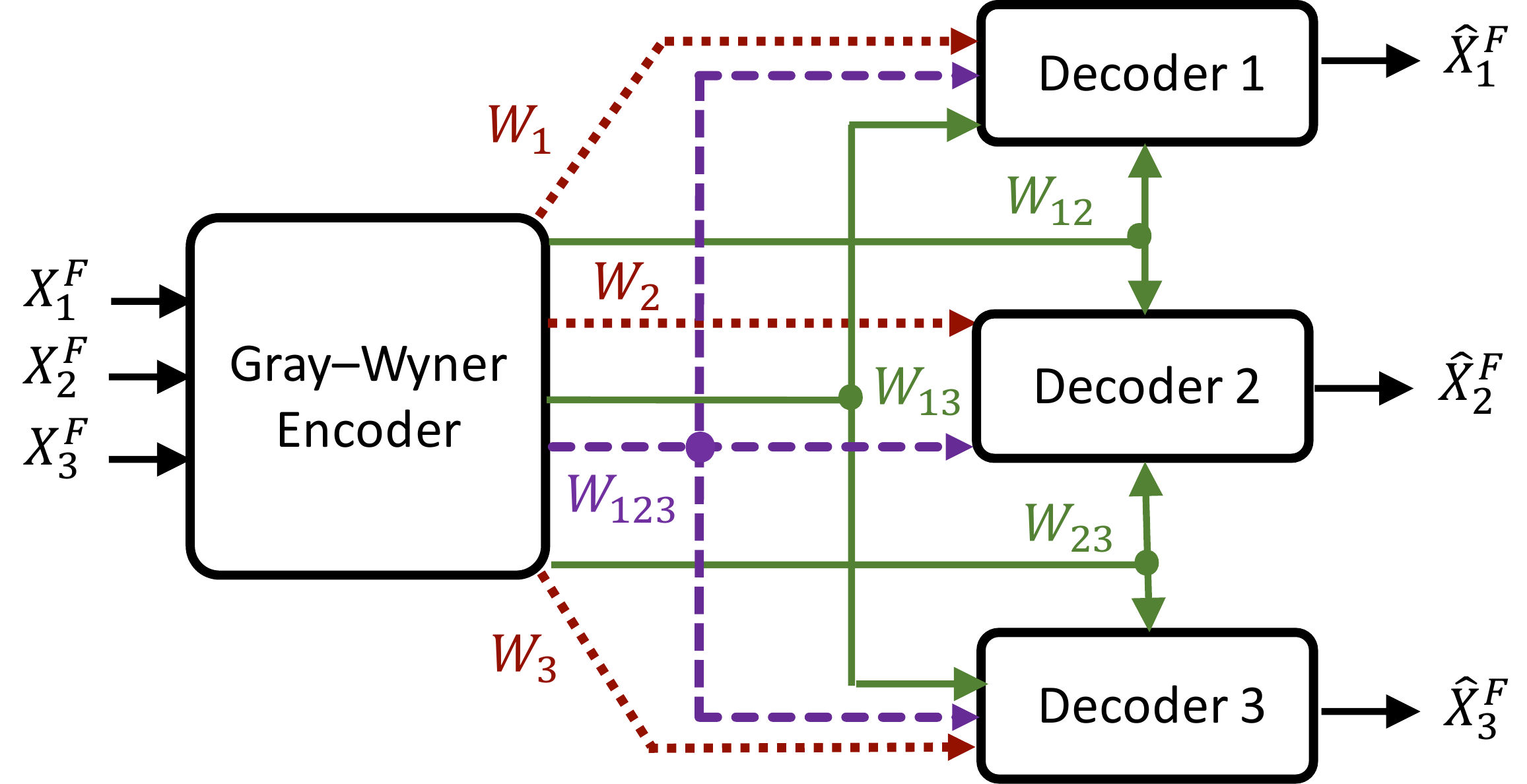}
		\subcaption{}
	\end{subfigure}
	\caption{Gray-Wyner network for (a) two files, and (b) three files.}
	\label{fig:GW}
\end{figure}

The Gray-Wyner network can be extended to $N$ files such that the Gray-Wyner encoder observes a $N$-DMS $(\Xsf_1,\dots,\Xsf_N)$ and communicates $\Xsf_i$ to decoder $i\in\{1,\, \dots\, N\}$. The encoder is connected to the decoders through $2^{N-1}$ error-free links, such that there is a link connecting the encoder to any subset $s \subseteq\{1,\dots,N\}$ of the decoders. In particular, there is one common
link connecting the encoder to all $N$ decoders, there are $\binom{N}{\ell}$ links common to any $\ell\in\{2,\dots, N-1\}$ of the decoders, and finally, there are $N$ private links connecting the encoder to each
decoder. For any nonempty set $ \mathcal A \subseteq\{1,\dots,N\}$, description $\Wsf_{\mathcal A} \in [1:2^{F\R_{\mathcal A}})$ is communicated to all decoders $i \in { \mathcal A}$. The Gray-Wyner rate region, $\GWregion$, is represented by the set of all rate-tuples  $\Rscr$ for the $2^N-1$ descriptions, 
for which file $X_i^F$, $i\in\{1,\dots,N\}$ can be losslessly reconstructed from the descriptions 
$$\Big\{   \Wsf_s: s  \subseteq \{1,\dots,N\},\; i \in s \Big\},$$
asymptotically, as $F \rightarrow \infty$. In general, $N$ files are encoded into $2^N-1$ descriptions, such that: i) $N$ of the descriptions contain information exclusive to only one file, and ii) the remaining descriptions comprise information common to more than one file.
In this paper we will study in detail the three-file Gray-Wyner network depicted in Fig.~\ref{fig:GW}. The encoded descriptions are such that\footnote{With an abuse of notation, the subscripts of $\Wsf$ and $\R$ denote sets.}
\begin{itemize}
	\item	$\Wsf_{123}$ $\in [1:2^{F\R_{123}})$,  
	\item	$\Wsf_{12} \in [1:2^{F\R_{12}})$, $\Wsf_{13} \in [1:2^{F\R_{13}})$,  $\Wsf_{23} \in [1:2^{F\R_{23}})$,  
	\item	$\Wsf_1 \in [1:2^{F\R_1})$, $\Wsf_2 \in [1:2^{F\R_2})$, $\Wsf_3 \in [1:2^{F\R_3})$, 
\end{itemize}
and the Gray-Wyner rate region is represented by the set of all rate-tuples
\begin{align}
	\Rscr =\Big(\R_{123},\R_{12},\R_{13},\R_{23},\R_{1},\R_{2},\R_{3}\Big),\label{eq:ratetuple}
\end{align}
for which file $X_i^F$, $i\in\{1,2,3\}$,  can be losslessly reconstructed from the descriptions $\Big\{\Wsf_{123},$ $\Wsf_{ij},$ $\Wsf_{ik},$ $\Wsf_i \Big\}$ with $j,k\in\{1,2,3\}\setminus \{i\}$  asymptotically, as $F \rightarrow \infty$. 

While the rate region of the $N$-file Gray-Wyner network has been studied in a number of papers \cite{tandon2011multi,liu2010common,viswanatha2012subset}, 
the optimal characterization for general sources is not known.

\subsection{Gray-Wyner Multiple-Request  CACM Scheme}\label{subsec: general GW-CACM} 
The proposed class of two-step schemes consists of:
\begin{itemize}
	\item {\bf Gray-Wyner Encoder}: Given the library $\{X_1^{F},$ $\dots,$ $X_N^{F}\}$,
	the Gray-Wyner encoder at the sender computes descriptions $\{W_s :\, s \in \mathcal S\}$, where $\mathcal S$ is the set of all nonempty subsets of $\{1,\dots,N\}$, using a mapping
	$${f}^{GW}: \Xc_1^{F}\times\dots\times\Xc_N^{F} \rightarrow \prod\limits_{s\in\mathcal S} \Big[1:2^{F\R_{\mathcal S}}\Big)  ,$$
	for $\Rscr\in\GWregion$.
	
	\item {\bf Multiple-Request Cache Encoder}: Given the compressed descriptions, the correlation-unaware cache
	encoder at the sender computes the cache content at receiver $r_k\in \{1,\dots,K\}$, as 
	$Z_{r_k} =   f_{r_k}^{{\mathfrak C_{MR}}}   \Big (\{W_s:\, s\in\mathcal S \}\Big)$.
	
	\item {\bf Multiple-Request Multicast Encoder}: For any demand realization $\dbf\in\mathcal D$ revealed to the sender, the correlation-unaware  
	multicast encoder generates and transmits the multicast codeword
	$Y_\dbf = f^{{\mathfrak M_{MR}}} \Big(\dbf,\{Z_{r_1},\dots,Z_{r_K}\}, \{W_s:\, s\in\mathcal S \}\Big)$.
	
	\item {\bf Multiple-Request Multicast Decoder}: Receiver $r_k$ decodes the descriptions corresponding to its requested file as
	$$\Big\{ \widehat{\Wsf}_{s}: s\in \mathcal S_{d_{r_k}}  \Big\} = g^{{\mathfrak M_{MR}}}_{r_k}\Big(\dbf, Y_\dbf, Z_{r_K}\Big),$$
	where $  \mathcal S_{d_{r_k}}  \triangleq \Big\{ s\in\mathcal S:\, d_{r_k} \in \, s  \Big\}$.
	\item {\bf Gray-Wyner Decoder}: Receiver $r_k$ decodes its requested file using the descriptions recovered by the multicast decoder, as $\widehat{X}_{d_{r_k}}^{F} = g^{GW}_{r_k}\Big(  \Big\{ \widehat{\Wsf}_{s}: s  \in \mathcal S_{d_{r_k}} \Big\}\Big)$, via the Gray-Wyner decoder 
	$$g^{GW}_{r_k}:   \prod\limits_{ s\in \mathcal S_{d_{r_k}} } [1:2^{F\R_{s} })  \rightarrow \Xc_{d_{r_k}}^F.$$
\end{itemize}

The Gray-Wyner encoder and decoder correspond to the first step, namely the encoder and decoder of the Gray-Wyner network, and the multiple-request cache encoder, multiple-request multicast encoder, and multiple-request multicast  decoder comprise the second multiple-request CACM (MR) step of the proposed two-step scheme. 

Note that the performance of the class of schemes described above depends on the operating point of the  Gray-Wyner network $\Rscr\in\GWregion$. For a given $\Rscr$, the performance of the overall two-step scheme is dictated by the  peak and average multicast rates of the MR scheme, 
which
similar to \eqref{peak-rate} and \eqref{average-rate}, 
are defined as
\begin{align} 
	& R_{MR}^{(F)}(\Rscr) =      \max_{\dbf\in\mathcal D} \;  \frac{\EE[L(Y_{\dbf})]}{F}, \label{eq:ach peak MR}\\
	& {\bar R}_{MR}^{(F)}(\Rscr) =   \frac{\EE[L(Y_{\dbf})]}{F}, \label{average-rate-GW}
\end{align}
respectively.

Furthermore, the worst-case probability of error of the class of two-step schemes depends on the probability of error of the Gray-Wyner source coding step and the probability of error of the multiple-request CACM step. Since $\Rscr \in \GWregion$,  and $ \Big\{ {\Wsf}_s: s \in \mathcal S_{d_{r_k}}  \Big\} $ is a Gray-Wyner description of $X_{d_{r_k}}^F$ with $d_{r_k}\in\{1,\dots,N \}$, it is guaranteed that Gray-Wyner decoding is asymptotically lossless with $F$.  Hence, the probability of error of a two-step scheme is approximately upper bounded by the probability of error of the MR scheme, given by
\begin{align}
	P_{e,MR}^{(F)}  = \max_{\dbf\in\mathcal D} \;  
	\PP\left( \bigcup\limits_{r_k \in\{1,\dots,K\}}  \Big\{ \widehat{\Wsf}_s \neq  \Wsf_s ,\, \forall s \in \mathcal S_{d_{r_k}}  \Big\}  \right) . \notag
\end{align}


\begin{definition} \label{def:achievable-peak GW}
	For a given rate-tuple $\Rscr\in\GWregion$, an MR peak rate-memory pair $(R,M)$ is {\em achievable} if there exists a sequence of MR schemes with rate ${R}_{MR}^{(F)}(\Rscr)$, for cache capacity $M$ and
	increasing file size $F$, such that $\lim_{F \rightarrow \infty}  P_{e,MR}^{(F)}  = 0 \notag$, and 
	$\lim\sup_{F\rightarrow \infty} R^{(F)}_{MR}(\Rscr) \leq R$. 
\end{definition}

\begin{definition} \label{def:infimum-peak GW}
	For a given rate-tuple $\Rscr \in\GWregion$, the MR peak rate-memory region $\regopGW(\Rscr)$ is the closure of the set of achievable MR peak rate-memory pairs $(R,M)$, and the MR peak rate-memory function $\RGWstarR$ is defined as
	\begin{align}
		&\RGWstarR= \inf \{R:  (R,M) \in  \regopGW(\Rscr)\}. \notag
	\end{align}
\end{definition}

In the class of two-step schemes, we refer to the scheme operating at the rate-tuple $\Rscr\in\GWregion$ 
that minimizes the MR peak rate-memory function 
as the {\em GW-MR scheme}. The peak rate-memory pair achieved by this scheme is the GW-MR peak rate-memory function defined below.
\begin{definition} \label{def:GW-CACM peak} 
	The GW-MR peak rate-memory function $\RGWstar$  is given by
	\begin{align}
		& \RGWstar = \inf\{   \RGWstarR:\,   \Rscr\in\GWregion\}. \notag
	\end{align}
\end{definition}

In line with Definitions~\ref{def:achievable-peak GW}-\ref{def:GW-CACM peak}, for the average rate criterion we have the following.
\begin{definition} \label{def:achievable-average GW} 
	For a given rate-tuple $\Rscr\in\GWregion$, an MR average rate-memory pair $(R,M)$ is {\em achievable}   if there exists a sequence of MR schemes, with rate ${\bar R}_{MR}^{(F)}(\Rscr)$, for cache capacity $M$ and increasing file size $F$, such that $\lim_{F \rightarrow \infty} P_{e,MR}^{(F)}  = 0 \notag$, and  $\lim\sup_{F\rightarrow \infty} {\bar R}_{MR}^{(F)}(\Rscr) \leq R$.
\end{definition}

\begin{definition} \label{def:infimum-average GW}
	For a given rate-tuple $\Rscr\in\GWregion$, the MR average rate-memory region $\regopGWE(\Rscr)$ is the closure of the set of achievable MR average rate-memory pairs $(R,M)$, and the MR  average rate-memory function $\RGWstarRE$ is defined as
	\begin{align}
		&\RGWstarRE= \inf \{R:  (R,M) \in  \regopGWE(\Rscr)\}. \notag
	\end{align}
\end{definition}

\begin{definition} \label{def:GW-CACM average} 
	The GW-MR average rate-memory function $\RGWstarE$ is given by
	\begin{align}
		& \RGWstarE = \inf\{   \RGWstarRE:\,   \Rscr\in\GWregion\}. \notag 
	\end{align}
\end{definition}


In the remainder of this paper, we first present the MR scheme in the second step for a general setting (Sec. \ref{sec: general MR}), and then describe and analyze the performance of the overall GW-MR scheme for the case of two files and $K$ receivers (Secs. \ref{sec: two files GW-CACM} and \ref{sec:Order Optimality}), and for the case of three files and two receivers (Sec.~\ref{sec:three files}).


\section{Multiple-Request CACM Scheme} \label{sec: general MR}
In this section, we focus on the second step of the GW-MR scheme depicted in Fig. 1, namely the  
MR scheme when the Gray-Wyner network operates at $\Rscr \in \GWregion$. 

Recall that the Gray-Wyner network converts the $N$-file library into $2^{N}-1$ descriptions, each required for the lossless reconstruction of a set of the files in the original library. The MR scheme arranges the descriptions generated by the Gray-Wyner encoder into $N$ groups, $L_1,\dots,L_N$, referred to as {\em sublibraries}. Sublibrary $L_\ell = \Big\{W_{s}:   s\subseteq \{1,\dots,N\}, |s| =\ell  \Big\}$ contains the descriptions that are common to exactly $\ell$ files. 
We refer to sublibrary $L_{N}=\{W_{12\dots N} \}$ as the {\em common-to-all} sublibrary, and to $L_1=\{W_1,\dots,W_N \}$, which contains all the private descriptions, as the {\em private} sublibrary. 
The MR scheme accounts for populating the receiver caches with content from sublibraries $L_1,\dots L_N$ and 
serving the demand realizations placed in the original library, which translate into a new class of demands from the compressed sublibraries. 
More specifically, each receiver demand corresponds to a demand for a set of descriptions from each of the sublibraries (hence the term multiple-request), 
such that the original library file $ X_{d_{r_k}}$ requested by receiver $r_k$ maps to descriptions $\Big\{ W_s\in L_\ell: d_{r_k}\in s   \Big\}$ from sublibrary $L_\ell$, $\ell \in\{1, \ldots N\}$.  
Even though receivers request single files from the original library independently and according to a uniform demand distribution, 
the Gray-Wyner encoding process leads to a non-uniform multiple-request demand for descriptions (files) from the compressed sublibraries.

Our proposed MR scheme treats each sublibrary independently during both caching and delivery phases. 
Specifically, the descriptions from each sublibrary are cached and delivered as follows: i) the description in the common-to-all sublibrary $L_N$, which is required for the reconstruction of all files, and hence requested by all receivers, is cached according to the Least Frequently Used (LFU)\footnote{LFU is a local caching policy that, in the setting of this paper, leads to all receivers caching the same part of the file.} strategy and delivered through naive (uncoded) multicasting, ii) sublibrary $L_1$ is cached and delivered according to any single-request correlation-unaware CACM scheme (such as \cite{zhang2015differentsize,cheng2017optimal,li2017rate} for unequal-length descriptions, and \cite{maddah14fundamental,lim2016information,tian2016caching,yu2016exact,chen2014fundamental,wang2016caching,gomez2016} for equal-length descriptions), 
and  iii) sublibrary $L_\ell$, $\ell\in \{N-1,\dots,2\}$ is cached and delivered according to any  correlation-unaware CACM scheme in which each receiver requests $\binom{N-1}{\ell-1}$  descriptions. 
Schemes where each receiver requests more than one file have been analyzed in previous works, e.g.,  \cite{ji14groupcast,ji15efficient,ji2015caching,daniel2017optimization,sengupta2017improved}. 
However, in addition to being limited to settings with equal-length files, they have been designed for arbitrary demand combinations and could therefore be suboptimal for the specific class of demands considered in our MR scheme.  One of the challenges addressed in the next sections is the design of 
near-optimal schemes for delivering the $\binom{N-1}{\ell-1}$ descriptions requested by each receiver from sublibrary $L_\ell$, $\ell\in \{N-1,\dots,1\}$. Since the ultimate goal is to characterize the performance of the overall two-step GW-MR scheme, and the optimal Gray-Wyner rate region is only known for two files \cite{gray1974source}, in the following sections, we first describe the proposed MR scheme and analyze the performance of the associated GW-MR scheme for the case of two files and $K$ receivers in Secs.~\ref{sec: two files GW-CACM} and \ref{sec:Order Optimality}, respectively. In addition, 
to illustrate how extensions to more files could be done, we focus on the setting with three files and two receivers in Sec.~\ref{sec:three files}.

It is worth noticing that the setting of our MR scheme, where  the descriptions generated by the Gray-Wyner network are modeled as independent subfiles common to a fix set of files in the original library, is a  generalization of  the problem considered in \cite{yang2017centralized}, where, differently from our setting, subfiles are assumed to have equal length.  Our results,  for two files and $K$ receivers, and for three files and two receivers, if specialized to equal-length descriptions (subfiles) are shown to yield optimal or near-optimal schemes for the problem formulation studied in \cite{yang2017centralized}. 



\section{Multiple-Request CACM Scheme for Two Files and $K$ Receivers}\label{sec: two files GW-CACM}
This section describes in more detail the MR scheme introduced in the previous section for a network with two files and $K$ receivers.
Let $\Rscr = (\R_0,\R_1,\R_2)\in \GWregion$ denote the operating point of the Gray-Wyner network, where $\R_0$ denotes the rate of the common description $W_{12}$, and $\R_1$ and $\R_2$ denote the rate of the private descriptions $W_1$ and $W_2$, respectively. As described in Sec.~\ref{sec: general MR}, the MR scheme for two files arranges the three descriptions generated by the Gray-Wyner network into a common-to-all (or simply common) sublibrary $L_2 = \{W_{12}\}$, and a private sublibrary $L_1=\{\Wsf_1,\Wsf_2\}$. Each receiver demand corresponds to requesting  two descriptions: one from the common sublibrary $L_2$, and one 
from the private sublibrary $L_1$. The specific caching and delivery strategies adopted for each sublibrary are provided in Sec.~\ref{subsec: scheme}.

In order to analyze the performance of the proposed MR scheme we also provide lower bounds on the MR peak and average rate-memory functions in Sec.~\ref{subsec: GW lower bound}, and compare the achievable rates of the proposed MR scheme with these lower bounds in Sec.~\ref{subsec: ML MR}.

\subsection{Lower Bounds on $\RGWstarR$ and $\RGWstarRE$}\label{subsec: GW lower bound}

\begin{theorem}\label{thm:LowerBoundGW}
	In the two-file  $K$-receiver network, for a given cache capacity $M$ and rate-tuple $\Rscr\in\GWregion$, a lower bound on $\RGWstarR$, the MR peak rate-memory function, is given by
	\begin{align}
		\RGWRlb = \inf \bigg \{ R: \quad
		&  R \,\geq\, \R_0 + \R_1 + \R_2 \,-\, 2M , \notag\\
		&  R \,\geq\,  \R_0 + \frac{1}{2}\Big(  \R_1 + \R_2 + \max \{\R_1 , \R_2 \} \Big)\,-\, M , \notag\\
		&  R \,\geq\, \frac{1}{2}\Big (\R_0 + \R_1 + \R_2 \,-\, M \Big)
		\bigg \}. \notag
	\end{align}
	A lower bound on $\RGWstarRE$, the MR average rate-memory function, is given by
	\begin{align}
		\RGWRlbE = \inf  \bigg \{ R: \quad 
		&  R \,\geq\, \R_0 + \Big ( 1-\frac{1}{2^K}\Big )\Big ( \R_1 + \R_2\Big) \,-\,2\Big ( 1-\frac{1}{2^K}\Big ) M , \notag\\
		&  R \geq \frac{3}{4}\R_0  +\frac{1}{2}( \R_1 + \R_2) +\frac{1}{4}\max \{\R_1 ,\R_2\}  - \frac{3}{4}M    , \notag\\
		&  R \,\geq\,  \R_0 + \frac{3}{4} ( \R_1 + \R_2    )\,-\, M , \notag\\
		&  R \,\geq\, \frac{1}{2}\Big( \R_0 + \R_1 + \R_2 \,-\, M \Big) \bigg \}. \notag
	\end{align}
\end{theorem}
\begin{proof} 
	The proof follows from the more general results presented in Theorem~\ref{thm:LowerBound} in Sec.~\ref{subsec: optimal lower bound}. By setting $\Xsf_1= (\Wsf_{12}, \Wsf_{1})$ and $\Xsf_2= (\Wsf_{12},\, \Wsf_2)$, the above results are readily obtained.
	
\end{proof}

\subsection{Proposed MR Scheme}\label{subsec: scheme}
As described in Sec.~\ref{sec: general MR}, the proposed MR scheme for two files treats the descriptions in $L_1$ and $L_2$ as independent content and operates as follows: $i)$ the cache capacity  $M$ is optimally divided among the two sublibraries, $ii)$ caching is done independently from each sublibrary, and $iii)$ the content requested from each sublibrary is delivered independently, i.e., with no further coding across them. 
While the common description $W_{12}$ in $L_2$ is cached according to the LFU strategy 
and delivered through uncoded multicasting,  the private descriptions in $L_1$ can be cached and delivered according to any correlation-unaware CACM scheme available in literature (e.g. \cite{maddah14fundamental,lim2016information,tian2016caching,yu2016exact,chen2014fundamental,wang2016caching,gomez2016}), properly generalized to a setting with unequal-length files. In the following, for a given generalized correlation-unaware CACM scheme adopted for sublibrary $L_1$, we describe how to optimally allocate the memory to each sublibrary, and we then  characterize the peak and average rate achieved by the corresponding MR scheme.

Let $\RX(M,\R_1,\R_2)$ and $\REX(M,\R_1,\R_2)$ denote the lower convex envelope of the peak and average rates achieved by the scheme adopted for sublibrary $L_1$, respectively.

Under the peak rate criterion, the 
cache allocation that minimizes the overall delivery rate  (i.e.,  the sum rate of each sublibrary) 
is as follows. Let 
\begin{align}
	M^* \triangleq \min\Big\{M:\, \Big|\frac{\partial_{-}}{\partial M}\RX(M,\R_1,\R_2) \Big| < 1\Big\}. \label{eq:m prime}
\end{align}  
denote the cache encoder threshold, where $\frac{\partial_{-}}{\partial M}$ denotes the left partial derivative with respect to $M$. For a given $\Rscr\in \GWregion$, 
the cache encoder allocates the memory to each sublibrary as follows: 
\begin{itemize}
	\item If $M\in \Big[0,\,  M^*\Big)$, the common description $\Wsf_{12}$ is not cached at either receiver, and descriptions $\{\Wsf_1, \Wsf_2\}$ from $L_1$ are cached according to the caching strategy of the adopted correlation-unaware CACM scheme.
	
	\item If $M\in \Big[ M^*,\,M^*+\R_0\Big)$, the first $F(M- M^*)$ bits of description $\Wsf_{12}$ are cached at both receivers (as per LFU caching), and descriptions $\{\Wsf_1, \Wsf_2\}$ are cached according to the scheme adopted for sublibrary $L_1$ over the remaining memory $M^*$. 
	
	\item If $M\in \Big[M^*+\R_0,\,  \R_0+\R_1+\R_2\Big]$, the common description $\Wsf_{12}$ is fully cached at both receivers, and descriptions $\{\Wsf_1, \Wsf_2\}$ are cached according to the adopted correlation-unaware CACM scheme over the remaining cache capacity $M-\R_0$.  
\end{itemize}


The optimality of the cache allocation described above is proved in Appendix~\ref{app:achievable rate}, and its graphical representation is depicted in Fig.~\ref{fig: waterfilling}. This representation can be understood as ``water-filling of two leaky buckets'', which is related to the well-known water-filling optimization problem \cite{Cover}. For small cache capacities up to $M^*$, all the memory is allocated to the private sublibrary. As the cache size increases, it is optimal to first store the common description until it is fully cached, and for cache capacities larger than $M^*$, the residual memory is allocated to storing the private descriptions. 
\begin{figure} [h!] 
	\includegraphics[width=1\linewidth]{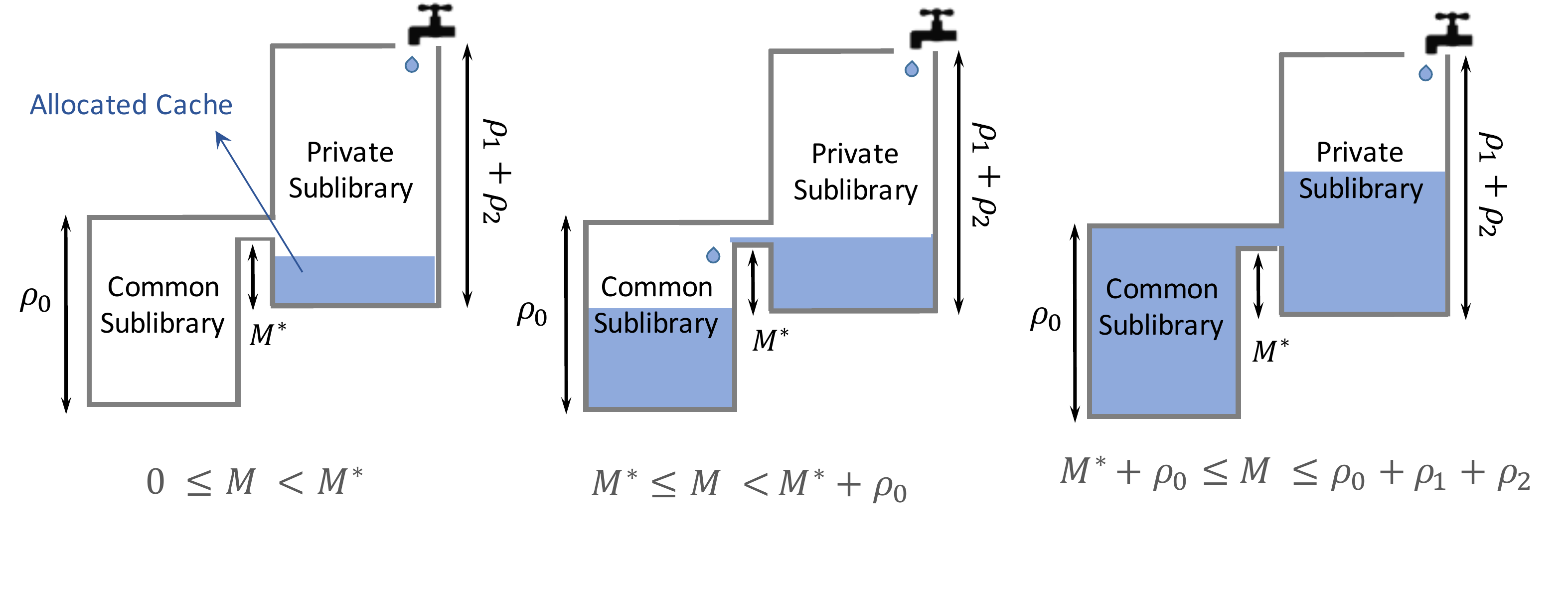}
	\caption{Optimal caching strategy in the MR scheme for two files with $M^*$ defined in \eqref{eq:m prime}.} 
	\label{fig: waterfilling}
\end{figure}

\begin{remark} 
	Differently from the single-cache setting analyzed in \cite{timo2016rate}, where it is always optimal to cache the common description, in our case, when the cache capacity is smaller than $M^*$, it is optimal to first cache the private descriptions. This is due to the fact that for small cache sizes 
	delivering the privare descriptions with coded multicast transmissions is more effective in reducing the overall reate compared to mutlicasting the commom description, and therefore it is preferable to prioratize caching from the private sublibrary. As the cache size increases, this difference in rate reduction diminishes to the extent that assigning memory to the common sublibrary is more effective in reducing the rate compared to the private sublibrary. 
	Therefore, it is preferable to fully store the common description prior to allocating additional memory for storing the private descriptions. 
\end{remark}

Under the average rate criterion, the optimal cache allocation among sublibraries $L_1$ and $L_2$ is  similar to the peak rate scenario described previously,   with respect to a cache encoder threshold
\begin{align}
	\bar{M}^*  \triangleq \min\Big\{M:\, \Big|\frac{\partial_{-}}{\partial M}\REX(M,\R_1,\R_2) \Big| < 1\Big\}, \label{eq:M bar star}
\end{align}
where $\REX(M,\R_1,\R_2)$ is the average delivery rate achieved  by the correlation-unaware CACM scheme adopted for sublibrary $L_1$. 


The following theorem provides the peak and average rates achieved by the proposed MR scheme for the optimal cache  allocation described above.
\begin{theorem}\label{thm:achievable rate peak}
	In the two-file $K$-receiver network, for a given cache capacity $M$, rate-tuple  $\Rscr$, and a given adopted correlation-unaware CACM scheme with peak rate $\RX(M,\R_1,\R_2)$ the peak rate achieved by the MR scheme, $\Rach(M,\Rscr)$, is given by
	\begin{equation}
		\Rach(M,\Rscr)  =
		\begin{cases}
			\R_0 + \RX(M,\R_1,\R_2) ,                   & \; M \in \Big[0,\,  M^*\Big)  \\
			\R_0 + \RX( M^*,\R_1,\R_2) + M^*-M,				       &\; M\in\Big[ M^*, \R_0+\, M^*\Big) \\
			\RX(M-\R_0,\R_1,\R_2),	    &\; M\in \Big[\R_0+ M^*,\, \R_0+\R_1+\R_2\Big] \label{eq:peak rate ach}
		\end{cases}
	\end{equation}
	where $M^*$ is defined in \eqref{eq:m prime}. Denoting by $\REX(M,\R_1,\R_2)$ the average rate achieved by an adopted correlation-unaware CACM scheme, the average delivery rate achieved by the corresponding MR scheme, $\RachE(M,\Rscr)$, is similar to that in \eqref{eq:peak rate ach} but with respect to $\bar{M}^* $ and $\REX(M,\R_1,\R_2)$.
\end{theorem}

\begin{proof}
	The proof is given in Appendix~\ref{app:achievable rate}.
\end{proof}

From Theorem~\ref{thm:achievable rate peak} 
it is observed that the performance of the proposed MR scheme depends on the correlation-unaware CACM scheme adopted for the private sublibrary.  In order to analyze the optimality of the proposed MR scheme in Sec.~\ref{subsec: ML MR}, for sublibrary $L_1$,  
we resort to a scheme that properly combines near optimal schemes available in literature, generalized to unequal file lengths, as described next.

\subsubsection{Peak Rate}
The scheme adopted for the private sublibrary $L_1$ is based on memory sharing among generalizations  of the correlation-unaware CACM schemes proposed in  \cite{tian2016caching} and  \cite{yu2016exact} to files with unequal lengths. Specifically, for small cache sizes, the adopted scheme uses the scheme in \cite{tian2016caching} (first introduced in \cite{chen2014fundamental} for a specific cache capacity), which prefetches coded content during the caching phase, 
while for large cache sizes, it uses the scheme in \cite{yu2016exact} with uncoded cache placement.
For ease of exposition, the specific details of the scheme adopted for the private sublibrary is given in Appendix \ref{app:CACM private peak}, where we also derive an upper bound on its achievable rate, $\RX(M,\R_1,\R_2)$, given in \eqref{eq:K user UB}. By combining $\RX(M,\R_1,\R_2)$ with Theorem~\ref{thm:achievable rate peak}, an upper bound on the peak rate achieved by the proposed MR scheme is given in the following theorem.
\begin{theorem}\label{thm:peak rate MR}
	In the two-file $K$-receiver network, for a given cache capacity $M$ and rate-tuple $\Rscr$, an upper bound on $\Rach(M,\Rscr)$, the peak  rate achieved with the proposed MR scheme, is given by
	\begin{equation}
		\Rach(M, \Rscr) \leq
		\begin{cases}
			\R_0 + \R_1+\R_2 -2M ,                   & \; M \in \Big[0,  \,\gamma_{K} \Big)  \\
			\R_0 + \R_1+\R_2- \gamma_{K}  -M,				       &\; M\in \Big[ \gamma_{K} , \, \lambda_{K} \Big)  \\
			\frac{1}{2}(\R_0+\R_1+\R_2-M)  ,	    &\; M\in \Big[\lambda_{K} , \, \R_0+\R_1+\R_2\Big] \label{eq: rate peak K}
		\end{cases} 	
	\end{equation}
	where
	\begin{align}
		& \gamma_K \triangleq  \frac{1}{K} \min\{\R_1,\R_2\},  \quad \lambda_K \triangleq \R_0+\R_1+\R_2 - 2 \gamma_K. \label{eq: gamma lambda K}
	\end{align}
\end{theorem}
\begin{proof}
	The proof follows from Theorem~\ref{thm:achievable rate peak} and from adopting the scheme described in Appendix~\ref{app:CACM private peak} for the private sublibrary. An upper bound on the peak rate achieved with this scheme, $\RX(M,\R_1,\R_2)$, is given in \eqref{eq:K user UB} and when replaced in \eqref{eq:m prime}, the cache encoder threshold becomes $M^* = \R_1+\R_2-\frac{2}{K}\min\{\R_1,\R_2\}$. Combining \eqref{eq:peak rate ach} with \eqref{eq:K user UB}, and using the fact that $\gamma_K\leq M^*\leq \lambda_K$, \eqref{eq: gamma lambda K} is obtained after algebraic manipulation.
\end{proof}
Specializing the upper bound on $\Rach(M,\Rscr) $ given in \eqref{eq: rate peak K} to  $K=2$ receivers results in a tight upper bound\footnote{By a  tight upper bound we mean that for $K=2$ equation \eqref{eq: rate peak K} holds with equality.} for this setting as per the following corollary.

\begin{corollary}\label{col:peak rate 2}
	In the two-file two-receiver network, for a given cache capacity $M$ and rate-tuple $\Rscr$, the proposed MR scheme achieves the following peak rate
	\begin{align}
		\Rach(M, &\Rscr) = \begin{cases}
			\R_0 + \R_1+\R_2  -2 M ,                   &  M \in \Big[0, \, \frac{1}{2}\min\{\R_1,\R_2\}\Big)  \\
			\R_0  + \frac{1}{2}\Big(\R_1+\R_2 +  \max\{\R_1,\R_2\}\Big ) -  M,				      &  M\in\Big[\frac{1}{2}\min\{\R_1,\R_2\}, \R_0+\max\{\R_1,\R_2\} \Big) \\
			\frac{1}{2}(\R_0  + \R_1 + \R_2 -M),	    &  M\in \Big[\R_0 + \max\{\R_1,\R_2\}  ,\, \R_0 + \R_1+\R_2\Big] . \notag
		\end{cases}
	\end{align}
\end{corollary}

\subsubsection{Average Rate} In this case, for the private sublibrary, we adopt a generalization of the close-to-optimal correlation-unaware CACM scheme in \cite{yu2016exact} to unequal-length files. The specific details of the adopted scheme are provided in Appendix \ref{app:CACM private avg}, where we also derive an upper bound on its achievable rate, $\REX(M,\R_1,\R_2)$, given in \eqref{eq:K user UB avg}. By combining $\REX(M,\R_1,\R_2)$ with the average counterpart of \eqref{eq:peak rate ach} given in Theorem~\ref{thm:achievable rate peak}, an upper bound on the average rate achieved with the proposed MR scheme is given in the following theorem.
\begin{theorem}\label{thm:avg rate MR}
	In the two-file $K$-receiver network, for a given cache capacity $M$ and rate-tuple $\Rscr$, an upper bound on $\RachE(M,\Rscr)$, the average  rate achieved with the proposed MR scheme, is given by 
	\begin{equation}
		\RachE(M, \Rscr) \leq
		\begin{cases}
			\R_0+\Big(1-\frac{1}{2^K}\Big)  (\R_1+\R_2)-\Big(\frac{3}{2}-\frac{2}{2^K}  \Big)M  ,        & \; \mu \in \Big[0, \; 2\, \gamma_K\Big)   \\
			\R_0+\Big(  1-\frac{1}{2^K}  \Big) ( \R_1+\R_2) -   \Big( 1-\frac{4}{2^K}  \Big)   \gamma_K-M  ,				     & \;   
			\mu \in \Big[2\,\gamma_K, \; \bar\gamma_K \Big)  \\
			\Big(  1-\frac{1}{2^K}  \Big)( \R_0+\R_1+\R_2-M )-  \Big(1-\frac{2}{2^K}  \Big) \gamma_K  ,				     & \;   
			\mu \in \Big[\bar\gamma_K, \;  \lambda_K  \Big)  \\
			\frac{1}{2} (\R_0+\R_1 + \R_2 - M) ,	    &\; \mu \in \Big[ \lambda_K    , \;\R_0+\R_1+\R_2\Big]   \label{eq: rate avg K}
		\end{cases} 	
	\end{equation}
	for $\gamma_K$ and $\lambda_K$ defined in Theorem~\ref{thm:peak rate MR}, and  	with
	\begin{align}
		\bar\gamma_K \triangleq \R_0+\frac{2}{K}\min\{\R_1,\R_2\} =  \R_0+2\gamma_K.	 \label{eq: gamma bar K}
	\end{align}
\end{theorem}	
\begin{proof} 
	The proof follows from Theorem~\ref{thm:achievable rate peak} and from adopting the scheme described in Appendix~\ref{app:CACM private avg} for the private sublibrary.  An upper bound on the average rate achieved with this scheme, $\REX(M,\R_1,\R_2)$, is given in \eqref{eq:K user UB avg} and when replaced in \eqref{eq:M bar star}, the cache encoder threshold becomes ${\bar M}^*= \frac{2}{K}\min\{\R_1,\R_2\} = 2\gamma_K$.  Eq \eqref{eq: gamma lambda K} is obtained by combining the average counterpart of \eqref{eq:peak rate ach} given in Theorem~\ref{thm:achievable rate peak} with \eqref{eq:K user UB avg}.
\end{proof}
Specializing the upper bound on $\RachE(M,\Rscr) $ given in \eqref{eq: rate avg K} to the setting with $K=2$ receivers leads to the following corollary.
\begin{corollary}\label{col:average rate 2}
	In the two-file two-receiver  network, for a given cache capacity $M$ and rate-tuple $\Rscr$, the proposed MR scheme achieves the following average  rate
	\begin{align}
		&\RachE(M, \Rscr) =\notag\\
		&\quad\begin{cases}
			\R_0+\frac{3}{4}(\R_1+\R_2) - M,				      &M \in \Big[0,\,\R_0 + \min\{\R_1,\R_2\} \Big ) \\
			\frac{3}{4} \R_0+\frac{1}{2}(\R_1+\R_2) +\frac{1}{4} \max\{\R_1,\R_2\}- \frac{3}{4} M,  			      &M \in \Big[\R_0 + \min\{\R_1,\R_2\} ,\, \R_0 + \max\{\R_1,\R_2\}  \Big )  \\
			\frac{1}{2}(\R_0+\R_1+\R_2-M),				     &M \in \Big[\R_0 + \max\{\R_1,\R_2\},\, \R_0+\R_1+\R_2  \Big ] 
		\end{cases}
	\end{align}
\end{corollary}

\subsection{Optimality of the Proposed MR Scheme}\label{subsec: ML MR}
In this section, we compare the performance of the proposed MR scheme charaterized in Sec.~\ref{subsec: scheme} 
with the lower bounds given in Sec.~\ref{subsec: GW lower bound}. 
The following theorems provide the memory regions for which the proposed MR scheme is optimal or near-optimal.
\subsubsection{Peak Rate}
\begin{theorem}\label{thm:ach and GW lower bound peak}  
	In the two-file $K$-receiver network and for a given rate-tuple $\Rscr$, when 
	\begin{align} 
		M\in \Big[ 0,\; \frac{1}{K}\min\{\R_1,\,\R_2\} \Big]\bigcup \Big[ \R_0+\R_1+\R_2-  \frac{2}{K}\min\{\R_1,\R_2\} , \;  \R_0+\R_1+\R_2 \Big], \notag
	\end{align}
	the proposed MR scheme is optimal under the peak rate criterion among all multiple-request schemes, i.e., $\Rach(M,\Rscr)=\RGWstarR$. For all other cache capacities,   
	$$\Rach(M, \Rscr)- \RGWstarR \leq \Big(\frac{1}{2}- \frac{1}{K}\Big) \min\{\R_1,\R_2\}.$$
\end{theorem}
\begin{proof} 
	The proof is given in Appendix \ref{app:ach and GW lower bound peak}.
\end{proof}

\begin{corollary}\label{col:peak and lower bound 2}
	In the two-file two-receiver network, for any cache capacity $M$ and rate-tuple $\Rscr$, the proposed MR scheme is optimal under the peak rate criterion among all multiple-request schemes, i.e.,  $\Rach(M,\Rscr)=   \RGWstarR.$
\end{corollary}
\begin{proof}
	The optimality of the scheme with respect to the MR peak rate-memory function follows from setting $K=2$ in Theorem~\ref{thm:ach and GW lower bound peak}, for which the achievable rate meets the lower bound for the entire region of the memory.
\end{proof}

\subsubsection{Average Rate} 

\begin{theorem}\label{thm:ach and GW lower bound avg}  
	In the two-file $K$-receiver network and for a given rate-tuple $\Rscr$, when
	\begin{align} 
		M\in  \Big[ \R_0+\R_1+\R_2-  \frac{2}{K}\min\{\R_1,\R_2\} , \;  \R_0+\R_1+\R_2 \Big], \notag
	\end{align}
	the proposed MR scheme is optimal under the average rate criterion  among all multiple-request schemes, i.e,  $\RachE(M,\Rscr)=\RGWstarRE.$ For all other cache capacities,  
	$$\RachE(M, \Rscr)- \RGWstarRE \leq \Big(\frac{1}{4}- \frac{1}{2^K}\Big) (\R_1+\R_2).$$
\end{theorem}
\begin{proof} 
	The proof is given in Appendix \ref{app:ach and GW lower bound avg}.
\end{proof}



\begin{corollary}\label{col:average and lower bound 2} 
	In the two-file two-receiver network, for any cache capacity $M$ and rate-tuple $\Rscr$, the proposed MR scheme is optimal under the average rate criterion among all multiple-request schemes, i.e.,  $\RachE(M,\Rscr)=   \RGWstarRE.$
\end{corollary}

\begin{proof}
	The optimality of the scheme with respect to the MR average rate-memory function follows from setting $K=2$ in Theorem~\ref{thm:ach and GW lower bound avg},  for which the achievable rate meets the lower bound for the entire region of the memory.

\end{proof}

\begin{remark}
	Corollaries~\ref{col:peak and lower bound 2} and \ref{col:average and lower bound 2}  imply that in the two-file two-receiver network, caching and delivering content independently across the sublibraries, as described in Sec.~\ref{subsec: scheme}, is sufficient to achieve optimality among all MR schemes, rendering coding across the multicast codewords pertaining to each sublibrary unnecessary.
\end{remark}

\section{Optimality of the GW-MR Scheme for Two Files}\label{sec:Order Optimality}
In this section, we first provide lower bounds on the optimal peak and average rate-memory functions, $\Rstar$ and $\RstarE$, given in Definitions \ref{def:infimum-peak} and \ref{def:infimum-average}, respectively, for the two-file $K$-receiver network 
described in Sec.~\ref{sec:Problem Formulation}.  Then, by using the results of the MR scheme for two files (Sec.~\ref{sec: two files GW-CACM}), we evaluate the performance of the proposed GW-MR scheme (Sec.~\ref{sec: GW-CACM scheme}), by comparing the presented lower bounds with the peak and average rates achieved by the proposed GW-MR, 
\begin{align}
	&\RGW = \inf\{ \Rach(M, \Rscr) : \Rscr\in\GWregion\}, \label{eq:ach peak GWMR}\\
	& \RGWE = \inf\{ \RachE(M, \Rscr) : \Rscr\in\GWregion\}, \label{eq:ach avg GWMR}
\end{align}
respectively,  where $\Rach(M, \Rscr)$ and $\RachE(M, \Rscr)$ are the peak and average rates achieved by the MR scheme (Sec.~\ref{subsec: scheme}).


\subsection{Lower bounds on $\Rstar$ and $\RstarE$}\label{subsec: optimal lower bound}
\begin{theorem}\label{thm:LowerBound}
	In the two-file $K$-receiver network with library distribution $p(x_1,x_2)$, for a given cache capacity $M$, a lower bound on $\Rstar$, the optimal peak rate-memory function, is given by
	\begin{align}
		\Rlb  =  \inf \bigg \{ R: \;\;\; 
		& R \geq H(\Xsf_1, \Xsf_2)\,-\,2M,\notag\\
		& R \geq \frac{1}{2}\Big( H(\Xsf_1, \Xsf_2) \,+\,   \max\Big\{H(\Xsf_1), H(\Xsf_2) \Big\}\Big) \,-\, M , \notag\\
		& R \geq \frac{1}{2}\Big(H(\Xsf_1, \Xsf_2)\,-\,M\Big)   \bigg \}. \notag
	\end{align}
	A lower bound on $\RstarE$, the optimal average rate-memory function, is given by
	\begin{align} 
		\RlbE  =  \inf \bigg \{ R: \;\;\;
		& R \geq \Big ( 1-\frac{2}{2^K}\Big )  H(\Xsf_1, \Xsf_2)+\frac{1}{2^K} \Big(H(\Xsf_1)+H(\Xsf_2) \Big) \,-\, 2\Big ( 1-\frac{1}{2^K}\Big ) M ,\notag\\
		& R \geq \frac{1}{2}H(\Xsf_1, \Xsf_2) \,+\,  \frac{1}{4}\Big( H(\Xsf_1)+H(\Xsf_2) \Big) \,-\, M  , \notag\\
		& R \geq \frac{1}{2}H(\Xsf_1, \Xsf_2) + \frac{1}{4}\max\Big\{H(\Xsf_1) ,H(\Xsf_2)  \Big\}  - \frac{3}{4}M , \notag\\ 
		& R \geq \frac{1}{2}\Big(H(\Xsf_1, \Xsf_2)-M\Big)  \bigg \}. \notag
	\end{align}
	
\end{theorem}
\begin{proof}  
	The proof is given in Appendix~\ref{app:LowerBound}.
	
\end{proof}

\begin{remark} 
	When particularized to i.i.d. sources, the lower bounds in Theorem \ref{thm:LowerBound} match the corresponding best known bounds derived in \cite{yu2017characterizing}.  
\end{remark}

\subsection{Optimality of the Proposed GW-MR Scheme}\label{subsec: optimality}

The following theorems characterizes the performance of the proposed GW-MR scheme described in Sec.~\ref{sec: GW-CACM scheme} for different regions of $M$, and delineates the rate-memory region for which the scheme is optimal or near-optimal.  
\begin{theorem}\label{thm:optimality peak} 
	In the two-file $K$-receiver network, let 
	\begin{align}
		& M_K \triangleq  \max\limits_{\Xsf_1 - \Usf -\Xsf_2} \frac{1}{K}  \min \Big\{H(\Xsf_1|\Usf),H(\Xsf_2|\Usf) \Big\}, \label{eq:Mk}
	\end{align} 
	for $\Usf$ with $ |\mathcal U|\leq |\mathcal X_1|.|\mathcal X_2| +2$.
	When $M\in \Big[ 0, \, {M}_K \Big] \bigcup \Big[ H(\Xsf_1, \Xsf_2)- 2  {M}_K , \,  H(\Xsf_1, \Xsf_2) \Big] $, the proposed GW-MR scheme is optimal under the peak rate criterion, i.e, $\RGW= \Rstar $. For all other cache capacities, we have
	\begin{align}
		&\RGW- \Rstar \leq \frac{1}{2} \min \Big\{H(\Xsf_1| \Xsf_2),H(\Xsf_2| \Xsf_1)\Big\} -  {M}_K .\label{eq:gap peak}
	\end{align}
\end{theorem}
\begin{proof}
	The proof is given in Appendix \ref{app:optimality peak}.
\end{proof}
We note that for the two-file network, using an MR scheme in the second step of GW-MR that treats the common and private sublibraries independently (Sec.~\ref{sec: general MR}), achieves optimality among all (correlation-aware) CACM schemes for small and large cache sizes.



\begin{theorem}\label{thm:optimality avg} 
	Let  $\Usf^*$ denote the auxiliary random variable that achieves $M_K$ defined in \eqref{eq:Mk}, and let  
	\begin{align}
		& \Delta_K \triangleq   \frac{1}{2^K}   \Big(H(\Xsf_1|\Usf^*)+H(\Xsf_2|\Usf^*) \Big). \label{eq:M2 bar}
	\end{align}
	Then, in the two-file $K$-receiver network, 
	when $M\in \Big[ H(\Xsf_1, \Xsf_2)- 2 {{M}_K }, \,  H(\Xsf_1, \Xsf_2) \Big] $, the proposed GW-MR scheme is optimal under the average rate criterion, i.e, $\RGWE= \RstarE$. For all other cache capacities, we have
	\begin{align}
		& \RGWE- \RstarE  \leq   \frac{1}{4} \Big(H(\Xsf_1| \Xsf_2)+H(\Xsf_2| \Xsf_1)\Big)- {   \Delta_K  } . \label{eq:gap avg}
	\end{align}
\end{theorem}

\begin{proof}
	The proof is given in Appendix \ref{app:optimality avg}. 
\end{proof}

\begin{remark}\label{rmk: not wyner common} 
	From Theorems \ref{thm:optimality peak} and \ref{thm:optimality avg} it follows that a desirable operating point for the Gray-Wyner network is  the rate-tuple   $\Rscr^*=\arg \max_{\Rscr \in\GWregion} \min \{\R_1,\R_2 \}$,  with the constraint  that $\R_0+\R_1+\R_2 = H(X_1, X_2)$. In fact, from the proofs of Theorems \ref{thm:optimality peak} and \ref{thm:optimality avg} it is observed that a necessary condition for the GW-MR scheme to achieve the lower bound on the optimal rate-memory function for small and large cache capacities, is that the Gray-Wyner network operate at a rate-tuple  $\Rscr\in\GWregion$ such that $\R_0+\R_1+\R_2 = H(X_1, X_2)$. To satisfy  this condition, it is sufficient to choose a  $\Rscr\in\GWregion$ such that $\R_0 = I(X_1,X_2;\Usf)$, $\R_1 = H(X_1|\Usf)$ and $\R_2 = H(X_2|\Usf)$ as in \eqref{eq:gwregion2} with any $\Usf$ of the  form $\Xsf_1 - \Usf - \Xsf_2$. Among all such $\Usf$,  
	the one that achieves $M_K$ given in \eqref{eq:Mk} maximizes the region over which the gap to optimality is zero.
	Note that the operating point  $\Rscr^*\in\GWregion$ is related to Wyner's common information. In fact,  
	in Wyner's common information the goal is to minimize $\R_0$ subject to $\R_0+\R_1+\R_2 = H(X_1, X_2)$,
	while in our case the goal is to maximize  $\min\{\R_1,\R_2\}$ subject to $\R_0+\R_1+\R_2 = H(X_1, X_2)$. 

\end{remark}


In the next corollary, we particularize  Theorems \ref{thm:optimality peak} and \ref{thm:optimality avg} for a specific 2-DMS.
\begin{corollary}\label{cor:completoverlap}
	In the two-file $K$-receiver network with $\Xsf_1 = (\Xsf_1',\Vsf)$ and $\Xsf_2 =(\Xsf'_2,\Vsf)$ such that $\Xsf'_1$ and $\Xsf'_2$ are conditionally independent given $\Vsf$, the proposed GW-MR scheme is optimal under the peak rate criterion when $M\in \Big[ 0, \,\frac{1}{K}\Gamma\Big] \cup \Big[ H(\Xsf_1, \Xsf_2)-   \frac{2}{K}\Gamma, \,  H(\Xsf_1, \Xsf_2) \Big] $,  
	and it is optimal under the average rate criterion when $M\in \Big[ H(\Xsf_1, \Xsf_2)-   \frac{2}{K}\Gamma, \,  H(\Xsf_1, \Xsf_2)\Big]$, where
	$$ {\Gamma} \triangleq \min\Big\{ H(\Xsf_1|\Xsf_2),H(\Xsf_2|\Xsf_1)\Big\}.$$
	When $K=2$, i.e, in the two-file two-receiver network, the proposed GW-MR scheme is optimal for any $M\in \Big[0,\, H(\Xsf_1,\Xsf_2)\Big]$, i.e., $\RGW = \Rstar$ and $\RGWE = \RstarE$.
	
\end{corollary}
\begin{proof}
	By taking $\Usf= \Vsf$, we have
	\begin{align}
		& \R_0 =I(\Xsf_1,\Xsf_2;\Usf) = H(\Vsf) = I(\Xsf_1;\Xsf_2),\notag\\
		& \R_1=H(\Xsf_1|\Usf)= H(\Xsf'_1|\Vsf) = H(\Xsf_1|\Xsf_2),\notag\\
		& \R_2=H(\Xsf_2|\Usf)= H(\Xsf'_2|\Vsf) = H(\Xsf_2|\Xsf_1),\notag
	\end{align}
	with point $\Rscr = \Big(I(\Xsf_1;\Xsf_2),H(\Xsf_1|\Xsf_2), H(\Xsf_2|\Xsf_1)\Big)$ belonging to the Gray-Wyner rate region. The region of memory over which the GW-MR scheme is optimal in the $K$-receiver network can be readily obtained  from Theorems~\ref{thm:optimality peak} and \ref{thm:optimality avg}. From specializing these theorems to a setting with $K=2$ receivers, the gaps to optimality in \eqref{eq:gap peak} and \eqref{eq:gap avg} vanish; hence, $\RGW= \Rstar$ and $\RGWE = \RstarE$ for any cache capacity $M$.  
\end{proof}

\begin{remark}
	The 2-DMS  considered in Corollary~\ref{cor:completoverlap} leads to a set of descriptions generated by the  Gray-Wyner network operating at $\Rscr = \Big(I(\Xsf_1;\Xsf_2),H(\Xsf_1|\Xsf_2), H(\Xsf_2|\Xsf_1)\Big)$, that is 
	equivalent to the library considered in \cite{yang2017centralized}, where each file is composed of two independent subfiles, one of which is common among the two files. Hence, from Corollary~\ref{cor:completoverlap} it follows that our proposed GW-MR scheme provides an optimal solution to the two-file two-receiver setting considered in \cite{yang2017centralized}.
\end{remark}

\begin{remark}
	From  Corollaries~\ref{col:peak and lower bound 2} and \ref{col:average and lower bound 2} it follows that 
	for any given  $\Rscr\in\GWregion$ the proposed  MR scheme for two files and two receivers, as described in  Sec.~\ref{sec: general MR}, is optimal in the sense that 
	it achieves both MR peak and average rate-memory functions.  However,  as stated in Theorems~\ref{thm:optimality peak} and \ref{thm:optimality avg},  
	the overall GW-MR scheme does not meet the lower bounds  on $\Rstar$ and on $\RstarE$  for intermediate values of $M$, 
	leaving the optimality of the two-step approach in this region of memory unresolved.
	An interesting future direction would be to investigate {\em joint} compression and caching strategies or a different two-step approach for this setting.
\end{remark}

\section{Three Files and Two Receivers}\label{sec:three files}  
As explained in Sec.~\ref{subsec: GW network}, due to the exponential complexity of Gray-Wyner source coding with the number of files, the overall characterization of the GW-MR scheme with large number of files is exceedingly difficult. Therefore,  in the following we focus on the three-file scenario, as it captures the essence of caching in broadcast networks with multiple files.  
{Specifically, looking at the peak rate criterion  we characterize the GW-MR peak rate-memory function $\RGWstar$, and we provide a lower bound on the optimal peak rate-memory function $\Rstar$.  To this end, for a given $\Rscr\in\GWregion$ corresponding to the  three-file Gray-Wyner network described in Sec.~\ref{subsec: GW network} and depicted in Fig.~\ref{fig:GW}(b), we  present a detailed description of our proposed MR scheme for three files, which is used in the second step of the GW-MR scheme. We then upper bound the peak rate achieved by the proposed MR scheme, $\Rach(M,\Rscr)$ defined in \eqref{eq:ach peak MR}, and use it to evaluate the peak rate achieved by the GW-MR scheme, $\RGW$, defined in \eqref{eq:ach peak GWMR} as
	\begin{align}
		&\RGW = \inf\{ \Rach(M, \Rscr) : \Rscr\in\GWregion\},\label{eq:ach peak GWMR three}
	\end{align}	
	and finally, we compute its gap to optimality using the lower bound on $\Rstar$.
}


\subsection{Proposed MR Scheme for Three Files}\label{sec: achievable three}
This section describes the MR scheme for three files used in the second step of the GW-MR scheme, when the Gray-Wyner network in the first step is restricted to operate at a symmetric rate region,  denoted by $\GWregionS$, and defined as 
\begin{align}
	\GWregionS \triangleq \Big\{\Rscr\in\GWregion:\, \R_{12}=\R_{13}=\R_{23} = \R',\;  \R_{1}=\R_{2}=\R_{3} = \R  \Big\}. \label{eq:symmetric GW}
\end{align}
In the following, for notational simplicity, we use $\R_0$ instead of $\R_{123}$ to denote the rate of the common description $\Wsf_{123}$. 
As described in Sec.~\ref{sec: general MR}, the MR scheme for three files arranges the seven descriptions generated by the Gray-Wyner network into a common-to-all sublibrary $L_3=\{\Wsf_{123} \}$, a {\em common-to-two} sublibrary $L_2 = \{\Wsf_{12}, \Wsf_{13}, \Wsf_{23} \}$, and a private sublibrary $L_1=\{\Wsf_1 ,\Wsf_2 ,\Wsf_3\}$, and treats the sublibraries independently during the caching and delivery phases. Each receiver  demand corresponds to requesting multiple descriptions: one description from $L_3$, two descriptions from $L_2$, and one from $L_1$. Recall that even though receivers request files from the original library independently and according to a uniform demand distribution, the structure of the corresponding demand in the MR scheme is dictated by the collective of the requested files, resulting in a non-uniform multiple-request demand that is not independent across the receivers.


As described in Sec.~\ref{sec: general MR}, the proposed MR scheme for three files treats the descriptions in $L_1$, $L_2$ and $L_3$ as independent content, and can adopt any pair of appropriate correlation-unaware CACM schemes for sublibraries $L_1$ and $L_2$. Here, we select the schemes as follows: 
$i)$ description $\Wsf_{123}$ in $L_3$ is cached according to the LFU strategy and delivered through uncoded  transmissions, $ii)$ for the descriptions in $L_2$, we adopt the new two-request CACM scheme proposed in Sec.~\ref{sec: new scheme}, and finally $iii)$ the private descriptions in sublibrary $L_1$ are cached and delivered according to the scheme proposed in \cite{yu2016exact}. Similar to  Sec.~\ref{subsec: scheme}, we first present the optimal allocation of the memory to each sublibrary, and  characterize the peak  rate achieved by the corresponding MR scheme.

Under the peak rate criterion, the optimal cache allocation which minimizes the overall delivery rate is as follows. For a given symmetric $\Rscr\in\GWregionS$, 
the cache encoder allocates the memory to each sublibrary such that: 	
\begin{itemize}
	\item If $M\in \Big[0,\, \frac{3}{2}\R' \Big)$, the descriptions in $L_1$ and $L_3$ are not cached at either receiver, and only the descriptions in $L_2$ are cached according to the caching strategy of the two-request CACM scheme described in Sec.~\ref{sec: new scheme}.
	
	\item If $M\in \Big[ \frac{3}{2}\R'  ,\, \R_0+\frac{3}{2}(\R'+\R) \Big)$, the  receivers fill a portion equal to $\frac{3}{2}\R'$ from their cache with the descriptions in $L_2$ according to the caching strategy of the two-request CACM scheme. The remainder of the cache, $M-\frac{3}{2}\R'$, is first allocated to caching identical bits of $L_3=\{\Wsf_{123}\}$ at both receivers as per LFU caching, and the excess of the  capacity, if any, is used for storing the descriptions in $L_1$ according to the scheme in \cite{yu2016exact}.

	\item If $M\in \Big[ \R_0+\frac{3}{2}(\R'+\R)  ,\, \R_0+3\R'+\frac{3}{2}\R\Big)$, a portion equal to $M-\R_0-\frac{3}{2}\R$ of each receiver's cache is filled with the descriptions in $L_2$ according to the two-request CACM scheme, the common description $\Wsf_{123}$ is fully cached at both receivers, and $\frac{3}{2}\R$ of the capacity is allocated to storing the descriptions in $L_1$ according to the scheme in \cite{yu2016exact}. 
	
	\item If $M\in \Big[\R_0+3\R'+\frac{3}{2}\R ,\,  \R_0+3(\R'+\R)\Big]$, the descriptions in $L_2$ and $L_3$ are fully cached at both receivers, and the descriptions in $L_1$ are cached according to the scheme in \cite{yu2016exact} over the remaining memory $M-\R_0-3\R'$.  
\end{itemize}	
%
The optimality of the cache allocation described above is proved in Appendix~\ref{app:achievable rate three}. The following theorem provides the peak rate achieved by the proposed MR scheme for this cache  allocation.
\begin{theorem}\label{thm:achievable rate three}
	In the three-file two-receiver network, for a given cache capacity $M$ and symmetric rate-tuple $\Rscr\in\GWregionS$, 
	the peak rate achieved by the proposed MR scheme is given by
	\begin{align} 
		& \Rach(M,\Rscr)  =  
		\begin{cases} 
			\R_0 + 3\R'+2\R -2M,   & M \in \Big[0,\, \frac{1}{2}\R' \Big)  \\
			\R_0 + \frac{5}{2}\R'+2\R -M, 	&M\in\Big[ \frac{1}{2}\R' ,  \, \R_0+\frac{3}{2}(\R'+\R)   \Big) \\
			\frac{2}{3}\R_0 + 2\R'+\frac{3}{2}\R -\frac{2}{3}M,   &M\in\Big[ \R_0+\frac{3}{2}(\R'+\R)  ,  \, \R_0+3\R'+\frac{3}{2}\R  \Big) \\
			\frac{1}{3}\R_0 + \R'+\R -\frac{1}{3}M,	 & M\in \Big[\R_0+3\R'+\frac{3}{2}\R ,\,  \R_0+3\R'+3\R\Big]. \label{eq:rate three}
		\end{cases} 
	\end{align}
\end{theorem}

\begin{proof}
	The proof is given  in Appendix~\ref{app:achievable rate three}. 
\end{proof}

\subsection{Two-Request CACM Scheme Adopted for Sublibrary $L_2$} \label{sec: new scheme}
In this section, we describe in detail the CACM scheme adopted for the common-to-two sublibrary $L_2$. As mentioned in the previous section, for a given cache allocation among $L_1$, $L_2$ and $L_3$, caching and delivery of the content are done in an independent fashion across the sublibraries, i.e., there is no coding across the sublibraries in either phase.  As a result, the scheme adopted for sublibrary $L_2$ needs to be designed for a network with two receivers and a library composed of three independent files (descriptions), $\{W_{12},W_{13},W_{23}\}$ with length $\R' F$ bits, where each receiver requests two files from the library. We refer to this network as the {\em two-request network}. Specifically, in the worst-case scenario, the demand from sublibrary $L_2$ consists of $i)$ one file that is requested by both receivers, and $ii)$ two files, each requested only by one of the receivers. With a slight abuse of notation, we denote receiver $r_k$'s demand as $d_{r_k}\in \Big\{\{12,13\},\{12,23\},\{13,23\}\Big\}$. 
While CACM schemes available in the literature such as the scheme proposed in \cite{yu2016exact}, where an uncoded prefetching strategy is adopted, are optimal for a single-request three-file two-receiver network, they fall short to achieve optimality in this multiple-request setting. Furthermore, the schemes in \cite{ji2015caching} and \cite{ji15efficient,ji14groupcast,daniel2017optimization,sengupta2017improved}, where each receiver requests more than one file, are designed for arbitrary demand combinations and could be suboptimal for the specific class of demands considered in our MR scheme. Hence, a new CACM design is needed. The proposed scheme described next uses coding in the content placement for small cache capacities to further leverage the caches for reducing the network load. 

\noindent{\bf Caching and Delivery Strategy:\\} 
In the following, we focus on a  
few cache capacity values, and for each of them we describe the receiver cache configurations $Z_1$ and $Z_2$ and quantify the corresponding achievable peak rate  
for any worst-case demand realization, i.e., $d_{r_1}\neq d_{r_2}$. 
To this end, we provide the multicast codeword transmitted by the sender for the specific demand realization  $d_{r_1} = \{12,\,13\}$ and $d_{r_2}  = \{12,\,23\}$. For all other worst-case demands, the multicast codeword can be constructed analogously, and results in the same delivery rate.



\begin{itemize}
	\item[$\square$] When $M=\frac{1}{2}\R'$, each file is split into two packets of length $\frac{1}{2}\R'$, and receiver caches are filled as
	\begin{align}
		&Z_{r_1} = \{W_{12}^{(1)}\oplus W_{13}^{(1)}\oplus W_{23}^{(1)}\},\quad Z_{r_2} = \{W_{12}^{(2)}\oplus W_{13}^{(2)}\oplus W_{23}^{(2)} \}, \notag
	\end{align}
	where $W_{s}^{(i)}$ denotes packet $i$ of description $W_s$.  Codeword $Y = \{W_{12}^{(1)},\;    W_{12}^{(2)},\;   W_{13}^{(2)},   \; W_{23}^{(1)}\}$
	enables both receivers to losslessly recover their requested packets as follows:
	\begin{itemize}
		\item[-] In addition to receiving $\{W_{12}^{(1)},\;    W_{12}^{(2)},\;   W_{13}^{(2)}  \}$, receiver $r_1$  can decode $ W_{13}^{(1)}$ by combining its cached content with the received packets $W_{12}^{(1)}$ and $W_{23}^{(1)}$.	
		\item[-] Similarly, $r_2$ receives the requested packets $\{W_{12}^{(1)},\;    W_{12}^{(2)},\;   W_{23}^{(1)}  \}$, and is also able to decode $W_{23}^{(2)}$ using its cache content and the transmitted packets $W_{12}^{(2)}$ and $W_{13}^{(2)}$.
	\end{itemize}


	This cache placement results in a rate equal to $2\R'$, whereas an uncoded prefetching scheme, such as the one in \cite{yu2016exact}, achieves a delivery rate of $\frac{7}{3}\R'$.
	
	\item[$\square$] When $M=\R'$,  the cached content is
	\begin{align}
		&Z_{r_1} = \{W_{12}^{(1)}\oplus W_{13}^{(1)},\; W_{12}^{(1)} \oplus W_{23}^{(1)}\},\quad Z_{r_2} = \{W_{12}^{(2)}\oplus W_{13}^{(2)},\; W_{12}^{(2)}\oplus W_{23}^{(2)} \},\notag
	\end{align}
	which is symmetric across the three descriptions,  since $W_{13}^{(i)}\oplus W_{23}^{(i)}$ can be obtained from combining $W_{12}^{(i)}\oplus W_{13}^{(i)}$ with $ W_{12}^{(i)} \oplus W_{23}^{(i)}$ for any $i\in\{1,2\}$. For the worst-case demand considered, receiver $r_1$ needs $\{W_{12}^{(1)},W_{12}^{(2)},W_{13}^{(1)},W_{13}^{(2)}\}$, and receiver $r_2$ needs $\{W_{12}^{(1)},W_{12}^{(2)},W_{23}^{(1)},W_{23}^{(2)}\}$.  Codeword $Y = \{W_{12}^{(1)},\;    W_{12}^{(2)},\;   W_{13}^{(2)} \oplus W_{23}^{(1)}\}$  
	enables both receivers to losslessly recover their demands as follows:
	\begin{itemize}
		\item[-] Receiver $r_1$ combines its cached content with the received packet $W_{12}^{(1)}$ to recover $W_{13}^{(1)}$, and from combining its cache with $\{W_{12}^{(1)},\, W_{13}^{(2)} \oplus W_{23}^{(1)}\}$, the other requested packet $W_{13}^{(2)}$ can be recovered.
		
		\item[-] Similarly, $r_2$ combines its cached content with the received packets $W_{12}^{(2)}$, and $\{W_{12}^{(1)},\, W_{13}^{(2)} \oplus W_{23}^{(1)}\}$ to decode $W_{23}^{(2)}$ and $W_{23}^{(1)}$, respectively.
	\end{itemize}
	This strategy results in a delivery rate equal to $\frac{3}{2}\R'$.
	
	\item[$\square$] When $M=\frac{3}{2}\R'$,  the caches are filled with the following uncoded content
	\begin{align}
		&Z_{r_1} = \{W_{12}^{(1)},\,  W_{13}^{(1)},\, W_{23}^{(1)}\},\quad Z_{r_2} = \{W_{12}^{(2)},\,  W_{13}^{(2)},\, W_{23}^{(2)} \}, \notag
	\end{align}
	for which codeword 
	$Y = \{W_{12}^{(1)}\oplus  W_{12}^{(2)},\;   W_{13}^{(2)} \oplus W_{23}^{(1)}\}$ with rate $\R'$ is sent. 
	
	
\end{itemize}

The CACM scheme described above provides an {\em optimal} placement and delivery strategy for the two-request  network  at the memory-rate pairs
\begin{align}
	(M,\,R)\in\Big\{    \Big(0,\,3\R'\Big),\,  \Big(\frac{1}{2}\R', 2\R'\Big),\, \Big(\R', \frac{3}{2}\R'\Big),\, \Big(\frac{3}{2}\R',\,\R'\Big),\, \Big(3\R',\,0\Big)    \Big\},\label{eq: opt points of TR}
\end{align}
which is proved in Appendix~\ref{app:tworequest}. As in \cite{maddah14fundamental}, through memory-sharing, the lower convex envelope of the points given above is achievable, resulting in the following peak delivery rate $ R_{L_2}(M,\R')$: 
\begin{align} 
	R_{L_2}(M,\R') = \begin{cases}
		3\R'-  2M ,                
		& \; M \in  [0, \frac{1}{2}\R'  )  \\
		\frac{5}{2}\R'-M  ,                
		& \; M \in  [\frac{1}{2}\R', \frac{3}{2}\R'  )\\
		2\R'-\frac{2}{3} M ,                
		& \;M \in  [\frac{3}{2}\R', 3 \R' ].  \label{eq:two request rate}
	\end{cases}
\end{align}
It is shown in Appendix~\ref{app:tworequest} that the peak rate 
given in \eqref{eq:two request rate} results in zero gap to optimality for $M\in [0,\frac{3}{2}\R']$.

%
%
%

\subsection{Optimality Results}
As in Sec.~\ref{sec:Order Optimality}, we evaluate the performance of the proposed GW-MR scheme by comparing its achievable rate $\RGW$, defined in \eqref{eq:ach peak GWMR}, with a lower bound on the optimal peak rate-memory function, $\Rstar$. 

\subsubsection{Lower Bound on $\Rstar$}\label{subsec: optimal LB three}
\begin{theorem}\label{thm:LowerBound three}
	In the three-file two-receiver network with library distribution $p(x_1,x_2,x_3)$, for a given  cache capacity $M$, a lower bound on $\Rstar$, the optimal peak rate-memory function, is given by 
	\begin{align}
		\Rlb =  \inf \bigg \{ R: \;\;
		R \;&\; \geq   \max_{i,j }  H(\Xsf_i, \Xsf_j)-2M , \notag\\
		R \;&\; \geq \frac{1}{2} \Big(\max_{i,j}  H(\Xsf_i, \Xsf_j )-M\Big), \notag\\
		R \;&\; \geq \frac{1}{3}\Big(H(\Xsf_1, \Xsf_2 , \Xsf_3)-M\Big), \notag\\
		R \;&\; \geq \frac{1}{2}\Big( H(\Xsf_1, \Xsf_2, \Xsf_3) +   \max_{i} H(\Xsf_i) \Big) - M \bigg \}. \notag
	\end{align}
\end{theorem}

\begin{proof} 
	The proof is given in Appendix \ref{app:LowerBound three}.
\end{proof}

\begin{remark}   
	Particularizing the outer bound in Theorem \ref{thm:LowerBound three}  to i.i.d. sources meets the bound derived in \cite{yu2017characterizing}. 
\end{remark}

\subsubsection{Optimality of the Proposed  GW-MR Scheme}~\\
The following theorem characterizes the performance of the proposed GW-MR scheme described in Sec.~\ref{sec: achievable three}  for different regions of $M$, and delineates the cache capacity region for which the scheme is optimal or near-optimal. 

\begin{theorem}\label{thm:optimality three}
	In the three-file two-receiver network, let $\tilde \Rscr = (\tilde \R_0, \tilde\R',\tilde\R',\tilde\R',  \tilde\R,  \tilde\R,  \tilde\R) $ be a symmetric rate-tuple in the three-file Gray-Wyner rate region, such that
	\begin{align}\label{eq:rho tilde}
		\tilde \Rscr \in \Big\{   \Rscr \in \GWregionS :\;  \R_0+ 3 \R'+ 3\R =H(\Xsf_1, \Xsf_2,\Xsf_3) ,\;  \R \text{ is maximized}\Big\}
	\end{align}
	Then, for $M\in   \Big[  H(\Xsf_1, \Xsf_2,\Xsf_3)-\frac{3}{2}\tilde\R, \,  H(\Xsf_1, \Xsf_2,\Xsf_3) \Big]$, the proposed GW-MR scheme is optimal, i.e., $\RGW= \Rstar$.
	In addition, for $M\in \Big[ 0,\,  \tilde\R_0+ \frac{3}{2} (  \tilde\R'+ \tilde\R ) \Big) $, 
	\begin{align}
		&\RGW - \Rstar \leq  \frac{1}{2} \min_i H(\Xsf_j,\Xsf_k|\Xsf_i) -\tilde\R   , \notag
	\end{align}
	and for $M\in \Big[  \tilde\R_0+ \frac{3}{2} ( \tilde\R'+ \tilde\R ), \,  H(\Xsf_1, \Xsf_2,\Xsf_3)-\frac{3}{2}\tilde\R  \Big) $, we have
	\begin{align}
		&\RGW - \Rstar \leq  \frac{1}{4}\min_i H(\Xsf_j,\Xsf_k|\Xsf_i) - \frac{1}{2} \tilde\R . \notag
	\end{align}
\end{theorem}

\begin{proof}
	The proof is given in Appendix~\ref{app:optimality three}.
\end{proof}

\begin{remark}\label{rmk: not wyner common three} 
	Theorem \ref{thm:optimality three} suggests that operating at a symmetric point for which $\R_0+3\R'+3\R =$ $H(\Xsf_1,$ $\Xsf_2,\Xsf_3)$, and where the rate corrsponding to the descriptions in the private sublibrary is maximized, increases the memory region where the GW-MR scheme is optimal, and also decreases the gap to optimality for other values of the memory. This is analogous to the conclusion reached in Remark~\ref{rmk: not wyner common}, where we showed that for two correlated files it is desirable to maximize the rate of the smallest private description subject to an equivalent condition on the sum rate of the entire descriptions.  
\end{remark}

The next corollary particularizes Theorem~\ref{thm:optimality three} for a specific 3-DMS.
\begin{corollary}\label{cor:special source three}
	Consider a three-file two-receiver network with $\Xsf_1 = (\Vsf,\Usf_1,\Usf_2,\Xsf'_1)$, $\Xsf_2 =(\Vsf,\Usf_2,\Usf_3,\Xsf'_2)$, and
	$\Xsf_3 =(\Vsf, \Usf_1,\Usf_3,\Xsf'_3)$, such that 
	\begin{align}
		H(\Xsf'_1)=H(\Xsf'_2)=H(\Xsf'_3)=H_x, \;\;\;
		H(\Usf_1)=H(\Usf_2)=H(\Usf_3)=H_u,  \notag
	\end{align} 
	and where $\{\Xsf'_1$, $\Xsf_2'$, $\Xsf'_3$, $\Usf_1$, $\Usf_2$, $\Usf_3$, $\Vsf\}$ are all mutually independent. The proposed GW-MR scheme  is optimal under the peak rate criterion, i.e., $\RGW = \Rstar$, when 
	$$M\in \Big[0,\; \widetilde M\Big] \bigcup\Big[\widetilde M +\frac{3}{2}H_u ,\;H(\Xsf_1,\Xsf_2,\Xsf_3)
	\Big],\quad  \widetilde M \triangleq H(\Vsf)+\frac{3}{2}H_u+ \frac{3}{2} H_x.$$
	For any other $M$, we have
	\begin{align}
		&\RGW- \Rstar \leq  \frac{1}{4}  H_u  . \notag
	\end{align} 	
\end{corollary}
\begin{proof}
	See Appendix~\ref{app:special source three}.
\end{proof}

\begin{remark}
	For the 3-DMS considered in Corollary~\ref{cor:special source three}, operating the Gray-Wyner network at a point $\Rscr\in\GWregionS$ such that $\R_0 = H(\Vsf)$, $\R'=H_u$ and $\R = H_x$, leads to a set of descriptions that is equivalent to the library considered in \cite{yang2017centralized}, where each file is composed of four independent descriptions: $i)$ one that is common among all files,  $ii)$ two descriptions that are common with only one other file, and finally $iii)$ one that is exclusive to that file. While the scheme proposed in \cite{yang2017centralized} is designed for arbitrary number of files and receivers, its performance evaluation requires numerical optimization. For the three-file two-receiver setting, we are able to derive a closed-form expression for the achievable peak rate in Theorem~\ref{thm:achievable rate three}, and we establish its optimality in Corollary~\ref{cor:special source three}. In Sec.~\ref{subsec: sim three}, we numerically show that our proposed GW-MR scheme outperforms the scheme in \cite{yang2017centralized}.

\end{remark}

\section{Illustration of Results}\label{sec:numerical}
In this section, we numerically illustrate the results derived in the previous sections for the two-file and three-file networks under symmetric binary sources.  
\subsection{Two-File Network}
Consider, as a 2-DMS, a doubly symmetric binary source (DSBS) with joint pmf  
$$p(x_1,x_2) = \frac{1}{2} (1-p_0)\delta_{x_1,x_2} + \frac{1}{2} p_0(1-\delta_{x_1,x_2}),  \qquad x_1,x_2 \in\{0,1\},$$
and parameter $p_0\in[0,\frac{1}{2}]$.
Then,
\begin{align}
	& H(\Xsf_1) = H(\Xsf_2) = 1, \quad
	H(\Xsf_1|\Xsf_2) = H(\Xsf_2|\Xsf_1) = h(p_0),\quad
	H(\Xsf_1,\Xsf_2) = 1 + h(p_0),\notag
\end{align}
where $h(p) = -p \log(p)-(1-p)\log(1-p)$ is the binary entropy function. Even for the simple source considered here, the optimal Gray-Wyner rate region is not known. An achievable  region for a DSBS restricted to the plane $\{(\R_0,\R_1,\R_2) : \R_1=\R_2=\R\}$ was derived in \cite{wyner1975common}, which is described by the set of rate triplets  $(\R_0,\R,\R)$ with $\R_0$ given by
\begin{equation}
	\R_0 \geq \begin{cases} 1+ h(p_0) - 2\R,              &\hspace{0.2cm} 0 \leq \R < h(p_1) \\
		f(\R) 			                                &\hspace{0.2cm} h(p_1)\leq \R \leq 1
	\end{cases},
\end{equation}
where $p_1=\frac{1}{2}(1-\sqrt{(1-2p_0)})$, and
\begin{align}
	f(\R)  \triangleq 1+h(p_0)+ p_0 \log\Big(\frac{p_0}{2}\Big) &+ \Big(h^{-1}(\R)-\frac{p_0}{2} \Big)\log\Big(h^{-1}(\R)-\frac{p_0}{2} \Big)  \notag \\
	& + \Big( 1-h^{-1}(\R)-\frac{p_0}{2} \Big)\log\Big( 1-h^{-1}(\R)-\frac{p_0}{2} \Big).\notag
\end{align}

We compare the peak and average rates achieved by the proposed GW-MR scheme, $\RGW$ and $\RGWE$, given in Theorems~\ref{thm:optimality peak} and \ref{thm:optimality avg}, with: $i)$ {\em Correlation-Unaware CACM}, which refers to the best known CACM scheme proposed for independent files (a combination of \cite{tian2016caching} and \cite{yu2016exact} for  peak rate, and \cite{yu2016exact} for  average rate), $ii)$ {\em MR Lower Bound}, which refers to the lower bound on the MR rate-memory functions $\RGWstarR$ and $\RGWstarRE$, given in Theorem~\ref{thm:LowerBoundGW}, for the operating point $\Rscr$ of the proposed GW-MR scheme, i.e., the $\Rscr $ that minimizes $\Rach(M,\Rscr)$ and $\RachE(M,\Rscr)$ (see \eqref{eq:ach peak GWMR} and \eqref{eq:ach avg GWMR}), and $iii)$ {\em Optimal Lower Bound}, which refers to the lower bound on the optimal rate-memory functions $\Rlb$ and $\RlbE$,  given in Theorem~\ref{thm:LowerBound}.

The peak and average rate-memory trade-offs in a two-file network with $K=5$ receivers are shown in Fig.~\ref{fig:simulations K}  for a DSBS with $p_0 = 0.2$. In line with Theorem \ref{thm:optimality peak} and for $M_K$ defined in \eqref{eq:Mk}, Fig.~\ref{fig:simulations K} shows that the proposed GW-MR scheme meets the optimal peak rate-memory function $\Rstar$, when $M \leq  {M}_5=  0.1$ and $M\geq ( H(\Xsf_1,\Xsf_2) - 2 {M}_5) = 1.52$, and it meets the optimal average rate-memory function $\RstarE$, when $M\geq ( H(\Xsf_1,\Xsf_2) - 2  {M}_5) = 1.52$. Note that the correlation-unaware CACM is strictly suboptimal for all $M$.

\begin{figure} 
	\begin{subfigure}{0.5\linewidth}\centering
		\includegraphics[width=\linewidth]{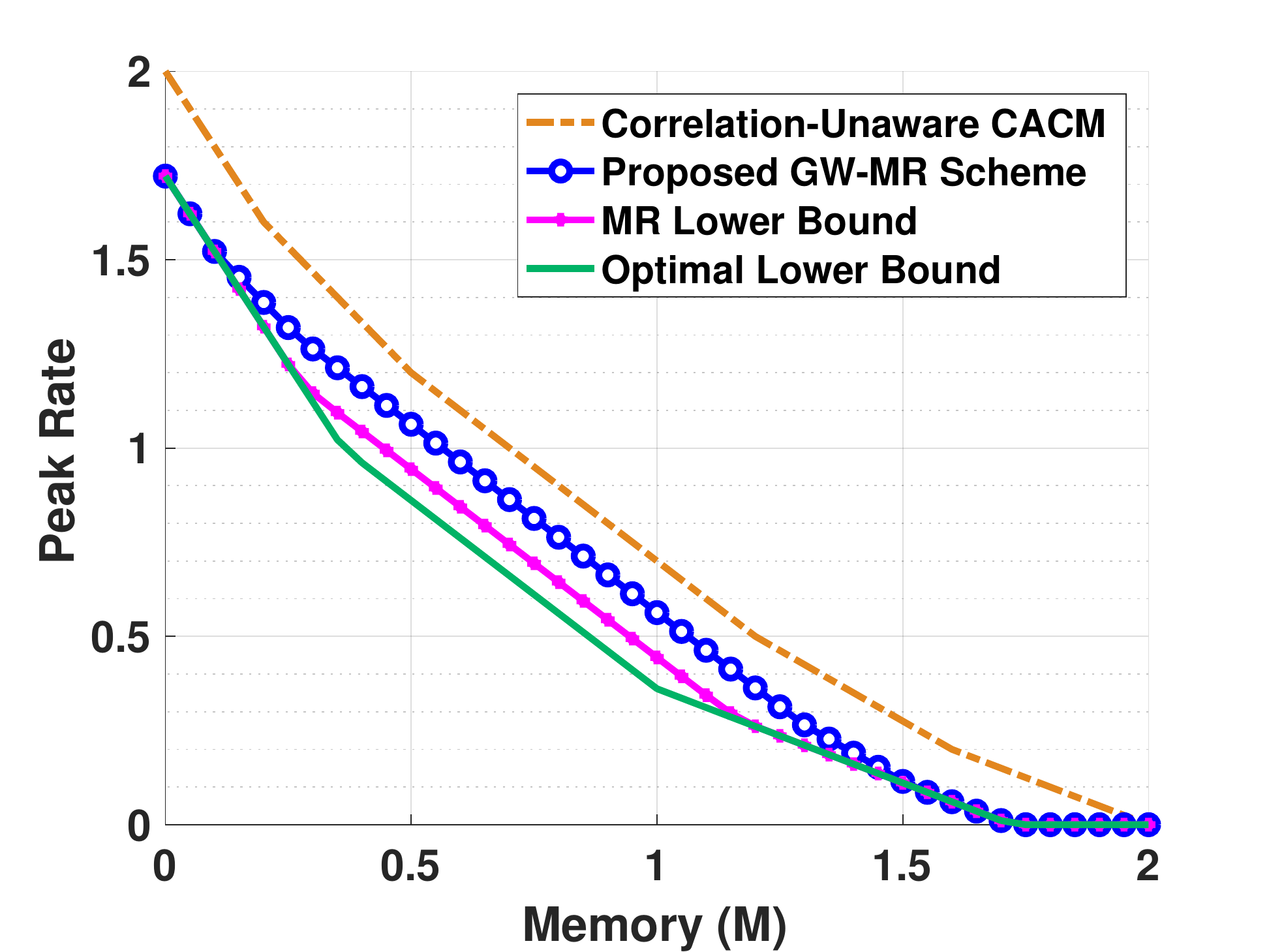}
	\end{subfigure}\hspace*{\fill}
	\begin{subfigure}{0.5\linewidth}\centering
		\includegraphics[width=\linewidth]{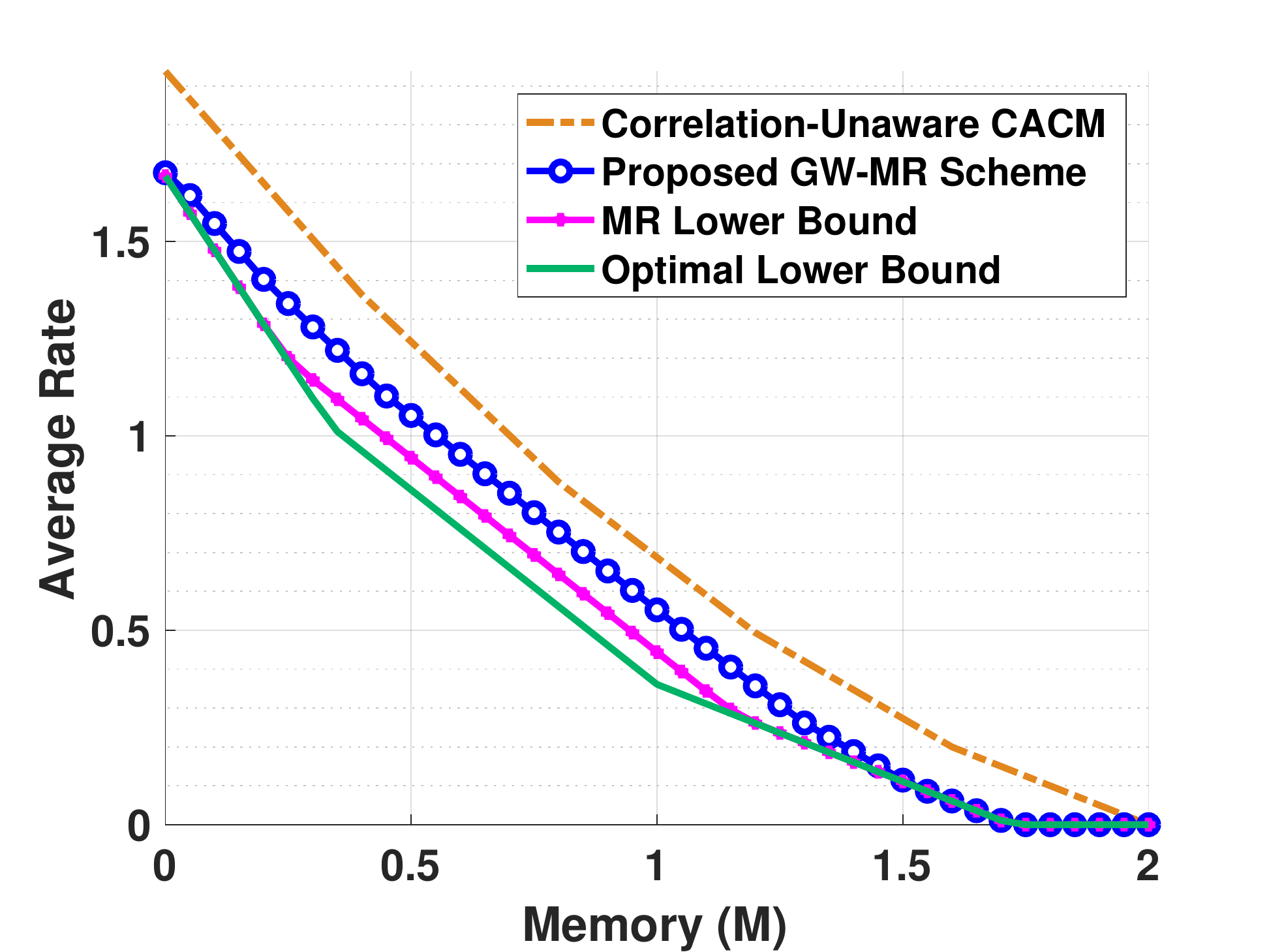}
	\end{subfigure}\hspace*{\fill}
	\caption{Rate-memory trade-off for a DSBS with $p_0 = 0.2$, $N=2$ files and $K=5$ receivers.}
	\label{fig:simulations K}
\end{figure}

Fig.~\ref{fig:simulations 2}  displays the peak and average rate-memory trade-offs in a two-file two-receiver network with $p_0 = 0.2$. In line with Theorem \ref{thm:optimality peak}, it is observed that the proposed GW-MR scheme meets the lower bound on the optimal peak rate-memory function $\Rstar$ (and is hence optimal) when $M \leq {M}_2=  0.25$ and $M\geq ( H(\Xsf_1,\Xsf_2) - 2 {M}_2) = 1.21$. In terms of average rate, and in line with Theorem \ref{thm:optimality avg}, the GW-MR scheme meets the lower bound on the average rate-memory function $\RstarE$ (and is hence optimal) when $M\geq ( H(\Xsf_1,\Xsf_2) - 2 {M}_2) = 1.21$. 
In addition, it is observed from the figure that  the proposed GW-MR scheme achieves the lower bounds on the MR peak and average rate-memory functions, $\RGWRlb$ and $\RGWRlbE$, respectively, for all values of the memory, which is in agreement with  Corollaries~\ref{col:peak and lower bound 2} and \ref{col:average and lower bound 2}.
\begin{figure} 
	\begin{subfigure}{0.5\linewidth}\centering
		\includegraphics[width=\linewidth]{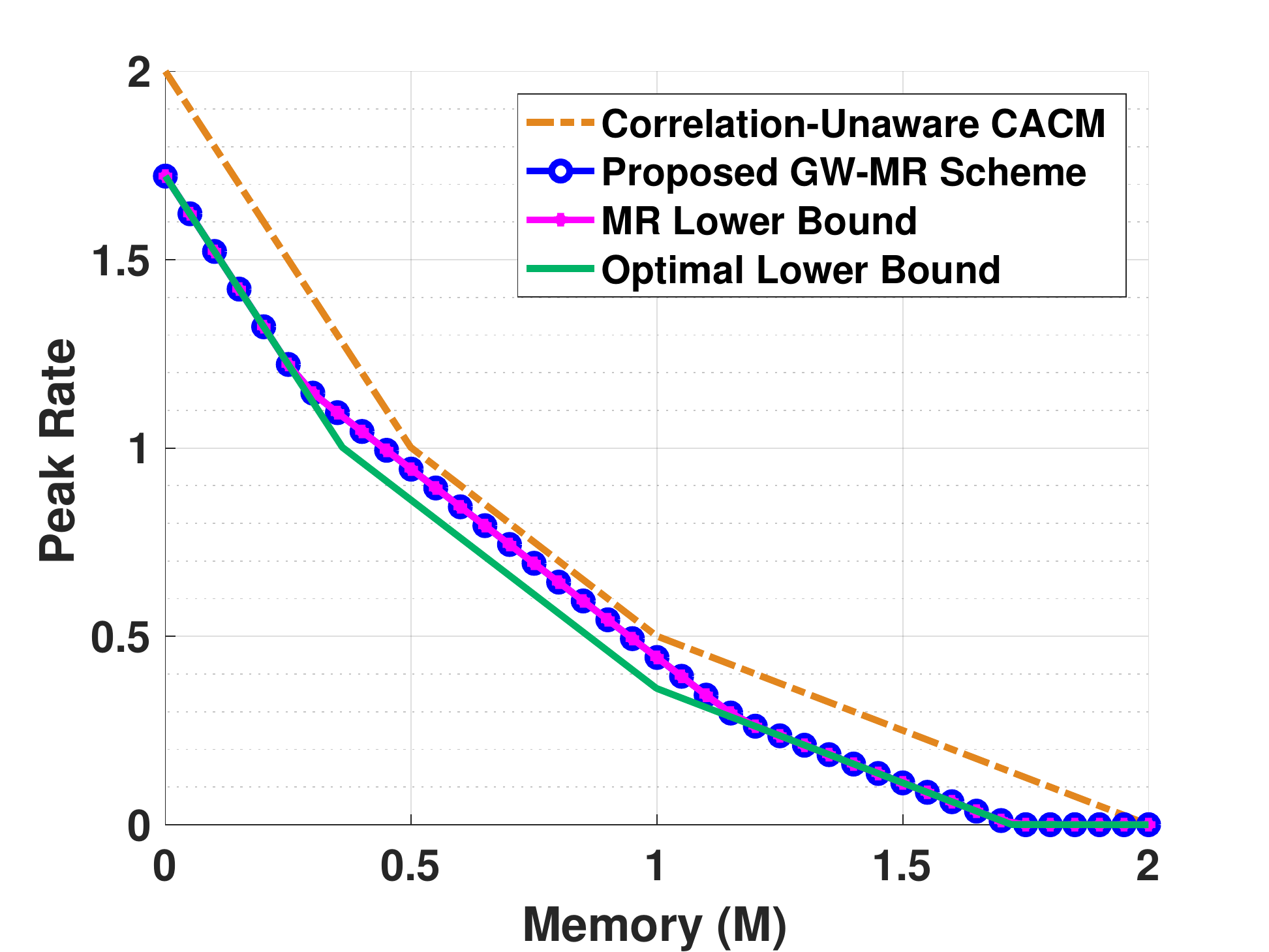}
	\end{subfigure}\hspace*{\fill}
	\begin{subfigure}{0.5\linewidth}\centering
		\includegraphics[width=\linewidth]{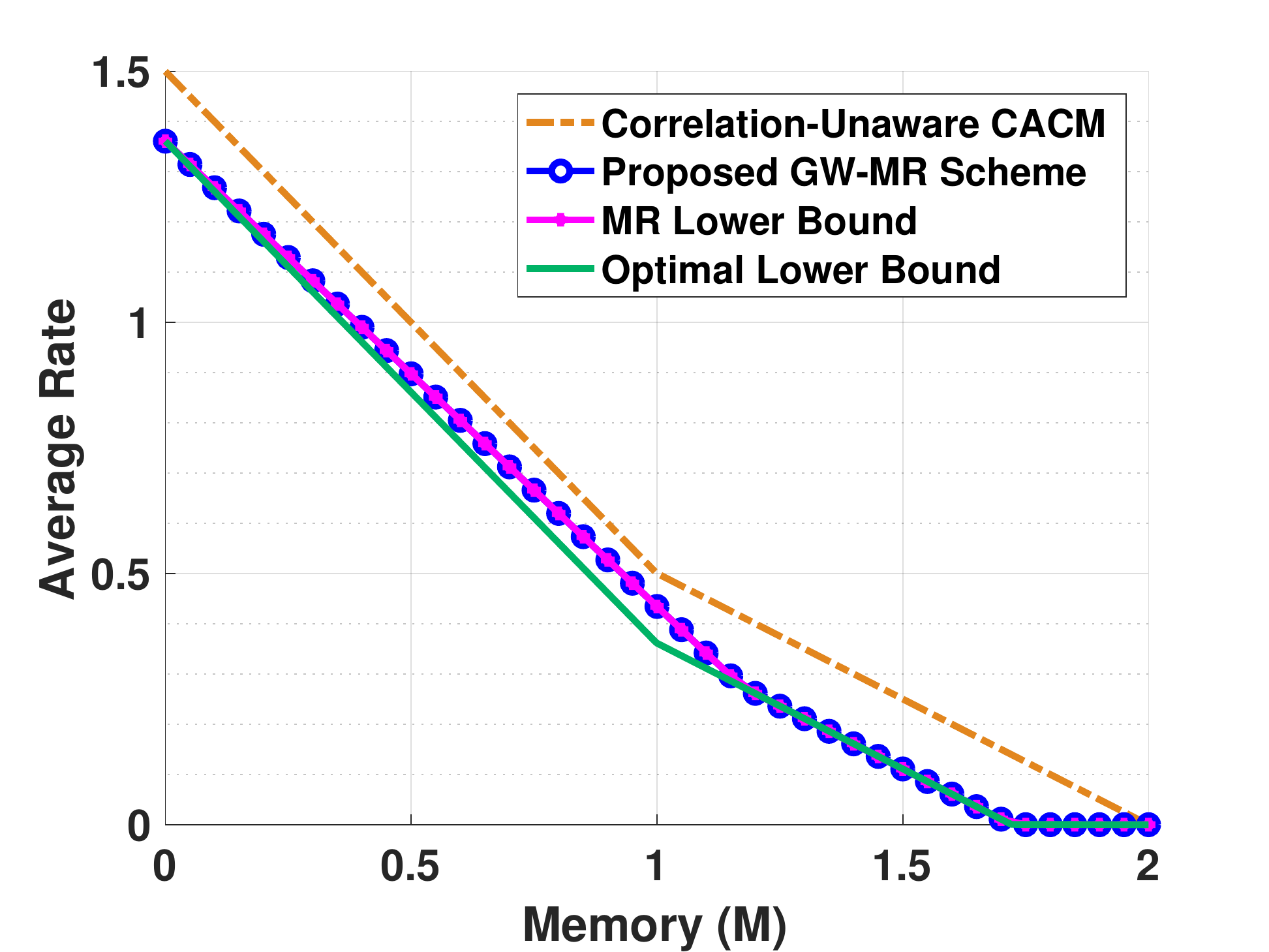}
	\end{subfigure}\hspace*{\fill}
	\caption{Rate-memory trade-off for a DSBS with $p_0 = 0.2$, $N=2$ files and $K=2$ receivers.
	} 
	\label{fig:simulations 2}
\end{figure}

\subsection{Three-File Network}\label{subsec: sim three}
For the three-file network we consider a $3$-DMS as follows. Let $\Vsf \sim$ Bern$(0.5)$, and $\Xsf_1$, $\Xsf_2$ and $\Xsf_3$, be the outputs of independent Binary Symmetric Channels (BSCs) with crossover probabilities $p_0 \in[0,0.5]$, fed by the same input $\Vsf$. Therefore,
$$p(x_1,x_2,x_3|v) =  p_1(x_1|v)\,p_2(x_2|v)\,p_3(x_3|v),\quad x_1,x_2,x_3 \in\{0,1\},$$
where 
\begin{equation}
	p_i(x_i|v) = \begin{cases}
		1-p_0 ,      & x_i = v  \\
		p_0,	    &x_i \neq v . \notag
	\end{cases}
\end{equation}
The joint pmf of the binary source is 
\begin{align}
	p(x_1,x_2, x_3) = \frac{1}{2}   {p_0}^s (1-p_0)^{(3-s)} + \frac{1}{2} {p_0}^{(3-s)} (1-p_0)^s, \label{eq:3-DMS}
\end{align}
with $s =x_1+x_2+x_3$. 
Similar to the derivation in \cite{wyner1975common}, an achievable Gray-Wyner rate region for this 3-DMS restricted to the symmetrical plane $\{\Rscr: \R_{12}=\R_{13} = \R_{23}=\R',\, \R_1=\R_2 = \R_3=\R\}$,  can be described by the following set of rate-tuples
\begin{align}\label{eq: lousy GW}
	(\R_0,\R',\R) \in \Big\{&\Big( H(\Xsf_1,\Xsf_2,\Xsf_3),\,0,\,0 \Big),\,\Big(0,\,0,\, H_x\Big),\,\Big(H_x,\,\frac{1}{3}H(\Xsf_2|\Xsf_1),\, \frac{1}{3}H(\Xsf_3|\Xsf_1,\Xsf_2)\Big),\,\notag\\
	&\Big(H_x,\,0,\, \frac{2}{3}H(\Xsf_2|\Xsf_1)\Big) 		\Big\},
\end{align}
which are achieved as follows:
\begin{itemize}
	\item $\Big( H(\Xsf_1,\Xsf_2,\Xsf_3),\,0,\,0 \Big)$ is achieved by simply transmitting $(\Xsf_1^F,\Xsf_2^F,\Xsf_3^F)$ over the common-to-all link without using the common-to-two and private links.
	\item $\Big(0,\,0,\, H_x\Big)$ is achieved by transmitting $\Xsf_1^F$, $\Xsf_2^F$ and $\Xsf_3^F$ over the private links to Gray-Wyner decoders $1$, $2$ and $3$, respectively, without using any of the common links. 
	\item $\Big(H_x,\,\frac{1}{3}H(\Xsf_2|\Xsf_1),\, \frac{1}{3}H(\Xsf_3|\Xsf_1,\Xsf_2)\Big)$ is achievable since the following three non-symmetric  points are achievable:
	\begin{itemize}
		\item[-] $\Rscr_1:\, \R_0 = H(\Xsf_1), \R_{12}=\R_{13} =0, \R_{23} =H(\Xsf_2|\Xsf_1),\R_{1}=\R_{2} =0, \R_{3} =H(\Xsf_3|\Xsf_1,\Xsf_2) $.
		\item[-] $\Rscr_2:\, \R_0 = H(\Xsf_2), \R_{12}=\R_{23} =0, \R_{13} =H(\Xsf_3|\Xsf_2),\R_{2}=\R_{3} =0, \R_{1} =H(\Xsf_1|\Xsf_2,\Xsf_3) $.
		\item[-] $\Rscr_3:\, \R_0 = H(\Xsf_3), \R_{13}=\R_{23} =0, \R_{12} =H(\Xsf_1|\Xsf_3),\R_{1}=\R_{3} =0, \R_{2} =H(\Xsf_2|\Xsf_1,\Xsf_3) $. 
	\end{itemize}  
	To see this, we describe how point $\Rscr_1$ can be achieved. By transmitting $W_{123}=\Xsf_1^F$ over the common-to-all link, and then transmitting $W_{23}$ over the common-to-two link to decoders $2$ and $3$ (with rate  $H(\Xsf_2|\Xsf_1)$), decoder $1$ losslessly reconstructs $\Xsf_1^F$, and decoders $2$ and $3$ reconstruct $\Xsf_2^F$. Finally, by transmitting $W_3$ over the private link to decoder $3$ (with rate $H(\Xsf_3|\Xsf_1,\Xsf_2)$), it can reconstruct file $\Xsf_3^F$. Given the achievability of the three points $\Rscr_1,\Rscr_2$ and $\Rscr_3$, their centroid   is also achievable, and it  lies on the symmetric Gray-Wyner rate region $\GWregionS$. 
	
	\item $\Big(H_x,\, 0,\, \frac{2}{3}H(\Xsf_2|\Xsf_1)\Big)$ is achievable since the following three non-symmetric points are achievable:
	\begin{itemize}
		\item[-] $\Rscr_4:\, \R_0 = H(\Xsf_1),\, \R_{12}=\R_{13} =\R_{23}=0,\,\R_{1}=0, \R_{2} =H(\Xsf_2|\Xsf_1)   ,\, \R_{3} =H(\Xsf_3|\Xsf_1) $.
		\item[-] $\Rscr_5:\, \R_0 = H(\Xsf_2), \,\R_{12}=\R_{13} =\R_{23}=0,\,\R_{1}=H(\Xsf_1|\Xsf_2),\, \R_{2} =0, \R_{3} =H(\Xsf_3|\Xsf_2) $.
		\item[-] $\Rscr_6:\, \R_0 = H(\Xsf_3), \,\R_{12}=\R_{13} =\R_{23}=0,\, \R_{1}=H(\Xsf_1|\Xsf_3),\, \R_{2} =H(\Xsf_2|\Xsf_3) , \R_{3} =0$.
	\end{itemize}
	To see this, we describe how point $\Rscr_4$ can be achieved. By transmitting $W_{123}=\Xsf_1^F$ over the common-to-all link, all three decoders can losslessly reconstruct $\Xsf_1^F$. Then, by transmitting $W_{2}$ and $W_{3}$ over the private links to decoders $2$ (with rate  $H(\Xsf_2|\Xsf_1)$) and decoder $3$ (with rate  $H(\Xsf_3|\Xsf_1)$), respectively, decoder $2$ reconstructs $\Xsf_2^F$, and decoder $3$  reconstructs $\Xsf_3^F$. Given the achievability of the three points $\Rscr_4,\Rscr_5$ and $\Rscr_6$, their centroid   is also achievable, and it  lies on the symmetric Gray-Wyner rate region $\GWregionS$.

\end{itemize}

%
%
%
%

Figs.~\ref{fig:3DMS}(a) and \ref{fig:3DMS}(b) display the rate-memory trade-off
as the memory size varies for a network with three files and two receivers. Fig.~\ref{fig:3DMS}(a) considers a 3-DMS distributed as \eqref{eq:3-DMS} with $p_0=0.05$, and compares the peak rate achieved with the proposed GW-MR scheme, $\RGW$, with the lower bound on $\Rstar$ given in Theorem~\ref{thm:LowerBound three}, and the optimal CACM scheme proposed for independent files in  \cite{yu2016exact}. The peak rate achieved by the GW-MR scheme,  defined in \eqref{eq:ach peak GWMR three}, is computed with respect to the achievable Gray-Wyner rate region described by the points in \eqref{eq: lousy GW}. From the figure, we observe that the proposed scheme is optimal when $M\geq 1.63$, which coincides with Theorem~\ref{thm:optimality three}, for which  $\tilde\R=0.12$ and  $H(\Xsf_1,\Xsf_2,\Xsf_3)-\frac{3}{2}\tilde\R =1.81-\frac{3}{2}(0.12)=1.63$. Moreover, the gap to optimality is less than $\frac{1}{2}H(\Xsf_1,\Xsf_2|\Xsf_3)-\tilde \R=0.29$, as stated in Theorem~\ref{thm:optimality three}. We note that the results in Fig.~\ref{fig:3DMS}(a) correspond to the achievable rate region given in  \eqref{eq: lousy GW}, however, a smaller gap can be achieved with an improved Gray-Wyner rate region.

\begin{figure}[h!]
	\begin{subfigure}{0.5\linewidth}\centering
		\includegraphics[width= \linewidth]{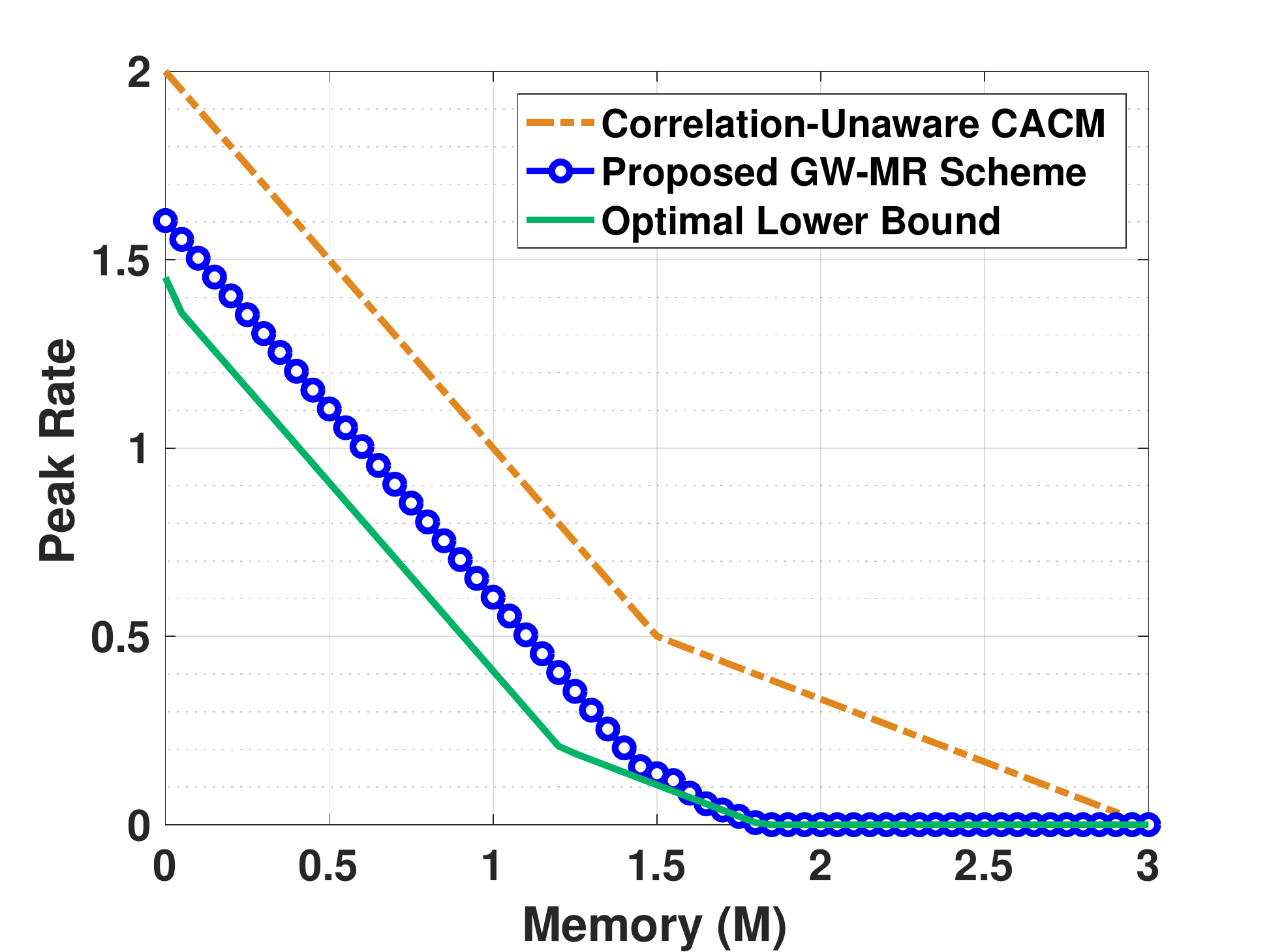}
		\subcaption{}
	\end{subfigure}\hspace*{\fill}
	\begin{subfigure} {0.5\linewidth}\centering
		\includegraphics[width= \linewidth]{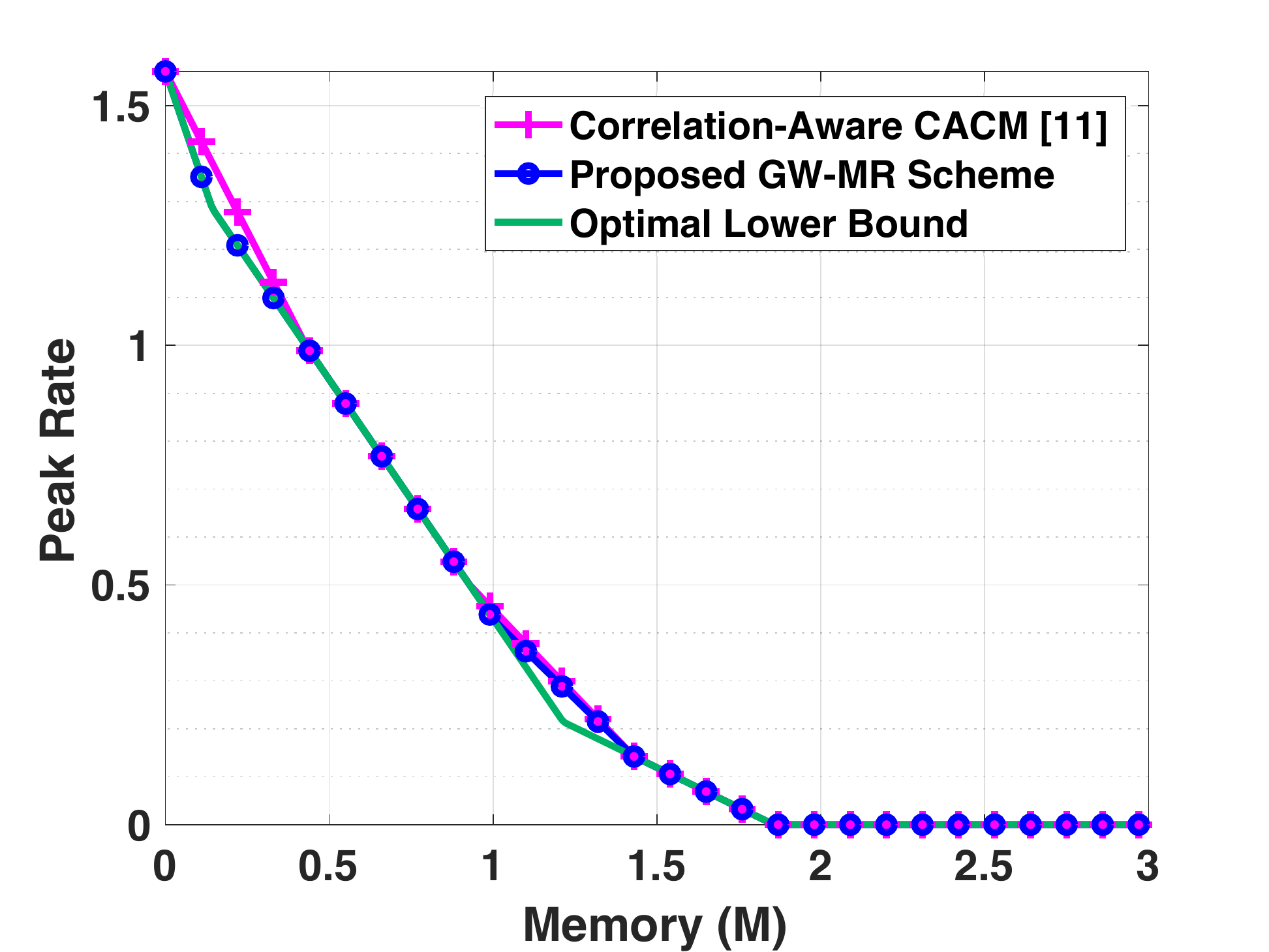}
		\subcaption{}
	\end{subfigure}\hspace*{\fill}
	\caption{Peak rate-memory trade-off for $N=3$, $K=2$, and the 3-DMS: (a) with distribution $p(x_1,x_2,x_3)$ given in \eqref{eq:3-DMS} for $p_0 = 0.05$, and (b) considered in Corollary~\ref{cor:special source three}.} 
	\label{fig:3DMS}
\end{figure}

Fig.~\ref{fig:3DMS}(b) displays the results for the 3-DMS given in Corollary~\ref{cor:special source three} for $H_x=H_u = 2H(\Vsf)=1$. It is observed that the proposed GW-MR scheme is optimal for all $M\in[0,1]\cup[1.43,3]$, and it is within a $0.07$ gap to optimality for other memory sizes, which is in line with the results in Corollary~\ref{cor:special source three}, for which $\widetilde M =1$. The plot shows that our proposed GW-MR scheme outperforms the correlation-aware scheme in \cite{yang2017centralized} for all memory sizes.


\section{Concluding Remarks}\label{sec:Conclusions}
In this paper, we have studied the fundamental rate-memory trade-off  in cache-aided communication systems under the assumption of correlated content. We 
have proposed a class of schemes based on a two-step source coding approach, in which files are first compressed using Gray-Wyner source coding, and then, the encoded descriptions are treated as independent content by a multiple-request cache-aided coded multicast scheme. We have characterized the rate-memory trade-off of such schemes, and analyzed the optimality of the overall proposed scheme with respect to a lower bound and proved its optimality or quantified its gap to optimality for different memory regimes. For the two-file network, we have shown that the two-step scheme is optimal for small and large caches sizes, and it is within half of the conditional entropy for all other memory values. The optimality holds for the regime of large memory in the three-file network, and the gap to optimality is less than the joint entropy of two of the sources conditioned on the third
source elsewhere.

\begin{appendices}
	
	\section{Optimal Cache Allocation  for Two Files  and Proof of Theorem \ref{thm:achievable rate peak}}\label{app:achievable rate}
	In this Appendix we prove the optimality of the cache allocation described in Sec.~\ref{subsec: scheme} and compute the corresponding MR peak and average rates given in Theorem~\ref{thm:achievable rate peak}. Recall that as explained in Sec.~\ref{subsec: scheme}, the proposed MR scheme for two files caches and delivers content from sublibraries $L_1$ and $L_2$  independently, as follows:
	\begin{itemize}
		\item Common Sublibrary $L_2$: Since the common description $\Wsf_0$ is required for the lossless reconstruction of either file, for any demand vector $\dbf$, it is required by all of the receivers. Therefore, for the common sublibrary it is optimal to adopt LFU caching and naive multicast delivery. Let $\mu_0\in [0, \min\{M,\R_0\}]$ denote the portion of memory allocated to sublibrary $L_2$. Then, each receiver caches the first $\mu_0F$ bits of $\Wsf_0$, and for any demand  the remaining $(\R_0 - \mu_0)F$ bits are delivered through uncoded multicast transmissions. 
		
		\item Private Sublibrary $L_1$: Let $\mu \in [0,$ $\min\{M,\R_1+\R_2\}]$ denote the portion of the memory allocated to sublibrary $L_1$. Then, the private descriptions $\{\Wsf_1,\Wsf_2\}$, with rates $\R_1$ and $\R_2$, are cached and delivered based on a correlation-unaware CACM scheme. Let $\RX(\mu,\R_1,\R_2)$ and $\REX(\mu,\R_1,\R_2)$ denote the achievable peak and average rates, respectively, for cache capacity $\mu$. 
	\end{itemize}
	Under the peak rate criterion, for a given cache capacity $M$, the cache allocations $\mu_0$ and $\mu$ satisfy $\mu_0+\mu = M$, and are optimally determined from the following optimization problem: 
	\begin{equation} \label{eq:cache allocation}
		\begin{aligned}
			\Rach(M,\Rscr) \;=\;&\min\limits_{\mu_0,\, \mu}
			& &   \R_0 - \mu_0 \,+\, \RX(\mu,\R_1,\R_2)\\
			& \text{s.t}
			& & \mu_0+\mu  \leq M, \\
			&&& 0 \leq \mu_0\leq \R_0,	\\
			&&& 0 \leq \mu\leq \R_1+\R_2	.
		\end{aligned}
	\end{equation}
	For any $M\in\Big[0,\,\R_0+\R_1+\R_2\Big]$, using the classical Lagrange multipliers method, and for 	
	\begin{align}
		{M}^*  \triangleq \min\Big\{M:\, \Big|\frac{\partial_{-}}{\partial M}\RX(M,\R_1,\R_2) \Big| < 1\Big\},  
	\end{align}
	the optimal solution to \eqref{eq:cache allocation} is as follows
	\begin{equation}\label{eq:opt mu}
		(\mu_0^*,\, \mu^*) =
		\begin{cases}
			\Big(  (M-\R_1-\R_2)^+     ,\, \min\{M,\R_1+\R_2\}  \Big) ,  \quad &\text{if } M\in  [0,\,M^*)  \\
			\Big( M-M^* ,\,  \min\{M^*,\R_1+\R_2\}  \Big) , &\text{if }  M\in [M^*,\,   \R_0 + M^*  ]
			\\
			\Big( \min\{M,\R_0\} ,\, (M-\R_0)^+  \Big)  , &\text{if }   M \in [\R_0 + M^*,\,  \R_0 +\R_1+\R_2]
		\end{cases}
	\end{equation}
	and the achievable rate provided in Theorem \ref{thm:achievable rate peak} is given by 
	$\Rach(M,\Rscr)= \R_0 - \mu_0^* \,+\, \RX(\mu^*,\R_1,\R_2)$ for $\mu_0^*$ and $\mu^*$ given in \eqref{eq:opt mu}.
	
	The slope of the tangent line to the rate functions $\RX(\mu,\R_1,\R_2)$ and $\R_0-\mu_0$ corresponding to the schemes used  for sublibraries $L_1$ and $L_2$, respectively, is a measure of their effectiveness in reducing the overall delivery rate. Given that $\RX(\mu,\R_1,\R_2)$ is a monotonically decreasing convex function, for small cache sizes, its slope (in absolute value) is larger than the one corresponding to the common sublibrary. Therefore, for $M<M^*$ it is preferable to allocate the entire cache capacity to the private sublibrary. As the memory increases the relative order of the two slopes switches such that the multicast delivery of the common description becomes more effective in reducing the  rate. At this point the additional memory is allocated to the common sublibrary, until it is fully stored.

	Under the average rate criterion, the optimal cache allocation is derived through an optimization problem similar to \eqref{eq:cache allocation}, but with respect to $\REX(\mu,\R_1,\R_2)$. This optimization leads to a solution $(\bar\mu_0^*,\bar\mu^*)$  similar to  \eqref{eq:opt mu} with $M^*$ replaced by $\bar M^*$ defined in \eqref{eq:M bar star}. Then the MR average rate in Theorem~\ref{thm:achievable rate peak} is given by $\RachE(M,\Rscr)= \R_0 - \bar\mu_0^* \,+\, \REX(\bar\mu^*,\R_1,\R_2)$.


	\section{Proof of Theorem \ref{thm:ach and GW lower bound peak}}\label{app:ach and GW lower bound peak}
	
	The proof follows from comparing $\Rach(M, \Rscr)$, the peak  rate achieved by the proposed MR scheme given in Theorem~\ref{thm:peak rate MR}, with $\RGWRlb$, the lower bound on the MR peak rate-memory function given in Theorem~\ref{thm:LowerBoundGW}, over different regions of the memory. 
	As per Theorem~\ref{thm:LowerBoundGW}, a lower bound on $\RGWstarR$ for the setting with  two files and $K$ receivers is given by 
	\begin{equation}
		\RGWRlb =
		\begin{cases}
			\R_0+\R_1+\R_2- 2M ,        
			& \; M \in \Big[0,\, \gamma\Big)  \\
			\R_0+ \frac{1}{2}\Big(\R_1+\R_2+\max\{ \R_1,\R_2\} \Big)- M ,				   
			& \;  M \in \Big[\gamma, \, \lambda \Big)\\
			\frac{1}{2}\Big( \R_0+\R_1+\R_2 - M\Big)  ,				     
			& \;  M \in \Big[\lambda, \, \R_0+\R_1+\R_2\Big]   \label{eq: LB GW peak appendix}
		\end{cases} 	
	\end{equation}
	where
	\begin{align}
		& \gamma \triangleq  \frac{1}{2} \min\{\R_1,\R_2\},  \quad \lambda \triangleq \R_0+\R_1+\R_2 - 2 \gamma= \R_0 + \max\{\R_1,\R_2\}  . \label{eq: gamma lambda}
	\end{align}
	
	For $\gamma_K$ and $\lambda_K$ defined in \eqref{eq: gamma lambda K}, when $K\geq 2$ we have $\gamma_K\leq\gamma\leq\lambda\leq\lambda_K$. Therefore, for a given $\Rscr=(\R_0,\R_1,\R_2)$:
	\begin{itemize}
		\item[(i)] When $M \in \Big[0,  \gamma_{K} \Big)$,
		\begin{align}
			\Rach(M,\Rscr) - \RGWRlb \leq 
			\; \R_0 + \R_1+\R_2-2M  - \Big(\R_0 + \R_1+\R_2 -2M\Big) = 0.  \notag
		\end{align}
		
		\item[(ii)]  When $M\in \Big[ \gamma_{K} , \,\gamma \Big)$,
		\begin{align}
			\Rach(M,\Rscr) - \RGWRlb  &\leq 
			\R_0 + \R_1+\R_2 - \frac{1}{K}\min\{\R_1,\R_2\}   -M 
			-  \Big(\R_0 + \R_1+\R_2 -2M\Big)\notag\\
			&  = M- \frac{1}{K}\min\{\R_1,\R_2\}   \notag\\
			& \stackrel{(a)}{\leq} \Big(\frac{1}{2}- \frac{1}{K}\Big) \min\{\R_1,\R_2\} ,
			\label{eq: K peak gap 1}
		\end{align}
		where $(a)$ follows from the fact that $M\leq\gamma= \frac{1}{2}  \min\{\R_1,\R_2\}$.
		
		\item[(iii)]   When $M\in \Big[ \gamma , \,\lambda \Big)$,
		\begin{align}
			\Rach(M,&\Rscr) - \RGWRlb  \notag\\
			&\leq 
			\R_0 + \R_1+\R_2 - \frac{1}{K}\min\{\R_1,\R_2\} -M 
			-  \bigg(\R_0 + \frac{1}{2}\Big( \R_1+\R_2 + \max\{\R_1,\R_2\}\Big) - M \bigg)\notag\\
			&  = \Big(\frac{1}{2}- \frac{1}{K}\Big) \min\{\R_1,\R_2\}. \label{eq: K peak gap 2}
		\end{align}
		
		\item[(iv)]  When $M\in \Big[ \lambda , \,\lambda_{K} \Big)$,
		\begin{align}
			\Rach(M,\Rscr) - \RGWRlb  &\leq 
			\R_0 + \R_1+\R_2 - \frac{1}{K}\min\{\R_1,\R_2\} -M  - \frac{1}{2}\Big(\R_0+\R_1+\R_2-M\Big)  \notag\\
			& =  \frac{1}{2}\Big(\R_0+\R_1+\R_2-M\Big) - \frac{1}{K}\min\{\R_1,\R_2\} \notag\\
			& \stackrel{(b)}{\leq} \Big(\frac{1}{2}- \frac{1}{K}\Big) \min\{\R_1,\R_2\},
			\label{eq: K peak gap 3}
		\end{align}
		where $(b)$ follows since  $M \geq \lambda=  \R_0+  \max\{\R_1,\R_2\}$.
		
		\item[(v)] When $M\in \Big[\lambda_K ,\, \R_0+\R_1+\R_2\Big]$, 
		\begin{align}
			\Rach(M,\Rscr) &- \RGWRlb \leq \frac{1}{2}(\R_0+\R_1+\R_2-M)   -  \frac{1}{2}\Big(\R_0+\R_1+\R_2-M\Big)  = 0.\notag
		\end{align}
	\end{itemize} 	
	It is observed that for any $\Rscr$, when
	\begin{align} 
		M &\in \Big[ 0,\, \frac{1}{K}\min\{\R_1,\R_2\} \Big) \bigcup
		\Big[\R_0+\R_1+\R_2- \frac{2}{K}\min\{\R_1,\R_2\} ,\,\R_0+\R_1+\R_2\Big], \notag
	\end{align} 
	we have $\Rach(M,\Rscr) = \RGWstarR =\RGWRlb$, and in the remaining memory region
	\begin{align} 
		\Rach(M,\Rscr) - \RGWstarR  \leq	\Rach(M,\Rscr) - \RGWRlb  \leq \Big(\frac{1}{2}- \frac{1}{K}\Big) \min\{\R_1,\R_2\}   . \notag
	\end{align} 
	
	

	\section{Upper Bound on $\RX (M,\R_1,\R_2)$} \label{app:CACM private peak}	
	In this appendix, we describe in detail the scheme adopted for the private sublibrary $L_1$ when considering the peak rate criterion, which is used to evaluate the performance of the proposed MR scheme for two files in Sec.~\ref{subsec: scheme}. The scheme adopted for sublibrary $L_1$  is based on memory-sharing among generalizations of the correlation-unaware CACM schemes proposed in  \cite{tian2016caching} and  \cite{yu2016exact} to files with unequal lengths. 
	The generalization is done by dividing the larger file into two parts such that one part is equal in size to the smaller file. Then, this part and the smaller file are cached and delivered according to a CACM scheme available in the literature designed for equal-length files, while the remaining part of the larger file is stored in the caches only for large enough cache capacities. We note that since, for a setting with independent files, each of the schemes proposed in \cite{tian2016caching}  and \cite{yu2016exact} is optimal (or close to optimal) for a different regime of the memory, for the generalization we use a combination of both schemes. Specifically, for small cache sizes we use a generalization of the scheme in \cite{tian2016caching}, and for large cache sizes we use a generalization of the scheme proposed in \cite{yu2016exact}. 
	The resulting scheme
	achieves a peak rate, $\RX (M,\R_1,\R_2)$,  whose upper bound is given in the following theorem. 
	\begin{theorem}\label{thm:CACM private peak}
		In the two-file $K$-receiver network, for a given cache capacity $M$ and  files with rates $\R_1$ and $\R_2$,  an  upper bound on $\RX (M,\R_1,\R_2)$ is given by
		\begin{equation}
			\RX (M,\R_1,\R_2) \leq
			\begin{cases}
				\R_1+\R_2  -2 M ,                   &   M\in \Big[0, \,\frac{1}{K}\min\{\R_1,\R_2\}\Big)  \\
				\R_1+\R_2 -   \frac{1}{K}\min\{\R_1,\R_2\} - M  ,				      &    M \in \Big[\frac{1}{K}\min\{\R_1,\R_2\}, \, \R_1+\R_2 -   \frac{2}{K}\min\{\R_1,\R_2\}  \Big)   \\
				\frac{1}{2} (\R_1 + \R_2 -  M) ,	    & M \in \Big[\R_1+\R_2 -   \frac{2}{K}\min\{\R_1,\R_2\}    , \,\R_1+\R_2\Big] \label{eq:K user UB}.
			\end{cases}
		\end{equation}	
	\end{theorem}	
	
	\begin{proof}
		For the purposes of our analysis we only provide the caching and delivery strategies for cache sizes 
		\begin{align}
			M\in\Big\{0,\, \frac{1}{K}\min\{\R_1,\R_2\},\, \R_1+\R_2- \frac{1}{K}\min\{\R_1,\R_2\} ,\, \R_1+\R_2\Big\}, 
		\end{align}
		and upper bound the peak rate by memory-sharing among the corresponding achievable rates.  For a setting with two files and $K\geq2$ receivers, the scheme operates as follows: 
		\begin{itemize}
			\item[$\bullet$] $M = 0$: In the worst case, i.e., when at least two of the receives request different files, the sender multicasts both files over the shared link with a total rate of $\R_1+\R_2$.
			
			\item[$\bullet$] $M = \frac{1}{K}\min\{\R_1,\R_2\}$:  File $\Wsf_i$, with $i \in\{1,2\}$, is divided into $2K+1$ packets.
			\begin{itemize}
				\item[$\circ$]   Packets $\{\Wsf_i^{(1)},\, \dots,\Wsf_i^{(2K)}\} $ with size $\frac{1}{2K}\min\{\R_1,\R_2\}$.
				\item[$\circ$]   Packet $\Wsf_i^{(2K+1)}$ with size $\R_i -  \min\{\R_1,\R_2\}$.
			\end{itemize}
			As in  \cite{tian2016caching} and \cite{chen2014fundamental}, receiver $r_k$, $k\in\{1,\dots,K \}$, fills its cache as 
			\begin{align}
				Z_{r_k}=\{ \Wsf_1^{(2k-1)}\oplus \Wsf_2^{(2k-1)} ,\,   \Wsf_1^{(2k)}\oplus \Wsf_2^{(2k)}  \} .\label{eq:cache peak 1}
			\end{align}
			The multicast codeword results from concatenating: $i)$ a codeword designed based on the delivery strategy in \cite[Sec.~V.B]{chen2014fundamental}, with rate $2(1-\frac{1}{K})\min\{\R_1,\R_2\}$, and $ii)$  packets $\Wsf_1^{(2K+1)}$ and $\Wsf_2^{(2K+1)}$, with rate $ \R_1+\R_2-2\min\{\R_1,\R_2\}$. The peak delivery rate is  
			\begin{align}
				\RX (M,\R_1,\R_2) &= 2(1-\frac{1}{K})\min\{\R_1,\R_2\}+ \R_1+\R_2-2\min\{\R_1,\R_2\} \notag\\
				&=  \R_1+\R_2-\frac{2}{K}\min\{\R_1,\R_2\}. 
			\end{align}

			\item[$\bullet$] $M =  \R_1+\R_2-\frac{2}{K}\min\{\R_1,\R_2\}$: File $\Wsf_i$, with $i \in\{1,2\}$, is divided into $K+1$ packets.
			\begin{itemize}
				\item[$\circ$] Packets  $\{\Wsf_i^{(1)},\,\dots,\,\Wsf_i^{(K)}\}$ with size $\frac{1}{K}\min\{\R_1,\R_2\}$.
				\item[$\circ$] Packet $\Wsf_i^{(K+1)}$ with size $\R_i -  \min\{\R_1,\R_2\}$.
			\end{itemize}
			Similar to the strategy in  \cite{yu2016exact}, receiver $r_k$'s cache is filled as 
			\begin{align}
				Z_{r_k}=\Big\{ \Wsf_1^{(j)},\, \Wsf_2^{(j)}: \; j =1,\dots,k-1,k+1,\dots,K+1   \Big\} .\label{eq:cache peak 2}
			\end{align}
			For any demand realization, when receivers request the same file or different files, the sender transmits a coded message designed based on the delivery strategy in \cite[Sec.~IV.B]{yu2016exact}, with peak rate 
			\begin{align}
				\RX (M,\R_1,\R_2) = \frac{1}{K}\min\{\R_1,\R_2\}.  
			\end{align}
			
			\item[$\bullet$] $M = \R_1+\R_2$: The files are fully stored at all receivers resulting in zero delivery rate. 
			
		\end{itemize}	
		We note that for other cache sizes, the schemes in \cite{tian2016caching} and \cite{yu2016exact} can be similarly generalized to files with unequal lengths but the corresponding caching and delivery strategies are not provided here. As in \cite{maddah14fundamental}, through memory-sharing the lower convex envelope of the memory-rate pairs 
		\begin{align}
			\Big(M, \RX(M,\R_1,\R_2)\Big) \in  \Bigg\{ &\Big(0,\,    \R_1+\R_2  \Big);\,  \Big(  \frac{1}{K}\min\{\R_1,\R_2\}  ,\,  \R_1+\R_2-\frac{2}{K}\min\{\R_1,\R_2\}    \Big);\,\notag\\
			&\Big(    \R_1+\R_2-\frac{2}{K}\min\{\R_1,\R_2\}  ,\, \frac{1}{K}\min\{\R_1,\R_2\}   \Big) ;  \Big(\R_1+\R_2 ,\,0 \Big)    \Bigg\},\notag
		\end{align}
		is achievable, resulting in the peak rate given in \eqref{eq:K user UB}. 
		
	\end{proof} 
	
	\begin{remark} 
		The generalized correlation-unaware scheme presented above, when particularized to $K=2$ receivers, is shown to be optimal in \cite{ITLowerBound} and \cite{li2017rate}, and outperforms schemes such as \cite{zhang2015differentsize} and \cite{cheng2017optimal} designed for files with different lengths. Furthermore, when particularized to files with equal lengths, i.e., $\R_1=\R_2$, this scheme coincides with the optimal scheme originally characterized in  \cite[Appendix A]{maddah14fundamental} for equal-length and independent files.
		
	\end{remark}
	\section{Proof of Theorem \ref{thm:ach and GW lower bound avg}}\label{app:ach and GW lower bound avg}
	
	The proof follows from comparing $\RachE(M, \Rscr)$, the average rate achieved by the proposed MR scheme given in Theorem~\ref{thm:avg rate MR}, with $\RGWRlbE$, the lower bound on the MR average rate-memory function given in Theorem~\ref{thm:LowerBoundGW}, over different regions of the memory. 
	In this appendix, in order to simplify the analysis when comparing $\RachE(M, \Rscr)$ with $\RGWRlbE$, we use a less tight lower bound by only considering the following three inequalities from Theorem \ref{thm:LowerBoundGW}
	\begin{equation}
		\RGWRlbE  \geq
		\begin{cases}
			\R_0 + \frac{3}{4}(\R_1+\R_2 )- M ,        
			& \; M \in \Big[0,\, \bar\gamma\Big)  \\
			\frac{3}{4}\R_0 + \frac{1}{2}(\R_1+\R_2 )+\frac{1}{4}\max\{\R_1,\R_2 \}-\frac{3}{4} M ,				   
			& \;     M \in \Big[\bar\gamma , \, \lambda\Big) \\
			\frac{1}{2}\Big(\R_0+\R_1+\R_2 - M\Big)   ,				     
			& \;  M \in \Big[ \lambda, \, \R_0+\R_1+\R_2\Big]    . \label{eq: LB GW avg appendix}
		\end{cases} 	
	\end{equation}	  
	with $\lambda$ defined in \eqref{eq: gamma lambda}, and 
	\begin{align}
		&	  \bar\gamma \triangleq  \R_0+\min\{\R_1,\R_2\}  .   \label{eq: gamma lambda bar} 
	\end{align}
	For $\gamma_K$ and $\lambda_K$ defined in \eqref{eq: gamma lambda K}, and for $\bar\gamma_K$ defined in \eqref{eq: gamma bar K}, when $K\geq2$ we have $ 2\gamma_K\leq \bar\gamma_K \leq \bar\gamma\leq  \lambda\leq\lambda_K $. Therefore, for a given $\Rscr = (\R_0,\R_1,\R_2)$:
	
	\begin{itemize}
		\item[(i)] When $M \in \Big[0, \, 2\gamma_K  \Big)$,  	
		\begin{align}
			\RachE(M,\Rscr)  -  \RGWRlbE  &\leq 
			\R_0+\Big(1-\frac{1}{2^K}\Big)  (\R_1+\R_2)-\Big(\frac{3}{2}-\frac{2}{2^K}  \Big)M  -\Big(\R_0 + \frac{3}{4}(\R_1+\R_2 )- M\Big)   \notag\\
			& = \Big(\frac{1}{4}-\frac{1}{2^K}\Big)  (\R_1+\R_2-2M)    \notag\\
			& \stackrel{(a)}{\leq}   \Big(\frac{1}{4}-\frac{1}{2^K}\Big)  (\R_1+\R_2 ) , \label{eq: K avg gap 1}
		\end{align}
		where $(a)$ follows since $M\geq 0$.
		
		\item[(ii)] When $M \in \Big[2\gamma_K,\,\bar\gamma_K  \Big)$,  
		\begin{align}
			\RachE(M,\Rscr)  -  \RGWRlbE   &\leq
			\R_0+\Big(  1-\frac{1}{2^K}  \Big) ( \R_1+\R_2) -   \Big( 1-\frac{4}{2^K}  \Big)   \gamma_K -M \notag\\
			& \quad  -\Big(\R_0 + \frac{3}{4}(\R_1+\R_2 )- M\Big)   \notag\\
			& =  \Big(\frac{1}{4}- \frac{1}{2^K} \Big) \Big(\R_1+\R_2 -4\gamma_K\Big) \notag\\
			& \leq \Big(\frac{1}{4}- \frac{1}{2^K} \Big) \Big(\R_1+\R_2 \Big)  .   
		\end{align}
		
		\item[(iii)] When $M \in \Big[ \bar\gamma_K,\,\bar\gamma \Big)$,  
		\begin{align}
			\RachE(M,\Rscr)  -  &\RGWRlbE    \leq
			\Big(  1-\frac{1}{2^K}  \Big)\Big( \R_0+\R_1+\R_2-M \Big)- \Big(  1-\frac{2}{2^K}  \Big) \gamma_K \notag\\  
			& \quad  -\Big(\R_0 + \frac{3}{4}(\R_1+\R_2 )- M\Big)   \notag\\
			& = \Big(  \frac{1}{4}-\frac{1}{2^K}  \Big) ( \R_1+\R_2) -\frac{1}{2^K} (\R_0-M)   -\Big(  1-\frac{2}{2^K}  \Big) \gamma_K \notag\\
			& \stackrel{(b)}{\leq}  \Big(  \frac{1}{4}-\frac{1}{2^K}  \Big) ( \R_1+\R_2) -\Big(1-\frac{K+2}{2^K} \Big) \gamma_K \notag\\
			&\leq    \Big(  \frac{1}{4}-\frac{1}{2^K}  \Big) ( \R_1+\R_2)    ,   
		\end{align}
		where $(b)$ follows due to the fact that $M< \bar\gamma=\R_0+\min\{\R_1,\R_2\}=\R_0 +K\gamma_K$.

		\item[(iv)] When $M \in \Big[\bar\gamma,\,\lambda \Big)$,  
		\begin{align}
			\RachE(M,\Rscr)  & -  \RGWRlbE  \leq 
			\Big(  1-\frac{1}{2^K}  \Big)\Big( \R_0+\R_1+\R_2-M \Big)- \Big(  1-\frac{2}{2^K}   \Big)\gamma_K \notag\\
			& \quad-   \Big(\frac{3}{4}\R_0 + \frac{1}{2}(\R_1+\R_2 )+\frac{1}{4}\max\{\R_1,\R_2 \}-\frac{3}{4} M\Big) \notag\\
			& =   \Big(  \frac{1}{4}-\frac{1}{2^K}  \Big) \Big(\R_0+\max\{\R_1,\R_2\} -M\Big)   + \Big(   \frac{1}{2} -\frac{1}{2^K}  \Big)  \Big(\min\{\R_1,\R_2\}   -2\gamma_K \Big) \notag  \\
			& \stackrel{(c)}{\leq}   \Big(  \frac{1}{4}-\frac{1}{2^K}  \Big) \max\{\R_1,\R_2\} +     \frac{1}{4}\min\{\R_1,\R_2\}   -\Big(  1-\frac{2}{2^K}   \Big)\gamma_K   \notag\\ 
			&  \stackrel{(d)}{\leq}     \Big(  \frac{1}{4}-\frac{1}{2^K}  \Big) \max\{\R_1,\R_2\} +     \frac{1}{4}\min\{\R_1,\R_2\} -\frac{1}{2^K}  \min\{\R_1,\R_2\}         \notag  \\
			&= \Big(  \frac{1}{4}-\frac{1}{2^K}  \Big) (\R_1+\R_2)    ,
		\end{align}
		where $(c)$ follows due to the fact that  $M\geq\bar\gamma=\R_0+\min\{\R_1,\R_2\} $, and $(d)$ follows since $\frac{1}{K}\Big(  1-\frac{2}{2^K}   \Big)\geq \frac{1}{2^K}$ for $K\geq 2$.
		
		\item[(v)] When $M \in \Big[\lambda,\,\lambda_K \Big)$,  
		\begin{align}
			\RachE(M,\Rscr)  -  \RGWRlbE  \leq &
			\Big(  1-\frac{1}{2^K}  \Big)\Big( \R_0+\R_1+\R_2-M \Big)- \Big(  1-\frac{2}{2^K}  \Big)\gamma_K \notag\\
			& \quad-  \frac{1}{2}\Big(\R_0+\R_1+\R_2 - M\Big) \notag\\
			&=   \Big(  \frac{1}{2}-\frac{1}{2^K}  \Big)\Big( \R_0+\R_1+\R_2-2\gamma_K  -M \Big)    \notag\\
			&\stackrel{(e)}{\leq} \Big(  \frac{1}{2}-\frac{1}{2^K}  \Big)  (K-2)\gamma_K \notag\\
			& \stackrel{(f)}{\leq} \Big(  \frac{1}{4}-\frac{1}{2^{K}}  \Big) (\R_1+\R_2)  , \label{eq: K avg gap 3}
		\end{align}
		where $(e)$ follows due to the fact that  $M\geq\lambda=\R_0+\max\{\R_1,\R_2\} $, and $(f)$ follows since $ \min\{\R_1,\R_2\}\leq (\R_1+\R_2)/2$, and since $(  \frac{1}{2}-\frac{1}{2^K}  )( \frac{1}{2}-\frac{1}{K}) \leq \frac{1}{4}-\frac{1}{2^{K}} $ for   $K\geq 2$.
		
		\item[(vi)] For $M \in \Big[\lambda_K  ,\; \R_0+\R_1+\R_2\Big]$, 
		\begin{align}
			\RachE(M,\Rscr)  &-  \RGWRlbE= \frac{1}{2}\Big(\R_0+\R_1+\R_2-M\Big)   -  \frac{1}{2}\Big(\R_0 +\R_1 + \R_2 - M\Big)  = 0.\notag
		\end{align}
		
	\end{itemize}
	
	It is observed that for any $\Rscr$, when 
	\begin{align} 
		M &\in \Big[\R_0+\R_1+\R_2- \frac{2}{K}  \min\{\R_1,\R_2\}  , \R_0+\R_1+\R_2\Big] , \notag
	\end{align} 
	we have $\RachE(M,\Rscr) = \RGWRlbE = \RGWstarRE$, and for all other cache capacities
	\begin{align} 
		\RachE(M,\Rscr) - \RGWstarRE  \leq \RachE(M,\Rscr) - \RGWRlbE  \leq \Big(  \frac{1}{4}-\frac{1}{2^{K}}  \Big) (\R_1+\R_2)  . \notag
	\end{align}

	\section{Upper Bound on $\REX (M,\R_1,\R_2)$} \label{app:CACM private avg}
	In this appendix, we describe in detail the scheme adopted for the private sublibrary $L_1$ under the average rate criterion, which is used to evaluate the performance of the proposed MR scheme for two files in Sec.~\ref{subsec: scheme}.  The scheme adopted for sublibrary $L_1$ is a generalization of the correlation-unaware CACM scheme proposed in \cite{yu2016exact} to files with unequal lengths, where as for the peak rate criterion, generalization is done by dividing the larger file into two parts such that one part is equal in length to the smaller file.  
	The resulting scheme achieves an average rate, $\REX (M,\R_1,\R_2)$,  whose upper bound is given in the following theorem. 
	
	\begin{theorem}\label{thm:CACM private avg}
		In the two-file $K$-receiver network, for a given cache capacity $M$ and  files with rates $\R_1$ and $\R_2$,  an  upper bound on $\REX (M,\R_1,\R_2)$ is given by
		\begin{equation}
			\REX(M,\R_1,\R_2)\leq \begin{cases}
				\Big(1-\frac{1}{2^K}\Big)  (\R_1+\R_2)-\Big(\frac{3}{2}-\frac{2}{2^K}  \Big)M  ,        & M \in \Big[0, \;2\,\gamma_K\Big)  \\
				\Big(  1-\frac{1}{2^K}  \Big) \Big( \R_1+\R_2 -M \Big) - 2\Big(  \frac{1}{2}-\frac{1}{2^K}  \Big) \gamma_K,				      &  M \in \Big[2\,\gamma_K, \; \R_1+\R_2 -   2\,\gamma_K \Big)  \\
				\frac{1}{2} (\R_1 + \R_2 -  M) ,	    &M \in \Big[\R_1+\R_2 -   2\,\gamma_K   , \;\R_1+\R_2\Big]\label{eq:K user UB avg}
			\end{cases}
		\end{equation}
		for $\gamma_K = \frac{1}{K}\min\{\R_1,\R_2 \} $  defined in \eqref{eq: gamma lambda K}.
	\end{theorem}	
	
	\begin{proof}
		For the purposes of our analysis we only provide the caching and delivery strategies for cache sizes
		\begin{align}
			M\in\Big\{0,\, \frac{2}{K}\min\{\R_1,\R_2\},\, \R_1+\R_2- \frac{1}{K}\min\{\R_1,\R_2\} ,\, \R_1+\R_2\Big\}. 
		\end{align}
		and upper bound the average rate by memory-sharing among the corresponding achievable rates. For a setting with two files and $K\geq2$ receivers, the scheme operates as follows:
		\begin{itemize}
			\item[$\bullet$] $M = 0$: In this case, if at least two of the receivers request different files the sender multicasts both files over the shared link, and if all receivers request the same file only one of the files is transmitted, resulting in an average rate equal to
			\begin{align}
				\REX (M,\R_1,\R_2) &= \Big(1-\frac{1}{2^K}\Big) (\R_1+\R_2)+ \frac{1}{2^K}\,\R_1+\frac{1}{2^K} \,\R_2 
				=  \Big(1-\frac{1}{2^K}\Big)  (\R_1+\R_2)  . 
			\end{align}
			\item[$\bullet$] $M =  \frac{2}{K}\min\{\R_1,\R_2\}$: File $\Wsf_i$, with $i \in\{1,2\}$, is divided into $K+1$ packets.
			\begin{itemize}
				\item   Packets $\{\Wsf_i^{(1)},\, \dots,\Wsf_i^{(K)}\} $ with size $\frac{1}{K}\min\{\R_1,\R_2\}$.
				\item   Packet $\Wsf_i^{(K+1)}$ with size $\R_i -  \min\{\R_1,\R_2\}$.
			\end{itemize}
			As in \cite{yu2016exact}, receiver $r_k$, $k\in\{1,\dots,K \}$, fills its cache as 
			$Z_{r_k}=\{ \Wsf_1^{(k)},\; \Wsf_2^{(k)}  \} .$
			When at least two of the receivers request different files, the multicast codeword results from concatenating: $i)$ a codeword designed based on the delivery strategy in \cite[Sec.~IV.B]{yu2016exact}, with rate $(2-\frac{3}{K})\min\{\R_1,\R_2\}$, and $ii)$  packets $\Wsf_1^{(K+1)}$ and $\Wsf_2^{(K+1)}$, with rate $ \R_1+\R_2-2\min\{\R_1,\R_2\}$. When all receivers request file $W_i$, the multicast codeword consists of: $i)$ a codeword designed as described in \cite[Sec.~IV.B]{yu2016exact}, with rate $(1-\frac{1}{K})\min\{\R_1,\R_2\}$, and $ii)$  packet $\Wsf_i^{(K+1)}$  with rate $ \R_i-\min\{\R_1,\R_2\}$. 
			The average delivery rate is  
			\begin{align}
				\REX (M,\R_1,\R_2) &= \frac{2}{2^K}\Big( 1-\frac{1}{K} \Big) \min\{\R_1,\R_2\}   + \Big(1-\frac{2}{2^K}\Big)\Big(2-\frac{3}{K} \Big)    \min\{\R_1,\R_2\}   \notag\\ &\qquad +\Big(1-\frac{1}{2^K}\Big)  \Big(   \R_1+\R_2-2\min\{\R_1,\R_2\}\Big)\notag\\
				&=  \Big(1-\frac{1}{2^K}\Big)  (\R_1+\R_2) \,-\, \Big(\frac{3-2^{(2-K)}}{K} \Big)\min\{\R_1,\R_2\} .
			\end{align}

			\item[$\bullet$] $M =  \R_1+\R_2-\frac{2}{K}\min\{\R_1,\R_2\}$: In this case, the cache configuration is the same as that for the peak rate given in \eqref{eq:cache peak 2} in Appendix \ref{app:ach and GW lower bound peak}-I, and for any demand realization $\dbf$, the delivery strategy in  \cite[Sec.~IV.B]{yu2016exact} achieves a rate equal to $\frac{1}{K}\min\{\R_1,\R_2\}$, resulting in an average  rate of
			\begin{align}
				\REX (M,\R_1,\R_2) = \frac{1}{K}\min\{\R_1,\R_2\}.   
			\end{align}
			\item[$\bullet$] $M = \R_1+\R_2$: The files are fully stored at all receivers resulting in zero delivery rate for any demand realization $\dbf$. 
			
		\end{itemize}	
		We note that for other cache sizes, the scheme in \cite{yu2016exact} can be similarly generalized to files with unequal lengths but the corresponding caching and delivery strategies are not provided here. Through memory-sharing the lower convex envelope of the memory-rate pairs 
		\begin{align}
			\Big(M, \REX(M,\R_1,\R_2)\Big) \in  \Bigg\{ &\Big(0,\,  \Big(1-\frac{1}{2^K}\Big) (\R_1+\R_2) \Big);\, 	 \notag\\
			& 
			\Big(\frac{2}{K} \min\{\R_1,\R_2\} ,\,    \Big(1-\frac{1}{2^K}\Big)  (\R_1+\R_2) - \Big(\frac{3-2^{(2-K)}}{K} \Big)\min\{\R_1,\R_2\}     \Big);\,	\notag\\
			&\Big(    \R_1+\R_2-\frac{2}{K}\min\{\R_1,\R_2\}  ,\, \frac{1}{K}\min\{\R_1,\R_2\}   \Big) ;\; \;  \Big(\R_1+\R_2 ,\,0 \Big)    \Bigg\},\notag
		\end{align}
		is achievable, resulting in the average rate given in \eqref{eq:K user UB avg}.

	\end{proof}
	\begin{remark} 
		The generalized correlation-unaware scheme presented above, when particularized to $K=2$ receivers, is shown to be optimal in \cite{ITLowerBound}, and outperforms schemes in the literature that are designed for files with different lengths. Furthermore, when particularized to files with equal lengths, i.e., $\R_1=\R_2$, this scheme coincides with the optimal scheme  characterized in  \cite{yu2016exact} for equal-length and independent files.
		

	\end{remark}

	\section{Proof of Theorem \ref{thm:LowerBound}}\label{app:LowerBound}
	We derive the lower bounds on the optimal peak and average rate-memory functions, $\Rstar$ and $\RstarE$, given in Theorem~\ref{thm:LowerBound}  using a result from \cite{ITLowerBound}, which is stated in the following theorem for a network with $K$ receivers $\{r_1,\dots,r_K\}$, each with cache capacity $M$, and a 
	library composed of $N$ files $\{X_1^F,\dots,X_N^F\}$ generated by $p(x_1,\dots,x_N)$.  
	\begin{theorem}\label{thm: general LB rule}
		Consider $\nu$ consecutive demands $\dbf^{(1)},\dots,\dbf^{(\nu)}$, and $\nu$ receiver subsets ${\mathcal S}_{1},\dots,{\mathcal S}_{\nu}\subseteq \{r_1,\dots,r_K\}$. The optimal sum rate required to deliver the $\nu$ demands is lower bounded by
		\begin{align}
			&
			\sum\limits_{i=1}^{\nu} R_{\dbf^{(i)}}^* \geq
			\sum\limits_{i=1}^\nu     H\Big( \Big\{ \Xsf_d: d\in \mathcal D_{{\mathcal S}_i}^{(i)}\Big\}   \Big|  	{\mathfrak X}_1,\dots,{\mathfrak X}_{i-1}   \Big)  
			-  |\, \widetilde {\mathcal S}  \,|  M  \notag\\
			& \qquad +     H\Big(\Big\{\Xsf_d: d\in  {\mathcal D}_{ \widetilde{\mathcal S}}^{(1)}\cup\dots\cup  {\mathcal D}_{ \widetilde{\mathcal S}}^{(\nu)}  \Big\}	\Big)   -    H\Big(\Big\{ \Xsf_d: d\in   {\mathcal D}_{{\mathcal S}_1}^{(1)}\cup\dots\cup  {\mathcal D}_{{\mathcal S}_\nu}^{(\nu)} \Big \} \Big)  
			,  
		\end{align}
		where $R_{\dbf^{(i)}}^*$ denotes the optimal delivery rate of demand $\dbf^{(i)}$,  ${\mathcal D}_{{\mathcal S}}^{(i)}  \triangleq \{d_{r_k}^{(i)}   :\, {r_k}\in  {\mathcal S}   \}$
		denotes the set of file indices requested by the receivers in set ${\mathcal S}$ in demand $\dbf^{(i)}$, and
		\begin{align}
			&\widetilde {\mathcal S} \triangleq  \{\mathcal S_1,\dots,\mathcal S_\nu\}, \\
			&{\mathfrak X}_\ell\triangleq   \bigcup\limits_{ j =1 }^{\ell-1}\{ \Xsf_d  : \,  d \in     {\mathcal D}_{{\mathcal S}_\ell \cap {\mathcal S}_j}^{(j)}  \}, \quad \ell\in \{1,\dots, \nu -1\}.
		\end{align}
	\end{theorem}

	Since the number of distinct requested files in the two-file $K$-receiver network is at most two, using Theorem \ref{thm: general LB rule}, in the following we lower bound $\Rstar$ and $\RstarE$ with those of the two-file two-receiver network. Without loss of generality we consider receivers $r_1$ and $r_2$ for which the cache-demand-augmented graph is depicted in Fig. \ref{fig:AugmentedGraph}.  
	\begin{figure}[h!]
		\centering
		\includegraphics[width=4.5in]{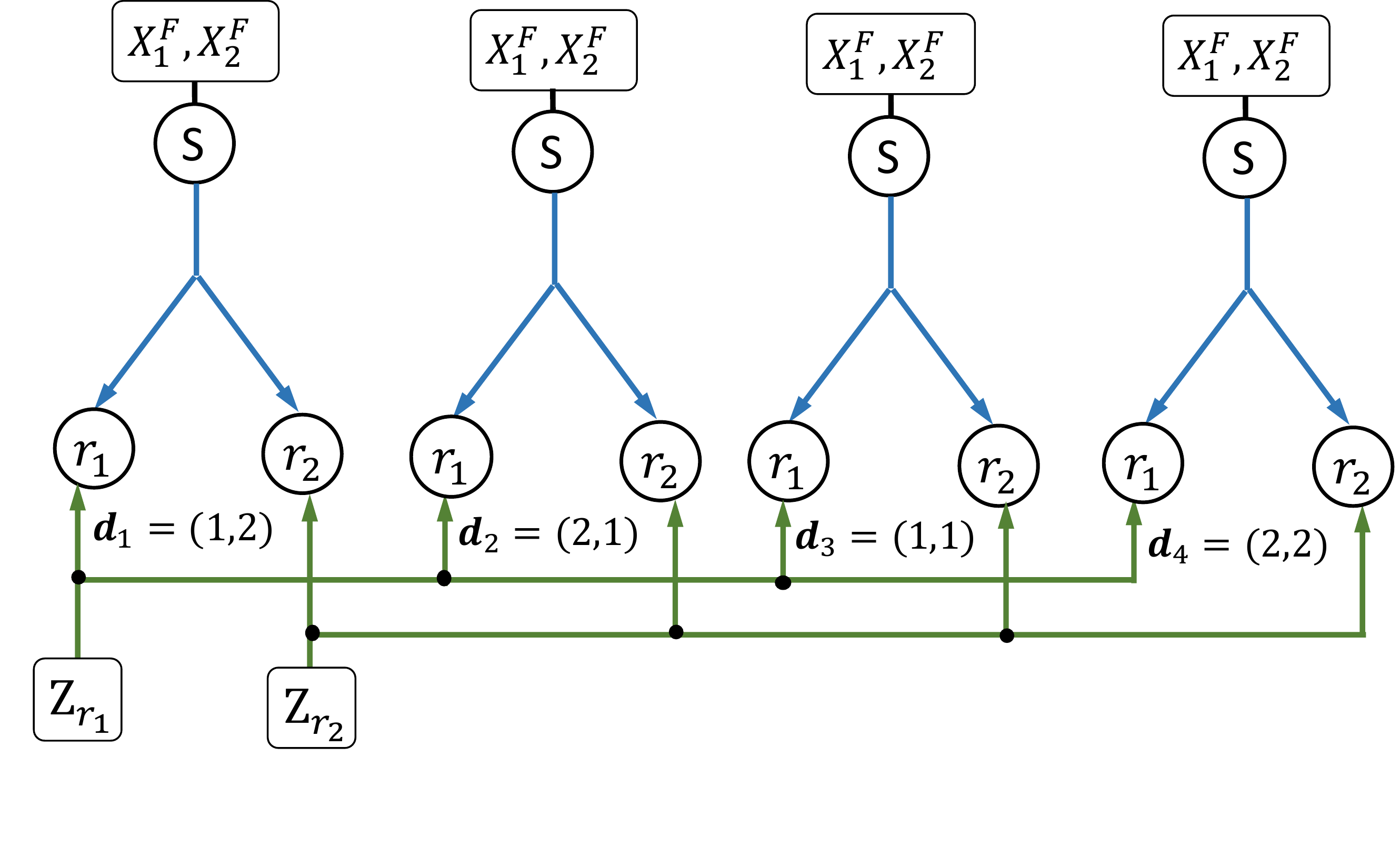}
		\caption{The cache-demand-augmented graph as described in \cite{llorca2013network} is constructed by replicating the caching-augmented graph for  different demands, while sharing the cache edges among all demands.}
		\label{fig:AugmentedGraph}
	\end{figure}
	
	\subsection*{I) \bf Lower Bound on $\mathbf \Rstar$:}  
	
	A lower bound  on the optimal peak rate-memory function, $\Rstar$, is obtained by applying Theorem~\ref{thm: general LB rule} to different sets of $\nu$  consecutive demand realizations $\dbf^{(1)},\dots,\dbf^{(\nu)}$, each taking values among the two  worst-case demands $\dbf_1=(1,2)$ and $\dbf_2=(2,1)$. Then, since $\Rstar \geq \frac{1}{\nu} \sum\limits_{i=1}^{\nu} R_{\dbf^{(i)}}^* $, a lower bound on  $\Rstar$ is obtained from normalizing the lower bound on the optimal sum rate $\sum\limits_{i=1}^{\nu} R_{\dbf^{(i)}}^* $ given by Theorem~\ref{thm: general LB rule}. Specifically, 
	
	\begin{itemize}
		\item Case $(i)$: For $\nu=1$ consecutive demand $\dbf^{(1)} ={\dbf}_1= (1,2)$ with receiver subset $\mathcal S_1 = \{r_1,r_2\}$,
		based on  Theorem~\ref{thm: general LB rule}, we have $\mathcal D_{{\mathcal S}_1}^{(1)}=\{1,2\}$ and $\widetilde{\mathcal S} =\mathcal S_1 = \{r_1,r_2\}$. Therefore,
		\begin{align}
			\Rstar	 
			& \geq R_{\dbf_1}^*\notag\\
			& \geq     H\Big( \Big\{ \Xsf_d: d\in \mathcal D_{{\mathcal S}_1}^{(1)}\Big\}   \Big)  
			-  |\, \widetilde {\mathcal S}  \,|  M  
			+     H\Big(\Big\{\Xsf_d: d\in  {\mathcal D}_{ \widetilde{\mathcal S}}^{(1)}  \Big\}	\Big)   -    H\Big(\Big\{ \Xsf_d: d\in   {\mathcal D}_{{\mathcal S}_1}^{(1)} \Big \} \Big)  \notag\\
			& = H(\Xsf_1, \Xsf_2) \,- \,2 \,M . \label{eq:LB1}
		\end{align}

		\item Case $(ii)$: For $\nu=2$ consecutive demands $\dbf^{(1)} = \dbf_1= (1,2)$ and $\dbf^{(2)} =\dbf_2= (2,1)$, with  corresponding receiver subsets $\mathcal S_1 = \mathcal S_2 = \{r_1\}$, based on Theorem \ref{thm: general LB rule}, we have $\mathcal D_{{\mathcal S}_1}^{(1)}=\{1\}$, $\mathcal D_{{\mathcal S}_2}^{(2)}=\{2\}$, $\widetilde{\mathcal S} = \{r_1\}$, and ${\mathfrak X}_1 = \{\Xsf_1\}$. Therefore,

		\begin{align}
			\Rstar	 
			& \geq \frac{1}{2}(R_{\dbf_1}^*+R_{\dbf_2}^* )\notag\\	 
			& \geq 
			\frac{1}{2} \;\bigg(    H\Big( \Big\{ \Xsf_d: d\in \mathcal D_{{\mathcal S}_1}^{(1)}\Big\}   \Big)  +  H\Big( \Big\{ \Xsf_d: d\in \mathcal D_{{\mathcal S}_2}^{(2)}\Big\}   \Big|  	{\mathfrak X}_1   \Big)  
			-  |\, \widetilde {\mathcal S}  \,|  M  \notag\\
			& \qquad +     H\Big(\Big\{\Xsf_d: d\in  {\mathcal D}_{ \widetilde{\mathcal S}}^{(1)}\cup {\mathcal D}_{ \widetilde{\mathcal S}}^{(2)}  \Big\}	\Big)   -    H\Big(\Big\{ \Xsf_d: d\in   {\mathcal D}_{{\mathcal S}_1}^{(1)}\cup   {\mathcal D}_{{\mathcal S}_2}^{(2)} \Big \} \Big) 
			\; \bigg)   \notag\\
			& = \frac{1}{2} \;\Big(   H (  \Xsf_1  )  +  H(   \Xsf_2 |  	\Xsf_1  )  
			\,-\,    M  \,+\,     H (\Xsf_1,\Xsf_2)   -    H (\Xsf_1,	\Xsf_2)  \Big)   \notag\\
			& = \frac{1}{2} \;\Big(      H(   \Xsf_1, 	\Xsf_2  )  	\,-\,    M   \Big).  \label{eq:LB2}
		\end{align}

		\item Case $(iii)$: For $\nu=2$ consecutive demands   $\dbf^{(1)} = \dbf_1=(1,2)$ and $\dbf^{(2)} =\dbf_2= (2,1)$, with corresponding receiver subsets $\mathcal S_1 = \{r_1\}$  and $\mathcal S_2 = \{r_2\}$, based on Theorem \ref{thm: general LB rule}, we have  $\mathcal D_{{\mathcal S}_1}^{(1)}=\mathcal D_{{\mathcal S}_2}^{(2)}=\{1\}$,  $\widetilde{\mathcal S} = \{r_1,r_2\}$, and  ${\mathfrak X}_1=\emptyset $. Therefore,
		\begin{align}
			\Rstar	 
			& \geq \frac{1}{2}(R_{\dbf_1}^*+R_{\dbf_2}^*  )\notag\\	 
			&\geq \frac{1}{2} \;\bigg(    H\Big( \Big\{ \Xsf_d: d\in \mathcal D_{{\mathcal S}_1}^{(1)}\Big\}   \Big)  +  H\Big( \Big\{ \Xsf_d: d\in \mathcal D_{{\mathcal S}_2}^{(2)}\Big\}   \Big|  	{\mathfrak X}_1   \Big)  
			-  |\, \widetilde {\mathcal S}  \,|  M  \notag\\
			& \qquad +     H\Big(\Big\{\Xsf_d: d\in  {\mathcal D}_{ \widetilde{\mathcal S}}^{(1)}\cup {\mathcal D}_{ \widetilde{\mathcal S}}^{(2)}  \Big\}	\Big)   -    H\Big(\Big\{ \Xsf_d: d\in   {\mathcal D}_{{\mathcal S}_1}^{(1)}\cup   {\mathcal D}_{{\mathcal S}_2}^{(2)} \Big \} \Big) 
			\; \bigg)   \notag\\
			& = \frac{1}{2} \;\Big(    H (  \Xsf_1  )  +  H(   \Xsf_1    )  
			\,-\,    2 M  \,+\,     H (\Xsf_1, 	\Xsf_2)   -    H (\Xsf_1 )  \Big)   \notag\\
			& = \frac{1}{2} \;\Big(      H(   \Xsf_1,  	\Xsf_2  )  +  H(   \Xsf_1    ) 	\,-\,   2 M   \Big)  .\notag
		\end{align}
		From considering the same $\nu=2$ demands $\dbf^{(1)}= \dbf_1$ and $\dbf^{(2)}= \dbf_2$, but with the subsets $\mathcal S_1 = \{r_2\}$  and $\mathcal S_2 = \{r_1\}$, $\Rstar$ can be lower bounded as
		\begin{align}
			\Rstar \geq \frac{1}{2} \Big(H(\Xsf_1,\,\Xsf_2)+ \max\Big\{H(\Xsf_1) ,H(\Xsf_2)  \Big\}  -   2M  \Big)  .  \label{eq:LB3}
		\end{align}
	\end{itemize}
	The optimal peak rate-memory function  is lower bounded  by eqs.~\eqref{eq:LB1}-\eqref{eq:LB3}.
	
	\subsection*{II) \bf Lower Bound on $\mathbf \RstarE$:}  
	A lower bound  on the optimal average rate-memory function, $\RstarE$, is obtained by repeatedly applying Theorem~\ref{thm: general LB rule}, to multiple sets of consecutive demand realizations, each taking values among $\dbf_1=(1,2)$, $\dbf_2=(2,1)$, $\dbf_3=(1,1)$ and $\dbf_4=(2,2)$. For a given set of $\nu$ consecutive demands $\dbf^{(1)},\dots,\dbf^{(\nu)}$, the optimal sum rate is given by $\sum\limits_{i=1}^{\nu} R_{\dbf^{(i)}}^*$, and can be lower bounded using Theorem~\ref{thm: general LB rule}. These sets of demands   are chosen such that 
	a proper linear combination of their sum rates results in $\RstarE = \frac{1}{4} \sum\limits_{j=1}^4R_{\dbf_j}^* $.
	Then, a lower bound on  $\RstarE$ is given by the same linear combination of  the lower bounds  on the 
	sum rates of the multiple sets of consecutive demands.
	
	\begin{itemize}
		\item Case $(iv)$: For the $\nu=2$ consecutive demands $\dbf^{(1)}=\dbf_1=(1,2)$ and $\dbf^{(2)}=\dbf_2 = (2,1)$ considered in case $(ii)$ in Appendix~\ref{app:LowerBound}-I, the optimal sum rate is lower bounded by   
		\begin{align}
			R_{\dbf_1}^* + R_{\dbf_2}^*  \geq    H(\Xsf_1,\Xsf_2)   -M  .  \label{eq:case 4-1}
		\end{align}
		Alternatively, for the $\nu=2$ consecutive demands $\dbf^{(1)}=\dbf_3=(1,1)$ and $\dbf^{(2)}=\dbf_4 = (2,2)$, with corresponding receiver subsets $\mathcal S_3 = \mathcal S_4= \{r_1\}$, based on Theorem \ref{thm: general LB rule}, we have $\mathcal D_{{\mathcal S}_1}^{(1)}=\{1\}$, $\mathcal D_{{\mathcal S}_2}^{(2)}=\{2\}$, $\widetilde{\mathcal S}= \{r_1\}$ and ${\mathfrak X}_1= \{\Xsf_1\}$. Therefore,
		\begin{align}
			R_{\dbf_3}^* + R_{\dbf_4}^*  
			& \geq     H\Big( \Big\{ \Xsf_d: d\in \mathcal D_{{\mathcal S}_1}^{(1)}\Big\}  \Big)   
			+  H\Big( \Big\{ \Xsf_d: d\in \mathcal D_{{\mathcal S}_2}^{(2)}\Big\}   \Big|  	{\mathfrak X}_1  \Big)  
			-  |\, \widetilde {\mathcal S}  \,|  M  \notag\\
			& \qquad +     H\Big(\Big\{\Xsf_d: d\in  {\mathcal D}_{ \widetilde{\mathcal S}}^{(1)}\cup   {\mathcal D}_{ \widetilde{\mathcal S}}^{(2)}  \Big\}	\Big)   -    H\Big(\Big\{ \Xsf_d: d\in   {\mathcal D}_{{\mathcal S}_1}^{(1)}\cup   {\mathcal D}_{{\mathcal S}_2}^{(2)} \Big \} \Big) 
			\notag\\
			& \geq    H(\Xsf_1)   + H(\Xsf_2|\Xsf_1)   -M + H(\Xsf_1,\Xsf_2)   - H(\Xsf_1,\Xsf_2)    \notag\\
			& =     H(\Xsf_1,\,\Xsf_2)   -M    . \label{eq:case 4-2}
		\end{align}
		By linearly combining the lower bounds on the optimal sum rates corresponding to the considered sets of $\nu=2$ demands, given in \eqref{eq:case 4-1} and \eqref{eq:case 4-2}, we have
		\begin{align}
			\RstarE \geq \frac{1}{4} \sum\limits_{j=1}^4R_{\dbf_j}^* \geq   \frac{1}{2}\Big(H(\Xsf_1,\,\Xsf_2)   -M  \Big)      . \label{eq: avg LB1}
		\end{align}
		
		\item Case $(v)$: For the $\nu=2$ consecutive demands  $\dbf^{(1)}=\dbf_1=(1,2)$ and $\dbf^{(2)}=\dbf_2 = (2,1)$  considered in case $(iii)$ in Appendix~\ref{app:LowerBound}-I, the optimal sum rate is lower bounded by   
		\begin{align}
			R_{\dbf_1}^* + R_{\dbf_2}^*  \geq    H(\Xsf_1,\Xsf_2) +   \max\Big\{H(\Xsf_1) ,H(\Xsf_2)  \Big\} -2M .  \label{eq:case 5-1}
		\end{align}
		By linearly combining the lower bounds given in \eqref{eq:case 5-1}  and \eqref{eq:case 4-2}, we have
		\begin{align}
			\RstarE \geq \frac{1}{4} \sum\limits_{j=1}^4R_{\dbf_j}^* \geq   \frac{1}{2}H(\Xsf_1,\,\Xsf_2)   
			+ \frac{1}{4}\max\Big\{H(\Xsf_1) ,H(\Xsf_2)  \Big\} 
			- \frac{3}{4}M      . \label{eq: avg LB2}
		\end{align}

		\item Case $(vi)$: 
		For the $\nu=2$ consecutive demands $\dbf^{(1)}=\dbf_1=(1,2)$ and $\dbf^{(2)}=\dbf_3=(1,1)$, with corresponding receiver subsets $\mathcal S_1 = \{r_1\}$ and $\mathcal S_2= \{r_2\}$, based on Theorem \ref{thm: general LB rule}, we have $\mathcal D_{{\mathcal S}_1}^{(1)}=\{1\}$, $\mathcal D_{{\mathcal S}_2}^{(2)}=\{1\}$, $\widetilde{\mathcal S}= \{r_1,r_2\}$ and ${\mathfrak X}_1= \emptyset$. Therefore,
		\begin{align}
			R_{\dbf_1}^* + R_{\dbf_3}^*  
			& \geq     H\Big( \Big\{ \Xsf_d: d\in \mathcal D_{{\mathcal S}_1}^{(1)}\Big\}  \Big)   
			+  H\Big( \Big\{ \Xsf_d: d\in \mathcal D_{{\mathcal S}_2}^{(2)}\Big\}   \Big|  	{\mathfrak X}_1  \Big)  
			-  |\, \widetilde {\mathcal S}  \,|  M  \notag\\
			& \qquad +     H\Big(\Big\{\Xsf_d: d\in  {\mathcal D}_{ \widetilde{\mathcal S}}^{(1)}\cup   {\mathcal D}_{ \widetilde{\mathcal S}}^{(2)}  \Big\}	\Big)   -    H\Big(\Big\{ \Xsf_d: d\in   {\mathcal D}_{{\mathcal S}_1}^{(1)}\cup   {\mathcal D}_{{\mathcal S}_2}^{(2)} \Big \} \Big) 
			\notag\\
			& \geq    H(\Xsf_1)   + H(\Xsf_1)   -2 M + H(\Xsf_1,\Xsf_2)   - H(\Xsf_1)    \notag\\
			& \geq     H(\Xsf_1, \Xsf_2)  + H(\Xsf_1)  -2M .   \label{eq:case 6-1}
		\end{align}
		For the $\nu=2$ consecutive demands $\dbf^{(1)}=\dbf_2=(2,1)$ and $\dbf^{(2)}=\dbf_4=(2,2)$, and receiver subsets $\mathcal S_1 = \{r_1\}$ and $\mathcal S_2= \{r_2\}$, based on Theorem \ref{thm: general LB rule}, we have $\mathcal D_{{\mathcal S}_1}^{(1)}=\{2\}$, $\mathcal D_{{\mathcal S}_2}^{(2)}=\{2\}$, $\widetilde{\mathcal S}= \{r_1,r_2\}$ and ${\mathfrak X}_1= \emptyset$. Therefore,
		\begin{align}
			R_{\dbf_2}^*+ R_{\dbf_4}^*  
			& \geq    H(\Xsf_2)   + H(\Xsf_2)   -2 M + H(\Xsf_1,\Xsf_2)   - H(\Xsf_2)    \notag\\
			& \geq     H(\Xsf_1,\Xsf_2)  + H(\Xsf_2)  -2M.   \label{eq:case 6-2}
		\end{align}
		Linearly combining the lower bounds given in \eqref{eq:case 6-1} and \eqref{eq:case 6-2} results in
		\begin{align}
			\RstarE \geq \frac{1}{4} \sum\limits_{j=1}^4R_{\dbf_j}^* \geq   \frac{1}{2} H(\Xsf_1,\Xsf_2) +  \frac{1}{4}\Big( H(\Xsf_1)+H(\Xsf_2)\Big)   -M     . \label{eq: avg LB3}
		\end{align}

	\end{itemize}
	The optimal average rate-memory function, $\RstarE$, in the two-file two-receiver network is lower bounded by eqs. \eqref{eq: avg LB1}, \eqref{eq: avg LB2} and \eqref{eq: avg LB3}, which also lower bounds $\RstarE$ for the two-file $K$-receiver network. However, in the $K$-receiver network, the lower bound on the average rate can be improved for small cache capacities since it is more probable that receivers request different files rather than the same file. Under a uniform demand distribution, the same file is requested with probability $\frac{2}{2^K}$ and distinct files are requested with probability  $1-\frac{2}{2^K}$.
	\begin{itemize}
		\item Consider demands $\dbf_1=(1,\dots,1)$ and $\dbf_2=(2,\dots,2)$, in which all $K$ receivers request the same file. The optimal rate for each demand is lower bounded as
		\begin{align}
			R_{\dbf_1}^* 	& \geq H(\Xsf_1)-M, \; \text{ and }\;  R_{\dbf_2}^* 	\geq H(\Xsf_2)-M.\label{eq: LB Kuser best}
		\end{align}
		
		\item Consider demands $\dbf_3,\dots, \dbf_{2^K}$, in which at least two of the receivers request different files. The optimal rate for each demand is lower bounded as
		\begin{align}
			R_{\dbf_j}^* 	&  \geq H(\Xsf_1,\Xsf_2)- 2 M,\quad  j=3,\dots ,2^K   .\label{eq: LB Kuser worst}
		\end{align}
		
	\end{itemize}
	Combining eqs. \eqref{eq: LB Kuser best} and \eqref{eq: LB Kuser worst}  results in
	\begin{align} 
		\RstarE \geq \frac{1}{2^K} \sum\limits_{j=1}^{2^K}R_{\dbf_j}^* 
		&\geq \frac{1}{2^K}\Big(H(\Xsf_1)+H(\Xsf_2)- 2M + (2^K-2) H(\Xsf_1,\Xsf_2) - (2^K-2)2M\Big) \notag\\
		& =    \Big(1-\frac{2}{2^K}\Big)   H(\Xsf_1,\Xsf_2)  +\frac{1}{2^K} \Big(H(\Xsf_1)+H(\Xsf_2) \Big)- \Big(2 - \frac{2}{2^K}\Big)  M   . \label{eq: lb avg star}
	\end{align}

	\section{Proof of Theorem \ref{thm:optimality peak}}\label{app:optimality peak} 
	In order to quantify the rate gap to optimality of the proposed GW-MR scheme, we need to compare the peak rate achieved by the GW-MR scheme, $\RGW$ defined in \eqref{eq:ach peak GWMR}, with a lower bound on the optimal peak rate-memory function, $\Rstar$, for different cache sizes. As per Theorem \ref{thm:LowerBound}, a lower bound on $\Rstar$ for the setting with two files and $K$ receivers is given by
	\begin{equation}
		\Rlb =
		\begin{cases}
			H(\Xsf_1,\Xsf_2) - 2M ,        
			& \; M\in [0,\alpha)  \\
			\frac{1}{2}\Big(H(\Xsf_1,\Xsf_2)+\max\Big\{ H(\Xsf_1),H(\Xsf_2)\Big\} \Big)- M ,				   
			& \;    M \in [\alpha, \, \beta ) \\
			\frac{1}{2}\Big( H(\Xsf_1,\Xsf_2) - M\Big)   ,				     
			& \;  M \in [\beta, \, H(\Xsf_1,\Xsf_2)]  . \label{eq: LB peak appendix}
		\end{cases} 	
	\end{equation}	  
	%
	%
	where
	\begin{align}
		& \alpha \triangleq \frac{1}{2}\min\Big\{H(\Xsf_1|\Xsf_2),H(\Xsf_2|\Xsf_1)\Big\},  
		\quad \beta \triangleq   \max\Big\{H(\Xsf_1),H(\Xsf_2)\Big\}= H(\Xsf_1,\Xsf_2) - 2\alpha.  \label{eq:alpha beta}
	\end{align}
	Note that based on \eqref{eq:ach peak GWMR}, for any $\Rscr\in \GWregion$, the gap to optimality satisfies 
	\begin{align}
		\RGW &-   \Rlb  \leq \Rach(M,\Rscr) -   \Rlb,
	\end{align}
	where $\Rach(M,\Rscr)$ is given in \eqref{eq: rate peak K}. 
	Therefore, in the following, we quantify the gap $\Rach(M,\Rscr) -   \Rlb$ for any $\Rscr\in\GWregion$ such that $\R_0 = I(\Xsf_1,\Xsf_2;\Usf), \, \R_1=  H(\Xsf_1|\Usf) ,\, \R_2 = H(\Xsf_2|\Usf)$, where $\Usf$ forms a Markov chain $\Xsf_1-\Usf-\Xsf_2$. For such $\Usf$, since
	\begin{align}
		I(\Xsf_1,\Xsf_2;\Usf)   =  H(\Xsf_1,\Xsf_2)-H(\Xsf_1,\Xsf_2|\Usf)  \stackrel{\Xsf_1-\Usf-\Xsf_2}{=}  H(\Xsf_1,\Xsf_2)-\Big( H(\Xsf_1|\Usf) + H(\Xsf_2|\Usf) \Big) , \notag
	\end{align}
	then 
	\begin{align}
		\R_0+\R_1+\R_2=I(\Xsf_1,\Xsf_2;\Usf)  + H(\Xsf_1|\Usf) + H(\Xsf_2|\Usf)=H(\Xsf_1,\Xsf_2) . \label{eq:special R}
	\end{align}
	From the analysis that follows it will become clear that this choice of $\Rscr$ is sufficient to achieve optimality over a certain region of the memory. For $\gamma_K$ and $\lambda_K$ defined in \eqref{eq: gamma lambda K},	it follows from Lemma~\ref{lemma:markov chain symmetry} in Appendix~\ref{app:lemma1} that $\gamma_K\leq\alpha\leq\beta\leq\lambda_K$. Then, using the upper bound on $\Rach(M,\Rscr)$ given in \eqref{eq: rate peak K} we have:
	\begin{itemize}
		\item[(i)] When $M \in \Big[0,\,  \gamma_K\Big)$, 
		\begin{align} 
			\Rach(M,\Rscr)  &-   \Rlb \leq     \R_0  + \R_1+\R_2 -2M -  \Big( H(\Xsf_1,\Xsf_2) - 2 M \Big)  =0 ,\notag
		\end{align}
		which follows from \eqref{eq:special R}.		
		\item[(ii)] When $M \in \Big[ \gamma_K,\, \alpha\Big)$, 
		\begin{align}
			\Rach(M,\Rscr)   -   \Rlb &\leq   \R_0  + \R_1+\R_2- \frac{1}{K}\min\{\R_1,\R_2\}  -M 
			-  \Big( H(\Xsf_1,\Xsf_2) - 2 M \Big)  \notag\\
			& = M - \frac{1}{K}\min\{ \R_1,\R_2  \}  \notag\\
			& \stackrel{(a)}{\leq}  \frac{1}{2}\min\{H(\Xsf_1|\Xsf_2),H(\Xsf_2|\Xsf_1)\}   -  \frac{1}{K}\min\{ \R_1,\R_2  \}   ,\notag
		\end{align}
		where $(a)$ follows from the fact that $M\leq\alpha = \frac{1}{2}\min\{H(\Xsf_1|\Xsf_2),H(\Xsf_2|\Xsf_1)\}$.

		\item[(iii)] When $M \in \Big[\alpha,\,  \beta   \Big)$:
		\begin{align}
			\Rach(M,\Rscr)    -   \Rlb  &\leq  \R_0  + \R_1+\R_2- \frac{1}{K}\min\{\R_1,\R_2\}  -M  \notag \\
			& \quad -  \frac{1}{2} \Big(H(\Xsf_1,\Xsf_2) + \max\Big\{H(\Xsf_1),H(\Xsf_2)\Big\} \Big) +M  \notag\\
			& = \frac{1}{2} \Big(H(\Xsf_1,\Xsf_2)-\max\Big\{H(\Xsf_1),H(\Xsf_2)\Big\} \Big) -  \frac{1}{K}\min\{ \R_1,\R_2\} \notag\\
			& = \frac{1}{2} \min\Big\{H(\Xsf_1|\Xsf_2),H(\Xsf_2|\Xsf_1)\Big\}  -  \frac{1}{K}\min\{ \R_1,\R_2 \} 
			.\notag
		\end{align}

		\item[(iv)] When $M \in \Big[ \beta ,\,   \lambda_K  \Big)$:
		\begin{align}
			\Rach(M,\Rscr)   -   \Rlb \leq & \R_0  + \R_1+\R_2- \frac{1}{K}\min\{\R_1,\R_2\}  -M  -  \frac{1}{2} \Big(H(\Xsf_1,\Xsf_2) - M \Big)  \notag\\
			& = \frac{1}{2} \Big(H(\Xsf_1,\Xsf_2)  -M  \Big) - \frac{1}{K}\min\{ \R_1,\R_2\} \notag\\
			& \stackrel{(b)}{\leq} \frac{1}{2} \Big(H(\Xsf_1,\Xsf_2)  -\max\{H(\Xsf_1),H(\Xsf_2)\}  \Big) - \frac{1}{K}\min\{ \R_1,\R_2\} \notag\\
			& =\frac{1}{2} \min\Big\{H(\Xsf_1|\Xsf_2),H(\Xsf_2|\Xsf_1)\Big\} - \frac{1}{K}\min\{ \R_1,\R_2\}
			, \notag 
		\end{align}
		where  $(b)$ follows from to the fact that  $M\geq\max\{H(\Xsf_1),H(\Xsf_2)\}$.

		\item[(v)] When $M \in \Big[\lambda_K ,\,H(\Xsf_1,\Xsf_2)\Big]$, 	
		\begin{align}
			\Rach(M,\Rscr)  &-   \Rlb \leq  \frac{1}{2}\Big(\R_0  + \R_1+\R_2-M\Big) -  \frac{1}{2} \Big(H(\Xsf_1,\Xsf_2) - M \Big) =0. \notag
		\end{align}

	\end{itemize}
	Based on the analysis given above it is observed that for all $\Usf$  that satisfy $\Xsf_1-\Usf-\Xsf_2$,  $\Rach(M,\Rscr) = \Rstar= \Rlb$ when $M\in \Big[ 0,\, \gamma_K \Big] \cup\Big[ H(\Xsf_1,\Xsf_2)-2\gamma_K ,\,H(\Xsf_1,\Xsf_2)\Big]$. 
	In order to maximize the region of memory where the proposed GW-MR scheme is optimal, we select  $\Usf$ to be the one that maximizes $\gamma_K$, and its maximum is given by $M_K \triangleq \max\limits_{\Xsf_1-\Usf-\Xsf_2}   \gamma_K$. In the remaining memory region, we have
	\begin{align}
		\Rstar- \Rlb \leq \Rach(M,\Rscr) - \Rlb \leq  \frac{1}{2} \min\Big\{H(\Xsf_1|\Xsf_2),H(\Xsf_2|\Xsf_1)\Big\} - \frac{1}{K}\min\{ \R_1,\R_2\} .\notag
	\end{align}

	\section{Lemma \ref{lemma:markov chain symmetry}}\label{app:lemma1}
	\begin{lemma}\label{lemma:markov chain symmetry}
		For any $\Usf$ with conditional pmf $p(u|x_1,x_2)$ forming a Markov chain $\Xsf_1 - \Usf - \Xsf_2$, we have
		\begin{align}
			&\min\Big\{H(\Xsf_1|\Usf),\, H(\Xsf_2|\Usf)\Big\}\leq \min\Big\{H(\Xsf_1|\Xsf_2),H(\Xsf_2|\Xsf_1)\Big\} ,  \label{eq:min} \\
			&\max\Big\{H(\Xsf_1|\Usf),\, H(\Xsf_2|\Usf)\Big\}\leq \max\Big\{H(\Xsf_1|\Xsf_2),H(\Xsf_2|\Xsf_1)\Big\} ,\label{eq:max} \\
			& \min\Big\{H(\Xsf_1),H(\Xsf_2)\Big\} \leq   I(\Xsf_1, \Xsf_2;\Usf) + \min\Big\{H(\Xsf_1|\Usf),H(\Xsf_2|\Usf)\Big\} , \label{eq:min H} \\
			& \max\Big\{H(\Xsf_1),H(\Xsf_2)\Big\} \leq   I(\Xsf_1, \Xsf_2;\Usf) + \max\Big\{H(\Xsf_1|\Usf),H(\Xsf_2|\Usf)\Big\} . \label{eq:max H}
		\end{align}
	\end{lemma}
	
	\begin{proof}
		The Markov chain $\Xsf_1 - \Usf - \Xsf_2$ is such that
		\begin{align}
			& H(\Xsf_1|\Usf)\leq H(\Xsf_1|\Xsf_2), \text{ and } H(\Xsf_2|\Usf) \leq H(\Xsf_2|\Xsf_1). \label{eq:markov1}
		\end{align}
		\subsection*{I) Proof of eqs. \eqref{eq:min} and \eqref{eq:max}}
		Let us consider two cases:
		\begin{itemize}
			\item If $H(\Xsf_1|\Usf) = \min\Big\{H(\Xsf_1|\Usf),\, H(\Xsf_2|\Usf)\Big\}$, then from \eqref{eq:markov1},
			\begin{align}
				& \min\Big\{H(\Xsf_1|\Usf),\, H(\Xsf_2|\Usf)\Big\}\leq H(\Xsf_1|\Xsf_2),\label{eq: a1}\\
				& \max\Big\{H(\Xsf_1|\Usf),\, H(\Xsf_2|\Usf)\Big\}\leq H(\Xsf_2|\Xsf_1). \label{eq: a2}
			\end{align}
			Therefore, $\min\Big\{H(\Xsf_1|\Usf),\, H(\Xsf_2|\Usf)\Big\}\leq \max\Big\{H(\Xsf_1|\Usf),\, H(\Xsf_2|\Usf)\Big\}\leq  H(\Xsf_2|\Xsf_1)$, and from \eqref{eq: a1}  we have
			\begin{align}
				& \min\Big\{H(\Xsf_1|\Usf),\, H(\Xsf_2|\Usf)\Big\}\leq \min\Big\{H(\Xsf_1|\Xsf_2),H(\Xsf_2|\Xsf_1)\Big\} .\notag
			\end{align}
			\item  If $H(\Xsf_1|\Usf) = \max\Big\{H(\Xsf_1|\Usf),\, H(\Xsf_2|\Usf)\Big\}$, then from \eqref{eq:markov1},
			\begin{align}
				& \max\Big\{H(\Xsf_1|\Usf),\, H(\Xsf_2|\Usf)\Big\}\leq H(\Xsf_1|\Xsf_2),\label{eq: b1}\\
				& \min\Big\{H(\Xsf_1|\Usf),\, H(\Xsf_2|\Usf)\Big\}\leq H(\Xsf_2|\Xsf_1). \label{eq: b2}
			\end{align}
			Therefore, $\min\Big\{H(\Xsf_1|\Usf),\, H(\Xsf_2|\Usf)\Big\}\leq   H(\Xsf_1|\Xsf_2)$, and from \eqref{eq: b2}  we have
			\begin{align}
				& \min\Big\{H(\Xsf_1|\Usf),\, H(\Xsf_2|\Usf)\Big\}\leq \min\Big\{H(\Xsf_1|\Xsf_2),H(\Xsf_2|\Xsf_1)\Big\} .\notag
			\end{align}
			
		\end{itemize}
		Similarly it can be shown that
		\begin{align}
			& \max\Big\{H(\Xsf_1|\Usf),\, H(\Xsf_2|\Usf)\Big\}\leq \max\Big\{H(\Xsf_1|\Xsf_2),H(\Xsf_2|\Xsf_1)\Big\} .\notag
		\end{align}		
		\subsection*{II) Proof of eqs. \eqref{eq:min H} and \eqref{eq:max H}}
		For any $\Usf$ forming a Markov chain we have
		\begin{align}
			H(\Xsf_1,\Xsf_2) & = I(\Xsf_1, \Xsf_2;\Usf)+ H(\Xsf_1|\Usf) +  H(\Xsf_2|\Usf)\notag\\
			&= I(\Xsf_1, \Xsf_2;\Usf)+  \min\Big\{ H(\Xsf_1|\Usf), H(\Xsf_2|\Usf) \Big\}  + \max\Big\{ H(\Xsf_1|\Usf), H(\Xsf_2|\Usf) \Big\} \label{eq:lemma a1}\\
			&\stackrel{(a)}{\leq} 
			I(\Xsf_1, \Xsf_2;\Usf)+  \min\Big\{ H(\Xsf_1|\Xsf_2), H(\Xsf_2|\Xsf_1) \Big\}  + \max\Big\{ H(\Xsf_1|\Usf), H(\Xsf_2|\Usf) \Big\}, \notag
		\end{align}
		where $(a)$ follows from \eqref{eq:min}. 	Therefore,
		\begin{align}
			\max\Big\{ H(\Xsf_1), H(\Xsf_2) \Big\} & = H(\Xsf_1,\Xsf_2) - \min\Big\{ H(\Xsf_1|\Xsf_2), H(\Xsf_2|\Xsf_1) \Big\} \notag\\
			&\leq I(\Xsf_1, \Xsf_2;\Usf)+ \max\Big\{ H(\Xsf_1|\Usf), H(\Xsf_2|\Usf) \Big\} .\notag
		\end{align}
		Similarly, from \eqref{eq:lemma a1} and \eqref{eq:max} we have
		\begin{align}
			H(\Xsf_1,\Xsf_2) &  {\leq} 
			I(\Xsf_1, \Xsf_2;\Usf)+  \min\Big\{H(\Xsf_1|\Usf), H(\Xsf_2|\Usf)  \Big\}  + \max\Big\{ H(\Xsf_1|\Xsf_2), H(\Xsf_2|\Xsf_1) \Big\},
		\end{align}
		which results in
		\begin{align}
			\min\Big\{ H(\Xsf_1), H(\Xsf_2) \Big\} & \leq I(\Xsf_1, \Xsf_2;\Usf)+ \min\Big\{ H(\Xsf_1|\Usf), H(\Xsf_2|\Usf) \Big\} .\notag
		\end{align}

		
	\end{proof}

	
	\section{Proof of Theorem \ref{thm:optimality avg}}\label{app:optimality avg} 
	In order to quantify the rate gap to optimality of the proposed GW-MR scheme, we need to compare the average rate achieved by the GW-MR scheme, $\RGWE$ defined in \eqref{eq:ach avg GWMR}, with a lower bound on the optimal average rate-memory function, $\RstarE$, for different cache sizes.    A lower bound on $\RstarE$ for the setting with two files and $K$ receivers is given in Theorem \ref{thm:LowerBound}. In this Appendix, in order to simplify the analysis  we use a less tight lower bound by only considering the following inequalities from Theorem \ref{thm:LowerBound}
	\begin{equation}
		\RlbE \geq
		\begin{cases}
			\frac{1}{2}  H(\Xsf_1,\Xsf_2)+   \frac{1}{4} \Big( H(\Xsf_1)+H(\Xsf_2)\Big)- M ,        
			& \;  M \in [0,  \, \bar \alpha ) \\
			\frac{1}{2}H(\Xsf_1,\,\Xsf_2) + \frac{1}{4}\max\Big\{H(\Xsf_1) ,H(\Xsf_2)  \Big\}  - \frac{3}{4}M ,				   
			& \;   M \in [\bar \alpha, \,   \beta ) \\
			\frac{1}{2}\Big( H(\Xsf_1,\Xsf_2) - M\Big) ,				     
			& \; M \in [  \beta, \, H(\Xsf_1,\Xsf_2)]  . \label{eq: LB avg appendix}
		\end{cases} 	
	\end{equation}	   
	%
	%
	for $\beta=  \max\{H(\Xsf_1),\,H(\Xsf_2)\} $ defined in \eqref{eq:alpha beta}, and
	\begin{align}
		& \bar\alpha \triangleq  \min\{H(\Xsf_1),\,H(\Xsf_2)\}   \label{eq:alpha bar}.
	\end{align} 	
	Note that based on \eqref{eq:ach avg GWMR}, for any $\Rscr\in \GWregion$ the gap to optimality satisfies 
	\begin{align}
		\RGWE &-   \RlbE  \leq \RachE(M,\Rscr) -   \RlbE,
	\end{align}
	where $\RachE(M,\Rscr)$ is given in \eqref{eq: rate avg K}. 
	Therefore, in the following, we quantify the gap $\RachE(M,\Rscr) -   \RlbE$ for any $\Rscr\in\GWregion$ such that $\R_0 = I(\Xsf_1,\Xsf_2;\Usf), \, \R_1=  H(\Xsf_1|\Usf) ,\, \R_2 = H(\Xsf_2|\Usf)$, where $\Usf$ forms a Markov chain $\Xsf_1-\Usf-\Xsf_2$, and  from \eqref{eq:special R} we have  $\R_0+\R_1+\R_2=H(\Xsf_1,\Xsf_2)$.  
	From the analysis that follows it will become clear that this choice of $\Rscr$ is sufficient to achieve optimality over a certain region of the memory. 
	
	For $\gamma_K$ and $\lambda_K$ defined in \eqref{eq: gamma lambda K} as ${\gamma_{K} }=  \frac{1}{K}\min\{\R_1,\R_2\}$ and $\lambda_{K} =  \R_0+\R_1+\R_2-2 \gamma_{K}$, we have $2\gamma_K\leq \bar\alpha\leq  \beta\leq \lambda_K$. Then, with $\bar\gamma_K$ defined in \eqref{eq: gamma bar K} as $\bar\gamma_{K}=  \R_0+2\gamma_K$, the gap to optimality of the proposed GW-MR scheme is upper bounded as follows.	
	\begin{itemize}
		\item[(i)] When $M \in \Big[0,\, 2 \gamma_K \Big)$,	
		\begin{align} 
			\RachE&(M,\Rscr)  -   \RlbE \notag\\
			&\leq\R_0+\Big(1-\frac{1}{2^K}\Big)  (\R_1+\R_2)-\Big(\frac{3}{2}-\frac{2}{2^K}  \Big)M   -  \Big( \frac{1}{2}  H(\Xsf_1,\Xsf_2)+   \frac{1}{4} \Big( H(\Xsf_1)+H(\Xsf_2)\Big)- M \Big)  \notag\\
			&\stackrel{(a)}{=}    \frac{1}{2}H(\Xsf_1,\Xsf_2)-\frac{1}{4}\Big(H(\Xsf_1 )+H(\Xsf_2 ) \Big) -\frac{1}{2^K}  (\R_1+\R_2) -\Big(\frac{1}{2}-\frac{2}{2^K}  \Big)M \notag\\
			& =\frac{1}{4}\Big(H(\Xsf_1|\Xsf_2 )+H(\Xsf_2 |\Xsf_1) \Big) -\frac{1}{2^K}  (\R_1+\R_2) -\Big(\frac{1}{2}-\frac{2}{2^K}  \Big)M \notag\\
			&  \leq \frac{1}{4}\Big(H(\Xsf_1|\Xsf_2 )+H(\Xsf_2 |\Xsf_1) \Big) -\frac{1}{2^K}  (\R_1+\R_2) , \notag
		\end{align}
		where $(a)$ follows from $\R_0+\R_1+\R_2=H(\Xsf_1,\Xsf_2)$.

		\item[(ii)] When $M \in \Big[ 2 \gamma_K ,\, \min\{\bar\gamma_K,\bar\alpha\} \Big)$,	
		\begin{align} 
			\RachE(M,\Rscr)  -  & \RlbE  \leq	\R_0+\Big(  1-\frac{1}{2^K}  \Big) ( \R_1+\R_2) -   \Big( 1-\frac{4}{2^K}  \Big)  \gamma_K  -M \notag\\
			&\quad  -  \Big( \frac{1}{2}  H(\Xsf_1,\Xsf_2)+   \frac{1}{4} \Big( H(\Xsf_1)+H(\Xsf_2)\Big)- M \Big)  \notag\\
			& =  \frac{1}{4}\Big(H(\Xsf_1|\Xsf_2 )+H(\Xsf_2 |\Xsf_1) \Big)     -\frac{1}{2^K}  (\R_1+\R_2)-   \Big( 1-\frac{4}{2^K}  \Big)  \gamma_K  \notag\\
			&\leq \frac{1}{4}\Big(H(\Xsf_1|\Xsf_2 )+H(\Xsf_2 |\Xsf_1) \Big)     -\frac{1}{2^K}  (\R_1+\R_2) . \notag
		\end{align}

		\item[(iii)] When $M \in \Big[ \min\{\bar\gamma_K,\bar\alpha\} ,\, \max\{\bar\gamma_K,\beta\} \Big)$,	
		\begin{itemize}
			\item[$\bullet$] If $\bar\gamma_K\leq \bar\alpha$, then
			\begin{itemize}
				\item[$\circ$] When $M\in\Big[ \bar\gamma_K,\,  \bar\alpha \Big)$, 		
				\begin{align} 
					\RachE&(M, \Rscr)  -   \RlbE  \leq	\Big(  1-\frac{1}{2^K}  \Big) \Big( \R_0+\R_1+\R_2-M\Big) -   \Big( 1-\frac{2}{2^K}  \Big)   \gamma_K  \notag\\
					& \quad-  \Big( \frac{1}{2}  H(\Xsf_1,\Xsf_2)+   \frac{1}{4} \Big( H(\Xsf_1)+H(\Xsf_2)\Big)- M \Big)  \notag\\
					& \stackrel{(b)}{=}   \frac{1}{4}\Big(H(\Xsf_1|\Xsf_2 )+H(\Xsf_2 |\Xsf_1) \Big)    -\frac{1}{2^K}  		\Big(H(\Xsf_1,\Xsf_2)-M\Big)  -   \Big( 1-\frac{2}{2^K}  \Big)   \gamma_K       \notag \\
					&\stackrel{(c)}{\leq}  \frac{1}{4} \Big(H(\Xsf_1|\Xsf_2 )+H(\Xsf_2 |\Xsf_1) \Big)     -\frac{1}{2^K}  	 \max \Big\{H(\Xsf_1|\Xsf_2 ),H(\Xsf_2 |\Xsf_1) \Big\} -   \Big( 1-\frac{2}{2^K}  \Big)   \gamma_K \notag\\
					&\stackrel{(d)}{\leq}    \frac{1}{4} \Big(H(\Xsf_1|\Xsf_2 )+H(\Xsf_2 |\Xsf_1) \Big)     -\frac{1}{2^K}  	 (\R_1+\R_2)      ,   \notag
				\end{align}
				where  $(b)$ follows from $\R_0+\R_1+\R_2=H(\Xsf_1,\Xsf_2)$, and $(c)$ follows due to the fact that $M \leq \bar\alpha=\min\{H(\Xsf_1),H(\Xsf_2)\}$. Step $(d)$ follows since for $K\geq 2$
				\begin{align}
					\frac{1}{2^K}  	 \max &\Big\{H(\Xsf_1|\Xsf_2 ),H(\Xsf_2 |\Xsf_1) \Big\} +  \Big( 1-\frac{2}{2^K}  \Big)   \gamma_K \notag\\
					&\stackrel{(e)}{\geq}      \frac{1}{2^K}  	 \max  \{\R_1,\R_2\}    +  \Big( 1-\frac{2}{2^K}  \Big)   \gamma_K       \notag\\
					&    \geq   \frac{1}{2^K}  	 \max  \{\R_1,\R_2\}    + \frac{K}{2^K}  \gamma_K  =  \frac{1}{2^K}  	 (\R_1+\R_2)   , \label{eq:extra}
				\end{align}
				where $(e)$ follows from \eqref{eq:max}.
				
				\item[$\circ$] When $M\in\Big[  \bar\alpha ,\, \beta \Big)$,
				\begin{align}
					&\RachE(M,\Rscr)    -    \RlbE 
					\leq  \Big(  1-\frac{1}{2^K}  \Big)\Big( \R_0+\R_1+\R_2-M \Big)-\Big( 1-\frac{2}{2^K}  \Big)   \gamma_K   \notag \\
					& \quad   -\Big( \frac{1}{2}H(\Xsf_1,\,\Xsf_2) + \frac{1}{4}\max\Big\{H(\Xsf_1) ,H(\Xsf_2)  \Big\}  - \frac{3}{4}M\Big) \notag\\
					&=  \frac{1}{4}\min\Big\{H(\Xsf_1|\Xsf_2 ),H(\Xsf_2 |\Xsf_1) \Big\}  + \Big(  \frac{1}{4}-\frac{1}{2^K}  \Big)\Big(  H(\Xsf_1,\Xsf_2)-M\Big)  -   \Big( 1-\frac{2}{2^K}  \Big)   \gamma_K      \notag\\
					&  \stackrel{(f)}{\leq}  \frac{1}{4}\Big(H(\Xsf_1|\Xsf_2 )+H(\Xsf_2 |\Xsf_1) \Big)     -\frac{1}{2^K}\max\Big \{ H(\Xsf_1|\Xsf_2 ),H(\Xsf_2 |\Xsf_1)  \Big\}  -   \Big( 1-\frac{2}{2^K}  \Big)   \gamma_K   \notag\\
					& \stackrel{(g)}{\leq} \frac{1}{4}   \Big(H(\Xsf_1|\Xsf_2)+H(\Xsf_2|\Xsf_1)  \Big)   -\frac{1}{2^K} (\R_1+\R_2)  ,  \notag
				\end{align}
				where $(f)$ follows due to the fact that $M\geq \bar\alpha= \min\{H(\Xsf_1),\,H(\Xsf_2)\} $, and $(g)$ follows from \eqref{eq:extra}.
				
			\end{itemize}

			\item[$\bullet$] If $\bar\gamma_K>  \bar\alpha$, then 
			\begin{itemize}
				\item[$\circ$] When $M\in\Big[ \bar\alpha,\,   \min\{ \bar\gamma_K,\beta \} \Big)$, 		
				\begin{align} 
					\RachE&(M, \Rscr)  -   \RlbE \leq	\R_0+\Big(  1-\frac{1}{2^K}  \Big) ( \R_1+\R_2) -   \Big( 1-\frac{4}{2^K}  \Big)  \gamma_K -M \notag\\
					& \quad -\Big( \frac{1}{2}H(\Xsf_1,\,\Xsf_2) + \frac{1}{4}\max\Big\{H(\Xsf_1) ,H(\Xsf_2)  \Big\}  - \frac{3}{4}M\Big) \notag\\
					& =  \frac{1}{2}H(\Xsf_1,\Xsf_2 )- \frac{1}{4}\Big(M+\max\Big\{H(\Xsf_1) ,H(\Xsf_2)  \Big\}  \Big)   -\frac{1}{2^K}  		(\R_1+\R_2)    -  \Big( 1-\frac{4}{2^K}  \Big)  \gamma_K      \notag \\
					&\stackrel{(h)}{\leq} \frac{1}{4}\Big(H(\Xsf_1|\Xsf_2 )+H(\Xsf_2 |\Xsf_1) \Big)     -\frac{1}{2^K}  		(\R_1+\R_2)    -  \Big( 1-\frac{4}{2^K}  \Big)  \gamma_K       \notag\\
					&\leq  \frac{1}{4}\Big(H(\Xsf_1|\Xsf_2 )+H(\Xsf_2 |\Xsf_1) \Big)     -\frac{1}{2^K}  		(\R_1+\R_2)  , \notag
				\end{align}
				where $(h)$ follows from the fact that $M \geq \bar\alpha=\min\{H(\Xsf_1),H(\Xsf_2)\}$.

				\item[$\circ$] When $M\in\Big[  \min\{ \bar\gamma_K,\beta \}  ,\,  \max\{ \bar\gamma_K,\beta \} \Big)$,
				
				- If $\bar\gamma_K \leq \beta$, then
				\begin{align} 
					\RachE&(M, \Rscr)  -   \RlbE \leq	 \Big(  1-\frac{1}{2^K}  \Big)\Big( \R_0+\R_1+\R_2-M \Big)- \Big(  1-\frac{2}{2^K}  \Big)\gamma_K   \notag\\
					& \quad -\Big( \frac{1}{2}H(\Xsf_1,\,\Xsf_2) + \frac{1}{4}\max\Big\{H(\Xsf_1) ,H(\Xsf_2)  \Big\}  - \frac{3}{4}M\Big) \notag\\
					&=  \frac{1}{4}\min\Big\{H(\Xsf_1|\Xsf_2 ),H(\Xsf_2 |\Xsf_1) \Big\}  + \Big(  \frac{1}{4}-\frac{1}{2^K}  \Big)\Big(  H(\Xsf_1,\Xsf_2)-M\Big) - \Big(  1-\frac{2}{2^K}  \Big)\gamma_K         \notag\\
					&  \stackrel{(i)}{\leq}  \frac{1}{4}\Big(H(\Xsf_1|\Xsf_2 )+H(\Xsf_2 |\Xsf_1) \Big)     -\frac{1}{2^K}\max\Big \{ H(\Xsf_1|\Xsf_2 ),H(\Xsf_2 |\Xsf_1)  \Big\}  - \Big(  1-\frac{2}{2^K}  \Big)\gamma_K     \notag\\
					& \stackrel{(j)}{\leq} \frac{1}{4}   \Big(H(\Xsf_1|\Xsf_2)+H(\Xsf_2|\Xsf_1)  \Big)   -\frac{1}{2^K} (\R_1+\R_2)  ,  \notag
				\end{align}
				where $(i)$ follows since  $M\geq \bar\alpha= \min\{H(\Xsf_1),\,H(\Xsf_2)\} $, and $(j)$ follows from \eqref{eq:extra}.

				- If $\bar\gamma_K > \beta$, then
				\begin{align} 
					\RachE&(M, \Rscr)  -   \RlbE \leq	\R_0+\Big(  1-\frac{1}{2^K}  \Big) ( \R_1+\R_2) -   \Big( 1-\frac{4}{2^K}  \Big)  \gamma_K -M \notag\\
					& \quad -  \frac{1}{2}\Big(H(\Xsf_1,\Xsf_2) - M \Big)  \notag\\
					& =  \frac{1}{2}\Big(H(\Xsf_1,\Xsf_2) - M \Big) -\frac{1}{2^K}  		(\R_1+\R_2)    -  \Big( 1-\frac{4}{2^K}  \Big)  \gamma_K      \notag \\
					&\stackrel{(k)}{\leq} \frac{1}{2}\min\Big\{H(\Xsf_1|\Xsf_2 ),H(\Xsf_2 |\Xsf_1) \Big\}-\frac{1}{2^K}  		(\R_1+\R_2)    -  \Big( 1-\frac{4}{2^K}  \Big)  \gamma_K      \notag\\
					& \leq\frac{1}{4}\Big(H(\Xsf_1|\Xsf_2 )+H(\Xsf_2 |\Xsf_1) \Big)     -\frac{1}{2^K}  		(\R_1+\R_2)      , \notag
				\end{align}
				where $(k)$ follows since $M\geq \beta= \max\{H(\Xsf_1),\,H(\Xsf_2)\} $.

			\end{itemize}

		\end{itemize}

		\item[(iv)] When $M \in \Big[\max\{ \bar\gamma_K, \beta\},\,   \lambda_K \Big)$,
		\begin{align}  
			\RachE(M,&\Rscr)  -   \RlbE \leq \Big(  1-\frac{1}{2^K}  \Big)\Big( \R_0+\R_1+\R_2-M \Big)- \Big(  1-\frac{2}{2^K}  \Big) \gamma_K  -  \frac{1}{2}\Big(H(\Xsf_1,\Xsf_2) - M \Big) \notag\\
			& =      \frac{1}{2}  \Big(H(\Xsf_1,\Xsf_2) - M \Big)  -\frac{1}{2^K}    \Big(H(\Xsf_1,\Xsf_2) - M \Big) - \Big(  1-\frac{2}{2^K}  \Big) \gamma_K  \notag\\
			&  \stackrel{(\ell)}{\leq}    \frac{1}{2} \min\Big\{H(\Xsf_1|\Xsf_2),H(\Xsf_2|\Xsf_1) \Big\} -\frac{1}{2^K}  \max\Big\{H(\Xsf_1|\Xsf_2),H(\Xsf_2|\Xsf_1) \Big\}      - \Big(  1-\frac{2}{2^K}  \Big) \gamma_K  \notag\\
			&  \stackrel{(m)}{\leq}  \frac{1}{4}   \Big(H(\Xsf_1|\Xsf_2)+H(\Xsf_2|\Xsf_1)  \Big)  -\frac{1}{2^K}      \max\Big\{H(\Xsf_1|\Xsf_2),H(\Xsf_2|\Xsf_1) \Big\}    - \Big(  1-\frac{2}{2^K}  \Big) \gamma_K  \notag\\
			&   \stackrel{(n)}{\leq}     \frac{1}{4}   \Big(H(\Xsf_1|\Xsf_2)+H(\Xsf_2|\Xsf_1)  \Big)   -\frac{1}{2^K} (\R_1+\R_2)  , \notag
		\end{align}
		where $(\ell)$ follows by upper bounding the first term using the fact that $M\geq \beta = \max\{H(\Xsf_1),$ $H(\Xsf_2)\} $, and upper bounding the second term using the fact that $M\geq \min\{H(\Xsf_1),H(\Xsf_2)\}$. Step $(m)$ follows since the minimum of the two terms is no more than their arithmatic mean, and finally, $(n)$ follows from \eqref{eq:extra}.

		\item[(v)] When  $M \in \Big[\lambda_K,\, H(\Xsf_1,\Xsf_2)\Big]$, 
		\begin{align}
			\RachE(M,\Rscr)  &-   \RlbE \leq    \frac{1}{2}  \Big( \R_0+\R_1+\R_2   -M\Big ) -  \frac{1}{2} \Big(H(\Xsf_1,\Xsf_2) - M \Big) =0 .\notag
		\end{align}
		
	\end{itemize}
	Based on the analysis given above it is observed that for all $\Usf$  that satisfy $\Xsf_1-\Usf-\Xsf_2$,  $\RachE(M,\Rscr) = \RstarE= \RlbE$ when $M\in  \Big[ \lambda_K ,\,H(\Xsf_1,\Xsf_2)\Big]$.  In order to maximize the region of memory where the proposed GW-MR  scheme is optimal, we select  $\Usf$ to be the one that minimizes $\lambda_K$, or equivalently maximizes $\gamma_K=H(\Xsf_1,\Xsf_2)-2\gamma_K$, and its maximum is given by $M_K \triangleq \max\limits_{\Xsf_1-\Usf-\Xsf_2}   \gamma_K$. In the remaining memory region, we have
	\begin{align}
		\RstarE- \RlbE \leq \RachE(M,\Rscr) - \RlbE \leq  \frac{1}{4}\Big(H(\Xsf_1|\Xsf_2 )+H(\Xsf_2 |\Xsf_1) \Big)     -\frac{1}{2^K}  		(\R_1+\R_2) .\notag
	\end{align}

	\section{Optimal Cache Allocation for Three Files and Proof of Theorem \ref{thm:achievable rate three} }\label{app:achievable rate three}
	In this Appendix we prove the optimality of the cache allocation described in Sec.~\ref{sec: achievable three} and we compute the corresponding MR peak rate given in Theorem~\ref{thm:achievable rate three}. As explained in Sec.~\ref{sec: achievable three}, the proposed MR scheme for three files  caches and delivers content from sublibraries $L_1$, $L_2$ and $L_3$ independently, as follows:
	\begin{itemize}
		\item Common-to-all sublibrary $L_3$: Since the common description $W_{123}$ is required for the lossless reconstruction of any of the files, for any demand realization $\dbf$, it is required by both receivers. Therefore, it is optimal to adopt LFU caching and naive multicasting delivery for sublibrary $L_3$. Let $\mu_0\in [0, \min\{M,\R_0\}]$ denote the portion of memory  allocated to this sublibrary. Then, each receiver caches the first $\mu_0F$ bits of $\Wsf_{123}$, and for any demand the remaining $(\R_0 - \mu_0)F$ bits are delivered through uncoded multicast transmissions. 
		
		\item Common-to-two sublibrary $L_2$: Let $\mu'\in\Big [0, \min\{M,3\R'\}\Big]$ denote the portion of memory allocated to sublibrary $L_2$. Then, descriptions $\{W_{12},W_{13},W_{23}\}$ with rate $\R'$ are cached and delivered according to the two-request CACM scheme proposed in Sec.~\ref{sec: new scheme}, whose achievable peak rate, denoted by $R_{L_2}(\mu',\, \R')$,  is given in  \eqref{eq:two request rate}.

		\item Private sublibrary $L_1$: Let $\mu\in\Big [0, \min\{M,3\R\}\Big]$ denote the portion of memory allocated to sublibrary $L_1$. Then, the private descriptions $\{W_{1},W_{2},W_{3}\}$ with rate $\R$ are cached and delivered according to the correlation-unaware scheme proposed in \cite{yu2016exact}, whose achievable peak rate, denoted by $R_{L_1}(\mu,\, \R)$,  is given by  
		\begin{equation}
			R_{L_1}(\mu,\R) = \begin{cases}
				2\R-\mu ,                
				& \; \mu \in \Big[0,\, \frac{3}{2}\R\Big)  \\
				\R-\frac{1}{3}\mu ,                
				& \; \mu \in \Big[\frac{3}{2}\R, \,3\R\Big] \notag
			\end{cases}
		\end{equation}
		
	\end{itemize}
	
	The optimal cache allocations among the three sublibraries, $(\mu_0^*,\mu'^*,\mu^*)$, are derived from the following linear program
	\begin{equation}\label{eq:general GW 3 files}
		\begin{aligned}
			&\min\limits_{\mu_0,\, \mu',\, \mu}
			& & \Rach (M,\Rscr) =   \R_0 - \mu_0   \,+\,   R_{ L_2}(\mu',\, \R') \,+\, R_{L_1}(\mu,\, \R) \\
			& \text{s.t}
			& & \mu_0+\mu'+\mu  \leq M, \\
			&&& 0 \leq \mu_0\leq \R_0,	\\
			&&& 0 \leq \mu'\leq 3\R',	\\
			&&& 0 \leq \mu\leq 3\R	.
		\end{aligned}
	\end{equation}
	For any $M\in\Big[0,\,\R_0+\R_1+\R_2\Big]$, the rate is minimized when the cache capacity is divided among the sublibraries such that a larger portion of the capacity is assigned to the sublibrary that achieves a larger reduction in delivery rate, i.e., whose rate function has the steepest descending slope. For given rate functions $R_{L_1}(\mu,\R)$, $R_{L_2}(\mu',\R')$  and $\R_0-\mu_0$  corresponding to the schemes adopted for sublibraries $L_1$, $L_2$ and $L_3$, respectively, the cache is optimally allocated as follows: 
	\begin{itemize}
		\item[$\bullet$] When $M \in \Big[0, \frac{1}{2}\R' \Big)$, since $\mu'\leq M\leq \frac{1}{2}\R'$, the two-request scheme used for sublibrary $L_2$ achieves a delivery rate equal to $R_{L_2}(\mu,\, \R') = 3\R'-2\mu'$, which has a larger slope (in absolute value) compared to $R_{L_1}(M,\, \R)$ and $\R_0-M$, i.e., it is the most effective in minimizing the delivery rate. 
		Therefore, $\mu'^* = M $ and $\mu_0^*= \mu^* = 0$, and
		\begin{align}
			\Rach(M,\Rscr)  & = \R_0 + 3\R'+2\R - \mu_0^* - 2\mu'^{*} - \mu^* = \R_0 + 3\R' + 2\R - 2M . \notag 
		\end{align}
		
		\item[$\bullet$] When $M \in \Big[   \frac{1}{2}\R', \, \R_0+\frac{3}{2}(\R'+\R ) \Big)$, for any cache allocation $(\mu_0,\mu',\mu)$ with $\mu_0\in[0,\R_0]$, $\mu'\in[\frac{1}{2}\R',\frac{3}{2}\R')$ and $\mu = M-\mu_0-\mu'$, the slope of the line tangent to the rate  functions of all three adopted schemes is equal. Therefore, any choice of $(\mu_0^*,\mu'^*,\mu^*)$ satisfying $\mu_0^*+\mu'^*+\mu^*=M$, and such that $\mu_0^*\in[0,\R_0]$,  $\mu'^*\in[\frac{1}{2}\R',\frac{3}{2}\R')$ and $\mu^*\in[0,\frac{3}{2}\R)$ is optimal. Then,	
		\begin{align}
			\Rach(M,\Rscr)  & = \R_0 + \frac{5}{2}\R'+2\R - \mu_0^* - \mu'^* - \mu^*  =\R_0 + \frac{5}{2}\R' + 2\R - M . \notag 
		\end{align}
		
		\item[$\bullet$] When $M \in \Big[ \R_0+\frac{3}{2}(\R'+\R )  , \,   \R_0+3\R'+\frac{3}{2}\R  \Big)$, for any $(\mu_0,\mu',\mu)$ with $\mu_0 =\R_0$, $\mu'\in[\frac{3}{2}\R',3\R']$, and $\mu\in[\frac{3}{2}\R,3\R]$, the steepest slope corresponds to the rate function of sublibrary $L_2$, $R_{L_2}(M,\, \R') = 2\R'-\frac{2}{3}M$. Therefore, the optimal cache allocations are given by  $\mu_0^* = \R_0$,  $\mu^* = \frac{3}{2}\R$ and  $\mu'^* = M -\mu_0^*-\mu^*$, and  hence,
		\begin{align}
			\Rach(M,\Rscr)  & = \R_0 + 2\R'+2\R - \mu_0^* - \frac{2}{3}\mu'^* - \mu^*  =\frac{2}{3}\R_0 + 2\R' + \frac{3}{2}\R - \frac{2}{3}M . \notag 
		\end{align}
		
		\item[$\bullet$] When $M \in \Big[ \R_0 +  3\R' + \frac{3}{2}\R,\, \R_0 +  3(\R' +\R) \Big]$, for any $(\mu_0,\mu',\mu)$ with $\mu_0=\R_0$, $\mu'=3\R '$, and $\mu\in[\frac{3}{2}\R,3\R]$, the steepest slope corresponds to the rate function of sublibrary $L_1$, $R_{L_1}(M,\, \R) = \R-\frac{1}{3}M$. Therefore, the optimal cache allocations are given by $\mu_0^* = \R_0$, $\mu'^* = 3\R ' $ and $ \mu^* = M-\mu_0^*- \mu'^*$, and hence,
		\begin{align}
			\Rach(M,\Rscr) 
			& = \R_0  + 2\R'+\R -  \mu_0^* - \frac{2}{3}\mu'^* -\frac{1}{3}\mu^* =  \frac{1}{3}\R_0  + \R'+ \R  \,-\, \frac{1}{3} M .  \notag
		\end{align}
		
	\end{itemize}
	The optimal caching strategy described above results in the rate provided in Theorem~\ref{thm:achievable rate three}.

	\section{Optimality of the two-request CACM scheme}\label{app:tworequest}
	In this Appendix, we prove that the two-request scheme proposed in Sec.~\ref{sec: new scheme} is optimal at the memory-rate pairs given in \eqref{eq: opt points of TR}. We refer to the network with two receivers and three independent files $\{W_{12},W_{13},W_{23}\}$ with length $\R' F$ bits, where each receiver requests two files as described in Sec.~\ref{sec: new scheme}, as the {\em two-request network}. The following lemma provides a lower bound on the optimal peak rate-memory trade-off in the two-request network, which we use to establish the optimality for any $M\in[0,\frac{3}{2}\R']$.
	\begin{lemma}
		For a given cache capacity $M$ and description rate $\R'$, a lower bound on the optimal peak rate-memory trade-off in the two-request network adopted for sublibrary $L_2$, denoted by $R_{L_2}^{LB} (M,\R')$,  is given by  
		\begin{align}
			R_{L_2}^{LB} (M,\R')  = \inf \Big \{ R: \quad
			&  R \,\geq\, 3\R'- 2M ,  \;\;
			R \,\geq\, \frac{5}{2}\R' -M, \;\;
			R \,\geq\,  \frac{3}{2} \R' - \frac{1}{2}M  
			\Big \}. \label{eq: LB two-request}
		\end{align}
	\end{lemma}	
	\begin{proof}
		The lower bound is derived using Theorem~\ref{thm: general LB rule} in Appendix~\ref{app:LowerBound}. Since the procedure is similar to those in Appendices \ref{app:LowerBound} and \ref{app:LowerBound three}, here, we only delineate the different sets of $\nu$ consecutive demands and receiver subsets which are used to compute the lower bound.
		
		\begin{itemize}
			
			\item Case $(i)$: For $\nu=1$ consecutive demand   
			$\dbf^{(1)}  = \Big( \{ 12,13 \},\{12,23 \} \Big)$ with receiver subset $\mathcal S_{1} = \{r_1,r_2\}$, based on  Theorem~\ref{thm: general LB rule}, we have $\mathcal D_{{\mathcal S}_1}^{(1)}=\{  12,13 ,23 \} $ and $\widetilde{\mathcal S} =\mathcal S_1 = \{r_1,r_2\}$. Therefore,
			\begin{align}
				\Rstar	 
				& \geq R_{\dbf^{(1)}}^* \geq H(W_{12},W_{13}, W_{23}) \,- \,2 \,M =3\R'-2M. \label{eq:TR LB 1}
			\end{align}

			\item Case $(ii)$: For the $\nu=2$ consecutive demands $\dbf^{(1)}  = \Big( \{ 12,13 \},\{12,23 \} \Big)$ and $\dbf^{(2)}  = \Big( \{ 12,23 \},\{12,13 \}  \Big)$ with receiver subsets $\mathcal S_{1} = \{r_1\}$ and $\mathcal S_{2} = \{r_2\}$, based on  Theorem~\ref{thm: general LB rule}, we have $\mathcal D_{{\mathcal S}_1}^{(1)}= \mathcal D_{{\mathcal S}_2}^{(2)}=\{ 12,13  \}$ and $\widetilde{\mathcal S} = \{r_1,r_2\}$, and ${\mathfrak X}_1 = \emptyset$.  Therefore,
			\begin{align}
				\Rstar	 
				& \geq \frac{1}{2}\Big(R_{\dbf^{(1)}}^*+ R_{\dbf^{(2)}}^*  \Big)\geq \frac{1}{2}\Big( H(W_{12},W_{13})+  H(W_{12},W_{13}, W_{23})  -2M\Big) =\frac{5}{2}\R'-M.  \label{eq:TR LB 2}
			\end{align}

			\item Case $(iii)$: For the $\nu=2$ consecutive demands $\dbf^{(1)}  = \Big(\{ 12,13 \},\{12,23 \}  \Big)$ and $\dbf^{(2)}  = \Big( \{ 12,23 \},\{12,13 \} \Big)$ with receiver subsets $\mathcal S_{1} =\mathcal S_{2} = \{r_1\}$, based on  Theorem~\ref{thm: general LB rule}, we have $\mathcal D_{{\mathcal S}_1}^{(1)}= \{ 12,13  \}$, $\mathcal D_{{\mathcal S}_2}^{(2)}= \{ 12,23 \} $ and $\widetilde{\mathcal S} = \{r_1\}$, and ${\mathfrak X}_1 = \{W_{12}\}$.  Therefore,
			\begin{align}
				\Rstar	 
				& \geq \frac{1}{2}\Big(R_{\dbf^{(1)}}^*+ R_{\dbf^{(2)}}^*  \Big)\geq \frac{1}{2}\Big( H(W_{12},W_{13})+  H(W_{23})  -M\Big) =\frac{3}{2}\R'-\frac{1}{2}M.  \label{eq:TR LB 3}
			\end{align}

		\end{itemize}
		A lower bound on the optimal peak rate-memory trade-off is given by eqs. \eqref{eq:TR LB 1}-\eqref{eq:TR LB 3}.
	\end{proof}
	From comparing $R_{L_2} (M,\R')$, the achievable peak rate given in \eqref{eq:two request rate}, with $R_{L_2}^{LB} (M,\R')$, the lower bound given in \eqref{eq: LB two-request}, we observe that the proposed two-request CACM scheme achieves the lower bound, and is therefore optimal for any $M\in[0,\, \frac{3}{2}\R']$.

	\section{Proof of Theorem \ref{thm:LowerBound three} }\label{app:LowerBound three}
	A lower bound  on the optimal peak rate-memory function, $\Rstar$, is obtained by applying Theorem~\ref{thm: general LB rule} given in Appendix~\ref{app:LowerBound} to different sets of $\nu$  consecutive demand realizations $\dbf^{(1)},\dots,\dbf^{(\nu)}$, each taking values among the worst-case demands $\dbf = (i,j)$ with $i,j\in\{1,2,3\}$, and $i \neq j$. Then, since $\Rstar \geq \frac{1}{\nu} \sum\limits_{i=1}^{\nu} R_{\dbf^{(i)}}^* $, a lower bound on  $\Rstar$ is obtained from normalizing the lower bound on the optimal sum rate $\sum\limits_{i=1}^{\nu} R_{\dbf^{(i)}}^* $ given by Theorem~\ref{thm: general LB rule}. Specifically,

	\begin{itemize}
		\item Case $(i)$: For $\nu=1$ consecutive demand $\dbf^{(1)} = (i,j)$ with receiver subset $\mathcal S_{1} = \{r_1,r_2\}$, based on  Theorem~\ref{thm: general LB rule}, we have $\mathcal D_{{\mathcal S}_{1}}^{(1)}=\{i,j\}$ and  $\widetilde{\mathcal S} = \{r_1,r_2\}$. Therefore,
		\begin{align}
			\Rstar	 
			& \geq R_{\dbf^{(1)}}^*\notag\\
			& \geq     H\Big( \Big\{ \Xsf_d: d\in \mathcal D_{{\mathcal S}_{1}}^{(1)}\Big\}   \Big)  
			-  |\, \widetilde {\mathcal S}  \,|  M  
			+     H\Big(\Big\{\Xsf_d: d\in  {\mathcal D}_{ \widetilde{\mathcal S}}^{(1)}  \Big\}	\Big)   -    H\Big(\Big\{ \Xsf_d: d\in   {\mathcal D}_{{\mathcal S}_{1}}^{(1)} \Big \} \Big)  \notag\\
			& = H(\Xsf_i, \Xsf_j) \,- \,2 \,M  .  \label{eq:LB1 three}
		\end{align}
		
		\item Case $(ii)$: For $\nu=2$ consecutive demands $\dbf^{(1)} = (i,j)$ and $\dbf^{(2)} = (j,i)$ with corresponding receiver subsets $\mathcal S_{1} = \mathcal S_{2} = \{r_1\}$, based on  Theorem~\ref{thm: general LB rule}, we have $\mathcal D_{{\mathcal S}_{1}}^{(1)}=\{i\}$, $\mathcal D_{{\mathcal S}_{2}}^{(2)}=\{j\}$, $\widetilde{\mathcal S} = \{r_1\}$, and ${\mathfrak X}_{1} = \{\Xsf_i\}$. Therefore,
		\begin{align}
			\Rstar	 
			& \geq \frac{1}{2}(R_{\dbf^{(1)}}^*+R_{\dbf^{(2)}}^*  )\notag\\	 
			& \geq 
			\frac{1}{2} \;\bigg[    H\Big( \Big\{ \Xsf_d: d\in \mathcal D_{{\mathcal S}_{1}}^{(1)}\Big\}   \Big)  +  H\Big( \Big\{ \Xsf_d: d\in \mathcal D_{{\mathcal S}_{2}}^{(2)}\Big\}   \Big|  	{\mathfrak X}_{t_1}   \Big)  
			-  |\, \widetilde {\mathcal S}  \,|  M  \notag\\
			& \qquad +     H\Big(\Big\{\Xsf_d: d\in  {\mathcal D}_{ \widetilde{\mathcal S}}^{(1)}\cup {\mathcal D}_{ \widetilde{\mathcal S}}^{(2)}  \Big\}	\Big)   -    H\Big(\Big\{ \Xsf_d: d\in   {\mathcal D}_{{\mathcal S}_{1}}^{(1)}\cup   {\mathcal D}_{{\mathcal S}_{2}}^{(2)} \Big \} \Big) 
			\; \bigg]   \notag\\
			& = \frac{1}{2} \;\Big[    H (  \Xsf_i  )  +  H(   \Xsf_j |  	\Xsf_i  )  
			\,-\,    M  \,+\,     H (\Xsf_i,	\Xsf_j)   -    H (\Xsf_i,	\Xsf_j)  \Big]   \notag\\
			& = \frac{1}{2} \;\Big(      H(   \Xsf_i,  	\Xsf_j )  	\,-\,    M   \Big) . \label{eq:LB2 three}
		\end{align}
		
		\item Case $(iii)$: For $\nu=3$ consecutive demands $\dbf^{(1)} = (i,j)$, $\dbf^{(2)} = (j,k)$ and $\dbf^{(3)} = (k,i)$ with  receiver subsets $\mathcal S_{1} = \mathcal S_{2} = \mathcal S_{3} = \{r_1\}$, based on  Theorem~\ref{thm: general LB rule}, we have $\mathcal D_{{\mathcal S}_{1}}^{(1)}=\{i\}$, $\mathcal D_{{\mathcal S}_{2}}^{(2)}=\{j\}$, $\mathcal D_{{\mathcal S}_{3}}^{(3)}=\{k\}$,  $\widetilde{\mathcal S} = \{r_1\}$, ${\mathfrak X}_{1} = \{\Xsf_i\}$, and ${\mathfrak X}_{2} = \{\Xsf_i,\Xsf_j\}$. Therefore,
		\begin{align}
			\Rstar	 
			& \geq \frac{1}{3}(R_{\dbf^{(1)}}^*+R_{\dbf^{(2)}}^* +R_{\dbf^{(3)}}^* )\notag\\	 
			& \geq 
			\frac{1}{3} \;\bigg[   \sum\limits_{p=1}^3 H\Big( \Big\{ \Xsf_d: d\in \mathcal D_{{\mathcal S}_{p}}^{(p)}\Big\} \Big|    {\mathfrak X}_{p},\dots,	{\mathfrak X}_{{p-1}}  \Big)    
			-  |\, \widetilde {\mathcal S}  \,|  M  \notag\\
			& \qquad +     H\Big(\Big\{\Xsf_d: d\in  \bigcup\limits_{p=1}^3  {\mathcal D}_{ \widetilde{\mathcal S}}^{(p)}  \Big\}	\Big)   -    H\Big(\Big\{ \Xsf_d: d\in  \bigcup\limits_{p=1}^3  {\mathcal D}_{{\mathcal S}_{p}}^{(p)} \Big \} \Big) 
			\; \bigg]   \notag\\
			& = \frac{1}{3} \;\Big[    H (  \Xsf_i  )  +  H(   \Xsf_j |  	\Xsf_i  )  +  H(   \Xsf_k |  	\Xsf_i,	\Xsf_j  )  
			-   M  +    H (\Xsf_1,	\Xsf_2,\Xsf_3)   -     H (\Xsf_1,	\Xsf_2,\Xsf_3) \Big]   \notag\\
			& = \frac{1}{3} \;\Big(      H(   \Xsf_1,  	\Xsf_2 ,	\Xsf_3 )  	\,-\,    M   \Big)  . \label{eq:LB3 three}
		\end{align}	
		
		\item Case $(iv)$: For $\nu=2$ consecutive demands $\dbf^{(1)} = (i,j)$ and $\dbf^{(2)} = (k,i)$ with receiver subsets $\mathcal S_{1} =\{r_1\}$ and $\mathcal S_{2} = \{r_2\}$, based on  Theorem~\ref{thm: general LB rule}, we have $\mathcal D_{{\mathcal S}_{1}}^{(1)}=\mathcal D_{{\mathcal S}_{2}}^{(2)}=\{i\}$, $\widetilde{\mathcal S} = \{r_1,r_2\}$ and ${\mathfrak X}_{1} = \emptyset$. Therefore,
		
		\begin{align}
			\Rstar	 
			& \geq \frac{1}{2}(R_{\dbf^{(1)}}^*+R_{\dbf^{(2)}}^* )\notag\\	 
			&\frac{1}{2} \;\bigg[    H\Big( \Big\{ \Xsf_d: d\in \mathcal D_{{\mathcal S}_{1}}^{(1)}\Big\}   \Big)  +  H\Big( \Big\{ \Xsf_d: d\in \mathcal D_{{\mathcal S}_{2}}^{(2)}\Big\}   \Big|  	{\mathfrak X}_{1}   \Big)  
			-  |\, \widetilde {\mathcal S}  \,|  M  \notag\\
			& \qquad +     H\Big(\Big\{\Xsf_d: d\in  {\mathcal D}_{ \widetilde{\mathcal S}}^{(1)}\cup {\mathcal D}_{ \widetilde{\mathcal S}}^{(2)}  \Big\}	\Big)   -    H\Big(\Big\{ \Xsf_d: d\in   {\mathcal D}_{{\mathcal S}_{1}}^{(1)}\cup   {\mathcal D}_{{\mathcal S}_{2}}^{(2)} \Big \} \Big) 
			\; \bigg]   \notag\\
			& = \frac{1}{2} \;\Big[    H (  \Xsf_i  )  +  H(   \Xsf_i    )  
			\,-\,    2 M  \,+\,     H (\Xsf_1, 	\Xsf_2,\Xsf_3)   -    H (\Xsf_i )  \Big]   \notag\\
			& = \frac{1}{2} \;\Big(      H(   \Xsf_1, 	\Xsf_2 ,\Xsf_3 )  +  H(   \Xsf_i    ) 	\,-\,   2 M   \Big) .\label{eq:LB4 three}
		\end{align}
		
	\end{itemize}
	A lower bound on the optimal peak rate-memory function is given by eqs.~\eqref{eq:LB1 three}-\eqref{eq:LB4 three} for $i,j\in\{1,2,3\}$.

	\section{Proof of Theorem \ref{thm:optimality three} }\label{app:optimality three} 
	In order to quantify the rate gap to optimality of the proposed GW-MR scheme, we need to compare the peak rate achieved by the GW-MR scheme, $\RGW$ defined in \eqref{eq:ach peak GWMR three}, with a lower bound on the optimal peak rate-memory function, $\Rstar$, for different cache sizes. As per Theorem \ref{thm:LowerBound three}, a lower bound on $\Rstar$ for the setting with three files and two receivers is given by
	\begin{equation}
		\Rlb =
		\begin{cases}
			\max\limits_{i,j} H(\Xsf_i, \Xsf_j)  \,-\,2 M,        
			& \; M \in \Big[0,\eta\Big)  \\
			\frac{1}{2}  H(\Xsf_1, \Xsf_2, \Xsf_3) \,+\, \frac{1}{2}    \max\limits_{i}  H(\Xsf_i)  -M,				   
			& \;   M \in \Big[\eta, \, \zeta \Big)\\
			\frac{1}{3}\Big(H(\Xsf_1, \Xsf_2 , \Xsf_3)\,-\,M\Big)   ,				     
			& \;  M \in \Big[\zeta, \,  H(\Xsf_1,\Xsf_2,\Xsf_3)  \Big] . \label{eq: LB three appendix}
		\end{cases} 	
	\end{equation}
	%
	%
	where
	\begin{align}
		&\eta \triangleq \max\limits_{i,j}\; H(\Xsf_i,\Xsf_j)- \frac{1}{2} H(\Xsf_1,\Xsf_2,\Xsf_3) -\frac{1}{2}  \max\limits_{i} \; H(\Xsf_i), \label{eq:eta}\\
		&\zeta \triangleq \frac{1}{4}H(\Xsf_1,\Xsf_2,\Xsf_3) +  \frac{3}{4}  \max\limits_{i} \; H(\Xsf_i).   \label{eq:zeta}
	\end{align}
	Note that based on \eqref{eq:ach peak GWMR three}, for any $\Rscr\in \GWregionS$ defined in \eqref{eq:symmetric GW}, the gap to optimality satisfies 
	\begin{align}
		\RGW &-   \Rlb  \leq \Rach(M,\Rscr) -   \Rlb,
	\end{align}
	where $\Rach(M,\Rscr)$ is given in \eqref{eq:rate three}.  
	Therefore, in the following, we quantify the gap $\Rach(M,\Rscr) -   \Rlb$ for any $\Rscr\in\GWregionS$ such that  $\R_0+3\R'+3\R =H(\Xsf_1,\Xsf_2,\Xsf_3)$. From the analysis that follows it will become clear that this choice of $\Rscr$ is sufficient to achieve optimality over a certain region of the memory. Let
	\begin{align}
		&\eta_{\Rscr} \triangleq  \frac{1}{2}\R'   
		,\quad
		\zeta_{\Rscr} \triangleq \R_0 + \frac{3}{2}\R'+\frac{3}{2}\R     
		,\quad  
		\chi_{\Rscr} \triangleq  \R_0 + 3\R'+\frac{3}{2}\R . 
		\label{eq:eta R}
	\end{align}
	From Lemma~\ref{lemma:threefiles} in Appendix~\ref{app:lemma3} it follows that 
	$\eta_{\Rscr}\leq \eta  \leq   \min\{\zeta, \zeta_{\Rscr} \}    \leq \max\{\zeta, \zeta_{\Rscr} \}  \leq \chi_{\Rscr}$, 
	and we have 
	\begin{itemize}
		\item[(i)] When $M \in \Big[0, \,  \eta_{\Rscr}  \Big)$,
		\begin{align}
			\Rach(M,\Rscr) - \Rlb & \leq   \R_0  + 3\R'+2\R -  2 M -\Big( \max\limits_{i,j}\; H(\Xsf_i,\Xsf_j) -2 M\Big) \notag\\
			& \stackrel{(a)}{=}   H(\Xsf_1,\Xsf_2,\Xsf_3)  -\R  -  \max\limits_{i,j}\; H(\Xsf_i,\Xsf_j) \notag\\
			& =   \min\limits_{i}  H(\Xsf_i|\Xsf_j, \Xsf_k)  -\R  \notag\\
			&  \stackrel{(b)}{ \leq} \frac{1}{2} \min_i  H(\Xsf_j,\Xsf_k|\Xsf_i)-\R ,  \notag  
		\end{align}
		where $(a)$ follows from $\R_0  + 3\R'+2\R = H(\Xsf_1,\Xsf_2,\Xsf_3)  -\R$, and $(b)$ follows from Lemma~\ref{lemma:lem4} in Appendix~\ref{app:lemma3}.

		\item[(ii)]  When $M \in \Big[ \eta_{\Rscr}    ,\,\eta     \Big)$,
		\begin{align}
			\Rach(M,\Rscr) - \Rlb & \leq   \R_0  + \frac{5}{2}\R'+2\R -   M -\Big(  \max\limits_{i,j} H(\Xsf_i,\Xsf_j)  -2 M\Big) \notag\\
			& = H(\Xsf_1, \Xsf_2,\Xsf_3)  -\, \frac{1}{2}\R'-\R \, -     \max\limits_{i,j} H(\Xsf_i,\Xsf_j) - M \notag\\
			& =  \min\limits_{i}  H(\Xsf_i|\Xsf_j, \Xsf_k)  -\, \frac{1}{2}\R'-\R \,-  M \notag\\   
			& \stackrel{(c)}{\leq}  \min\limits_{i}  H(\Xsf_i|\Xsf_j, \Xsf_k)-  \R' -  \R \notag\\
			& \stackrel{(d)}{\leq} \frac{1}{2} \min\limits_{i} H(\Xsf_j,\Xsf_k|\Xsf_i)-\R ,   \notag 
		\end{align}
		where $(c)$ is due to  $M\geq\eta_{\Rscr}$, and $(d)$ follows from Lemma~\ref{lemma:lem4} in Appendix~\ref{app:lemma3}.

		\item[(iii)] When $M \in \Big[ \eta     ,\,   \min\{\zeta, \zeta_{\Rscr} \}     \Big)$, 
		\begin{align}
			\Rach(M,\Rscr) - \Rlb & \leq   \R_0  + \frac{5}{2}\R'+2\R -   M -\frac{1}{2}\Big( H(\Xsf_1, \Xsf_2,\Xsf_3)  + \max\limits_{i}  H(\Xsf_i)  -2 M\Big) \notag\\
			& = H(\Xsf_1, \Xsf_2,\Xsf_3)  -\, \frac{1}{2}\R'-\R \,   -\frac{1}{2}  H(\Xsf_1, \Xsf_2,\Xsf_3)  -\frac{1}{2} \max\limits_{i}  H(\Xsf_i)    \notag\\
			& = \frac{1}{2} \min\limits_{i}     H(\Xsf_j, \Xsf_k|\Xsf_i)  - \frac{1}{2}\R' -\R \leq \frac{1}{2} \min\limits_{i}     H(\Xsf_j, \Xsf_k|\Xsf_i)    -\R  . \notag 
		\end{align}

		\item[(iv)] When $M \in \Big[   \min\{\zeta, \zeta_{\Rscr} \}    ,\,   \max\{\zeta, \zeta_{\Rscr} \}     \Big)$,
		\begin{itemize}
			\item If $\min\{ \zeta, \zeta_{\Rscr} \} =   \zeta$, for $M \in \Big[ \zeta, \zeta_{\Rscr}\Big)$
			\begin{align}
				\Rach(M,\Rscr) - \Rlb & \leq   \R_0  + \frac{5}{2}\R'+2\R -   M -\frac{1}{3}\Big(H(\Xsf_1, \Xsf_2 , \Xsf_3)\,-\,M\Big)  \notag\\
				& = H(\Xsf_1, \Xsf_2,\Xsf_3)  -\, \frac{1}{2}\R'-\R \,   -\frac{1}{3}  H(\Xsf_1, \Xsf_2,\Xsf_3)  -\frac{2}{3} M \notag\\
				& \stackrel{(e)}{\leq} \frac{2}{3}    H(\Xsf_1, \Xsf_2,\Xsf_3)  -\, \frac{1}{2}\R'-\R -\frac{1}{6}  H(\Xsf_1,\Xsf_2,\Xsf_3) - \frac{1}{2}  \max\limits_{i} \; H(\Xsf_i)
				\notag \\
				&= \frac{1}{2}    \min\limits_{i}     H(\Xsf_j, \Xsf_k|\Xsf_i)  - \frac{1}{2}\R' -\R \leq \frac{1}{2} \min\limits_{i}     H(\Xsf_j, \Xsf_k|\Xsf_i)    -\R  , \notag 
			\end{align}
			where $(e)$ follows since $M\geq\zeta$.

			\item  If $\min\{ \zeta, \zeta_{\Rscr} \} =   \zeta_{\Rscr}$, for $M \in \Big[ \zeta_{\Rscr}, \zeta\Big)$
			\begin{align}
				\Rach(M,\Rscr) - \Rlb & \leq  \frac{2}{3} \R_0  + 2\R'+\frac{3}{2}\R   -  \frac{2}{3} M -\frac{1}{2}\Big( H(\Xsf_1, \Xsf_2,\Xsf_3)  + \max\limits_{i}  H(\Xsf_i) -2 M\Big) \notag\\
				& = \frac{2}{3}  H(\Xsf_1, \Xsf_2,\Xsf_3)  -\frac{1}{2}\R - \frac{1}{2}H(\Xsf_1, \Xsf_2 , \Xsf_3) -\frac{1}{2} \max\limits_{i}  H(\Xsf_i)    +  \frac{1}{3}  M \notag\\
				& \stackrel{(f)}{\leq}    \frac{1}{6}  H(\Xsf_1, \Xsf_2,\Xsf_3)-\frac{1}{2} \max\limits_{i}  H(\Xsf_i)  -\frac{1}{2}\R \notag\\
				& \qquad+  \frac{1}{3} \Big(  \frac{1}{4} H(\Xsf_1, \Xsf_2,\Xsf_3)+ \frac{3}{4}\max\limits_{i}  H(\Xsf_i)  \Big) \notag\\
				& =  \frac{1}{4}  \min\limits_{i}  H(\Xsf_j , \Xsf_k|\Xsf_i)    -\frac{1}{2}\R , \notag
			\end{align}
			where $(f)$ follows since $M<\zeta$.
			
		\end{itemize}

		\item[(v)] When $M \in \Big[ \max\{ \zeta, \zeta_{\Rscr} \}   ,\,   \chi_{\Rscr}    \Big)$,
		\begin{align}
			\Rach(M,\Rscr) - \Rlb & \leq   \frac{2}{3}\R_0 +2\R'+\frac{3}{2}\R-\frac{2}{3}M - \frac{1}{3}\Big(H(\Xsf_1, \Xsf_2 , \Xsf_3)-M\Big)   \notag\\
			& = \frac{2}{3}  H(\Xsf_1, \Xsf_2,\Xsf_3)  -\frac{1}{2}\R- \frac{1}{3}H(\Xsf_1, \Xsf_2 , \Xsf_3)  -  \frac{1}{3}  M \notag\\
			& \stackrel{(g)}{\leq} \frac{1}{3}  H(\Xsf_1, \Xsf_2,\Xsf_3)  -\frac{1}{2}\R -  \frac{1}{3}  \Big(  \frac{1}{4}H(\Xsf_1,\Xsf_2,\Xsf_3) +  \frac{3}{4}  \max\limits_{i} H(\Xsf_i)    \Big)   \notag\\
			& =  \frac{1}{4}   \min\limits_{i}  H(\Xsf_j , \Xsf_k|\Xsf_i)   -\frac{1}{2}\R  , \notag
		\end{align}
		where  $(g)$ follows form the fact that $M\geq\zeta$.

		\item[(vi)] When $M \in \Big[ \chi_{\Rscr}  ,\,  H(\Xsf_1,\Xsf_2,\Xsf_3)\Big]$,
		\begin{align}
			\Rach(M,\Rscr) - \Rlb &\leq  \frac{1}{3} \R_0  + \R'+\R -  \frac{1}{3}M -  \frac{1}{3}\Big(H(\Xsf_1, \Xsf_2 , \Xsf_3)-M\Big) = 0 .\notag
		\end{align}

	\end{itemize}
	
	Based on the analysis given above it is observed that for any symmetric $\Rscr\in\GWregionS$ that satisfies $\R_0+3\R'+3\R = H(\Xsf_1, \Xsf_2,\Xsf_3)$,  we have $\Rach(M,\Rscr) = \Rstar=\Rlb$ when $M\in[   \chi_{\Rscr}  , \,  H(\Xsf_1, \Xsf_2,\Xsf_3) ]$. In order to maximize the region of memory where the proposed GW-MR scheme is optimal, we select the operating point $\Rscr$ to be the one that minimizes $ \chi_{\Rscr} = H(\Xsf_1, \Xsf_2,\Xsf_3) - \frac{3}{2}\R$, or equivalently the one that maximizes the private description rate $\R$. Note that this choice of $\R$ also reduces the rate gap to optimality in other regions of the memory.

	\section{Lemma \ref{lemma:threefiles} and Lemma \ref{lemma:lem4}}\label{app:lemma3} 
	In the following we provide two lemmas that are used in Appendix~\ref{app:optimality three}.
	\begin{lemma}\label{lemma:threefiles}
		For any symmetric rate-tuple $\Rscr\in\GWregionS$ satisfying $\R_0  + 3\R'+3\R = H(\Xsf_1,\Xsf_2,\Xsf_3)$, and for $\eta$ and $\zeta$ defined in \eqref{eq:eta} and \eqref{eq:zeta}, and for $\eta_{\Rscr}$, $\zeta_{\Rscr}$ and $\chi_{\Rscr}$ defined in \eqref{eq:eta R} we have
		\begin{align}
			\eta_{\Rscr} \leq \eta \leq   \min\{\zeta, \zeta_{\Rscr} \}    \leq \max\{\zeta, \zeta_{\Rscr} \}  \leq \chi_{\Rscr}.
		\end{align}
	\end{lemma}
	%
	%
	%

	\begin{proof}		
		For the three-file Gray-Wyner network, any achievable symmetric rate-tuple $\Rscr\in \GWregionS$ satisfies the following cut-set bounds
		\begin{align}
			&H(\Xsf_1,\Xsf_2,\Xsf_3) \leq   \R_0+3\R'+3\R, \label{eq:three1}\\
			&H(\Xsf_i,\Xsf_j) \leq \R_0+3\R'+2\R   ,\quad i, j \in \{1,2,3\} \label{eq:three2}\\
			&H(\Xsf_i) \leq      \R_0+2\R'+\R   .\quad\qquad i \in \{1,2,3\}  \label{eq:three3}
		\end{align}
		Therefore, for any $\Rscr$ satisfying $\R_0  + 3\R'+3\R = H(\Xsf_1,\Xsf_2,\Xsf_3)$, $\eta$ is lower bounded as follows 
		\begin{align}
			\eta &= \max\limits_{i,j}\; H(\Xsf_i,\Xsf_j)- \frac{1}{2} H(\Xsf_1,\Xsf_2,\Xsf_3) -\frac{1}{2}  \max\limits_{i} \; H(\Xsf_i)\notag\\
			&\stackrel{(a)}{=}  \max\limits_{i,j}\; H(\Xsf_i,\Xsf_j)- \Big( H(\Xsf_1,\Xsf_2,\Xsf_3) - \frac{1}{2} H(\Xsf_1,\Xsf_2,\Xsf_3)\Big) -\frac{1}{2}  \max\limits_{i} \; H(\Xsf_i)  \notag\\
			&=  \frac{1}{2} \Big(H(\Xsf_1,\Xsf_2,\Xsf_3) - \max\limits_{i} \; H(\Xsf_i)  \Big) -  \min\limits_{i} H(\Xsf_i|\Xsf_j,\Xsf_k)\notag\\
			&\stackrel{(b)}{\geq}  \frac{1}{2}\Big( H(\Xsf_1,\Xsf_2,\Xsf_3) - \R_0+2\R'+\R\Big) -  \min\limits_{i} H(\Xsf_i|\Xsf_j,\Xsf_k)\notag\\
			&\stackrel{(c)}{=}  \frac{1}{2}\R'+\R -  \min\limits_{i} H(\Xsf_i|\Xsf_j,\Xsf_k), \label{eq:lemma 3 middle step}
		\end{align}
		where $(a)$ follows from replacing $\frac{1}{2} H(\Xsf_1,\Xsf_2,\Xsf_3) = H(\Xsf_1,\Xsf_2,\Xsf_3)  -\frac{1}{2} H(\Xsf_1,\Xsf_2,\Xsf_3) $, and $(b)$ follows form \eqref{eq:three3}, and in $(c)$ we have used the fact that $\R_0  + 3\R'+3\R = H(\Xsf_1,\Xsf_2,\Xsf_3)$.  Note that for such $\Rscr$, and for any $i\in\{1,2,3 \}$ we have
		\begin{align}
			H(\Xsf_i|\Xsf_j,\Xsf_k) 
			&= H(\Xsf_1,\Xsf_2,\Xsf_3) -  H(\Xsf_j,\Xsf_k)   =\R_0  + 3\R'+3\R - H(\Xsf_j,\Xsf_k) \; \stackrel{(d)}{\geq} \; \R   ,\notag
		\end{align}
		where $(d)$ is due to $H(\Xsf_j,\Xsf_k)\leq \R_0  + 3\R'+2\R$  from \eqref{eq:three2}. Therefore,  $\R -  \min\limits_{i} H(\Xsf_i|\Xsf_j,\Xsf_k) \leq 0$, and from \eqref{eq:lemma 3 middle step} it follows that
		\begin{align}
			\eta\geq \frac{1}{2} \R' = \eta_{\Rscr}.\label{eq:lem eq1}
		\end{align}

		We can upper bound $\eta$ as follows.
		\begin{align}
			\eta &= \max\limits_{i,j}\; H(\Xsf_i,\Xsf_j)- \frac{1}{2} H(\Xsf_1,\Xsf_2,\Xsf_3) -\frac{1}{2}  \max\limits_{i} \; H(\Xsf_i)\notag\\
			&\stackrel{(e)}{\leq}  2\max\limits_{i}\; H(\Xsf_i)- \frac{1}{2} H(\Xsf_1,\Xsf_2,\Xsf_3) -\frac{1}{2}  \max\limits_{i} \; H(\Xsf_i) \notag\\
			&=  \frac{3}{2}\max\limits_i H(\Xsf_i) - \frac{1}{2}(\R_0+3\R'+3\R) \notag\\
			&\stackrel{(f)}{\leq} \R_0+\frac{3}{2}\R' \leq \zeta_{\Rscr}, \label{eq:lem eq2}
		\end{align}
		where $(e)$ follows since $ \max\limits_{i,j}\; H(\Xsf_i,\Xsf_j) \leq 2\max\limits_{i}\; H(\Xsf_i)$, and $(f)$ follows form \eqref{eq:three3} and since $\zeta_{\Rscr} = \R_0  + \frac{3}{2}\R'+\frac{3}{2}\R$. 
		
		We can upper bound $\zeta$ as follows.
		\begin{align}
			\zeta &=  \frac{1}{4}H(\Xsf_1,\Xsf_2,\Xsf_3) +  \frac{3}{4} \max\limits_{i}\; H(\Xsf_i) \notag\\
			&\stackrel{(g)}{\leq}  \frac{1}{4}H(\Xsf_1,\Xsf_2,\Xsf_3) +  \frac{3}{4} \Big( \R_0+2\R'+\R   \Big)  \notag\\
			&\stackrel{(h)}{=}  \R_0  + \frac{9}{4}\R'+\frac{3}{2}\R  \leq \chi_{\Rscr},\label{eq:lem eq3}
		\end{align}
		where $(g)$ follows from \eqref{eq:three3}, and $(h)$ follows from replacing $\R_0  + 3\R'+3\R = H(\Xsf_1,\Xsf_2,\Xsf_3)$  and since $ \chi_{\Rscr} = \R_0  + 3\R'+\frac{3}{2}\R$.

		From \eqref{eq:lem eq1}, \eqref{eq:lem eq2} and \eqref{eq:lem eq3} we conclude that $\eta_{\Rscr} \leq \eta \leq   \min\{\zeta, \zeta_{\Rscr} \}    \leq \max\{\zeta, \zeta_{\Rscr} \}  \leq \chi_{\Rscr} $.
	\end{proof}

	\vspace{0.2cm}

	\begin{lemma}\label{lemma:lem4}
		For the three random variables $\Xsf_1,\Xsf_2,\Xsf_3$ with distribution $p(x_1,x_2,x_3)$, we have  
		\begin{align}
			\min_{i} H(\Xsf_i|\Xsf_j,\Xsf_k) \leq \frac{1}{2} \min_{i} H(\Xsf_j,\Xsf_k|\Xsf_i)  .\label{eq:half}
		\end{align}	
		
	\end{lemma}
	\begin{proof}
		\begin{align}
			\min_{i} H(\Xsf_i|\Xsf_j,\Xsf_k) &\stackrel{(i)}{\leq}  \frac{1}{2}  H(\Xsf_j|\Xsf_k,\Xsf_i) +\frac{1}{2}  H(\Xsf_k|\Xsf_i,\Xsf_j) \notag\\
			& \leq  \frac{1}{2}  H(\Xsf_j|\Xsf_k,\Xsf_i) +\frac{1}{2}  H(\Xsf_k|\Xsf_i) \notag\\
			&= \frac{1}{2}  H(\Xsf_j,\Xsf_k|\Xsf_i)   \notag\\
			&\leq \frac{1}{2}  \min_{i} H(\Xsf_j,\Xsf_k|\Xsf_i)  , \notag
		\end{align}	
		where $(i)$ is due to the fact that the minimum of $H(\Xsf_i|\Xsf_j,\Xsf_k)$ across all choices of $i,j$ and $k$ is no more than the arithmetic mean of two of them.
	\end{proof}

	\section{Proof of Corollary \ref{cor:special source three}}\label{app:special source three}
	For the 3-DMS considered in Corollary \ref{cor:special source three}, the symmetric rate-tuple $\Rscr$ with $\R_0 = H(\Vsf)$,  $\R_{12} = \R_{13} = \R_{23} =\R'=H_u$ and $\R_1  =\R_2 = \R_3=\R=H_x$  belongs to $\GWregionS$, for which
	\begin{align}
		&H(\Xsf_i|\Xsf_j, \Xsf_k) =  H(\Xsf_1,\Xsf_2, \Xsf_3) - H(\Xsf_j, \Xsf_k) = \R  ,\quad  i,j,k \in \{1,2,3\}\label{eq:ex three 1}\\
		&H(\Xsf_j,\Xsf_k| \Xsf_i) =  H(\Xsf_1,\Xsf_2, \Xsf_3) - H(\Xsf_j) = \R'+2\R  ,\quad  i,j,k \in \{1,2,3\}\label{eq:ex three 2}
	\end{align}
	
	For $\eta,\eta_{\Rscr} ,\zeta_{\Rscr}, \zeta,\chi_{\Rscr} $ defined in \eqref{eq:eta}, \eqref{eq:zeta} and \eqref{eq:eta R}, we have	$ \eta = \eta_{\Rscr}  \leq \zeta_{\Rscr} \leq \zeta\leq \chi_{\Rscr}$. Therefore,

	\begin{itemize}
		\item[(i)] When $M \in \Big[0, \,   \zeta_{\Rscr}  \Big)$, from Appendix~\ref{app:optimality three}, we have 
		\begin{align}
			\Rach(M,\Rscr) - \Rlb & \leq   \min\limits_{i}  H(\Xsf_i|\Xsf_j, \Xsf_k)  -\R =0 , \notag
		\end{align}		
		which follows from \eqref{eq:ex three 1}.
		
		\item[(ii)] $M \in \Big[ \zeta_{\Rscr}   ,\,  \chi_{\Rscr}    \Big]$, from Appendix~\ref{app:optimality three}, we have 
		\begin{align}
			\Rach(M,\Rscr) - \Rlb & \leq  \frac{1}{4}  \min\limits_{i}  H(\Xsf_j , \Xsf_k|\Xsf_i)    -\frac{1}{2}\R    =  \frac{1}{4} \R' , \notag
		\end{align}
		which follows from \eqref{eq:ex three 2}.
		
		\item[(iii)] $M \in \Big[ \chi_{\Rscr}  ,\,  H(\Xsf_1,\Xsf_2,\Xsf_3)\Big]$, the GW-MR scheme is optimal as proven in Appendix~\ref{app:optimality three}.
	\end{itemize}
\end{appendices}

 
\bibliographystyle{IEEEtran}
\bibliography{references.bib}

\end{document}